\documentclass{birkjour}
\usepackage{amssymb,amsthm,slashed,mathabx,bbold}

\usepackage[T1]{fontenc}
\usepackage[cp1250]{inputenc}
\usepackage{cite}
\usepackage{bm}
\usepackage{enumitem}

\usepackage[showonlyrefs]{mathtools}
\mathtoolsset{showonlyrefs=true}

\MakeRobust{\eqref}

\usepackage{hyperref}
\usepackage[dvipsnames]{xcolor}

\hypersetup{colorlinks=false,
            hypertexnames=true,
            linkcolor=BlueViolet,
            linkbordercolor=LimeGreen,
            citecolor=Mahogany,
            citebordercolor=Goldenrod}

\newcommand{\Hc}{\mathcal{H}}
\newcommand{\Dc}{\mathcal{D}}
\newcommand{\mR}{\mathbb{R}}
\newcommand{\mC}{\mathbb{C}}
\newcommand{\Bc}{\mathcal{B}}

\newcommand{\Dh}{C_0^\infty(\mathbb{R}^3, \mathbb{C}^4)}
\newcommand{\Sc}{\mathcal{S}}
\newcommand{\inc}{\mathrm{in}}
\newcommand{\out}{\mathrm{out}}
\newcommand{\ret}{\mathrm{ret}}
\newcommand{\adv}{\mathrm{adv}}
\newcommand{\rad}{\mathrm{rad}}
\newcommand{\tr}{\intercal}

\newcommand{\al}{\alpha}
\newcommand{\alb}{\boldsymbol{\alpha}}
\newcommand{\alt}{\alpha_\bot}

\newcommand{\delt}{\delta_\bot}
\newcommand{\ga}{\gamma}
\newcommand{\gax}{\hat{\gamma}}
\newcommand{\gz}{g_{(z)}}
\newcommand{\ep}{\epsilon}
\newcommand{\la}{\lambda}
\newcommand{\La}{\Lambda}
\newcommand{\lab}{\mathbf{L}}
\newcommand{\si}{\sigma}
\newcommand{\sm}{\langle \sigma \rangle}

\newcommand{\ltz}{\ell_{\tau z}}
\newcommand{\gau}{\mathcal{G}}
\newcommand{\gf}{S_\ell}
\newcommand{\glo}{S_{\log}}
\newcommand{\gco}{S_\mathrm{cone}}
\newcommand{\gtr}{S_\mathrm{tr}}

\newcommand{\w}{\omega}
\newcommand{\W}{\Omega}
\newcommand{\vph}{\varphi}
\newcommand{\vep}{\varepsilon}
\newcommand{\dsp}{\displaystyle}

\newcommand{\ov}{\overline}
\newcommand{\p}{\partial}
\newcommand{\n}{\nabla}
\newcommand{\pav}{\bm{\partial}}

\newcommand{\con}{\mathrm{const}}
\newcommand{\tm}{\langle \tau\rangle}
\newcommand{\tmo}{\langle \tau_0\rangle}
\newcommand{\rrmo}{\langle r_0\rangle}
\newcommand{\rrm}{\langle r \rangle}

\newcommand{\tmz}{\langle\tau|z\rangle}
\newcommand{\tMz}{\langle\tau\|z\rangle}

\newcommand{\zm}{\langle z\rangle}

\newcommand{\nv}{\zeta}

\newcommand{\dW}{\dot{W}}

\newcommand{\x}{\mathbf{x}}
\newcommand{\z}{\mathbf{z}}
\newcommand{\vv}{\mathbf{v}}
\newcommand{\nz}{\hat{z}}
\newcommand{\nzb}{\mathbf{\hat{z}}}

\newcommand{\gab}{\boldsymbol{\gamma}}

\newcommand{\pv}{\mathbf{p}}

\newcommand{\1}{\mathbf{1}}
\newcommand{\hc}{\dagger}
\newcommand{\cA}{\mathcal{A}}
\newcommand{\hA}{\hat{\mathcal{A}}}
\newcommand{\hF}{\hat{F}}
\newcommand{\mH}{\tilde{H}}
\newcommand{\dV}{\dot{V}}

\newcommand{\f}{F}
\newcommand{\fd}{\dot{F}}

\DeclareMathOperator{\Rp}{Re}

\DeclareMathOperator{\sgn}{sgn}
\DeclareMathOperator*{\slim}{s-lim}

\DeclareMathOperator{\essupp}{ess\,supp}
\DeclareMathOperator*{\essup}{ess\,sup}

\newtheorem{thm}{Theorem}
\newtheorem{pr}[thm]{Proposition}
\newtheorem{lem}[thm]{Lemma}
\newtheorem{asu}{Assumption}
\newtheorem{rem}[thm]{Remark}

\setlist[itemize]{leftmargin=2em}

\begin{document}

\title[Infrared problem vs gauge choice: scattering of Dirac field]%
{Infrared problem vs gauge choice:\\ scattering of~classical Dirac field}

\author{Andrzej Herdegen}

\address{Institute of Theoretical Physics\\
    Jagiellonian University\\
    ul.\,S.\,{\L}ojasiewicza 11\\
    30-348 Krak\'{o}w\\
    Poland}

\email{herdegen@th.if.uj.edu.pl}
{\it}\date{}



\date{}

\begin{abstract}

We consider the Dirac equation for the classical spinor field placed in an
external, time-dependent electromagnetic field of the form typical for
scattering settings: $F=F^\ret+F^\inc=F^\adv+F^\out$, where the current
producing $F^{\ret/\adv}$ has past and future asymptotes homogeneous of
degree $-3$, and the free fields $F^{\inc/\out}$ are radiation fields
produced by currents with similar asymp\-totic behavior. We show the
existence of the electromagnetic gauges in which the particle has `in' and
`out' asymptotic states approaching free field states, with no long-time
corrections of the free dynamics. Using a special Cauchy foliation of the
spacetime we show in this context the existence and asymptotic completeness
of the wave operators. Moreover, we define a special `evolution picture' in
which the free evolution operator has well-defined limits for
$t\to\pm\infty$, thus the scattering wave operators do not need the free
evolution counteraction.
\end{abstract}

\maketitle

\section{Introduction}\label{int}

This article is a further step in a program of investigation of the infrared
problems in electrodynamics. Among them, the long-time asymptotic behavior of
the charged matter fields is one of the key issues. Such questions as
understanding what is a charged particle, or how to define scattering
operator in quantum electrodynamics, are correlated to this problem.

It is well-known that the long-time asymptotics poses problems in theories
with long-range interactions. Standard way to deal with that is to modify the
asymptotic dynamics of charged particles or fields, by Dollard or similar
methods, augmented by some `dressing' of charged particles. The cost is the
loss of a clear interpretation of the asymptotic dynamics. Moreover, in
physically realistic quantum field theory models this procedure has not
resulted, up to now, in a non-perturbational understanding of the
issue.\footnote{For (generalized) Dollard methods see the monograph
\cite{der97}. Recent examples of the use of `dressing' in construction of
simplified quantum field models include \cite{ms14} and \cite{dyb17}. A
recent attempt at a precise implementation of the Dollard idea in general
quantum field theory (in the form proposed much earlier in rather imprecise
terms by Kulish and Faddeev) may be found in \cite{duc19}.}

In a series of articles I have put forward the idea, that the long-time
asymptotics problem may be relieved by an appropriate choice of gauge of the
electromagnetic potential. In a recent article \cite{her19} a Schr\"{o}dinger
(nonrelativistic) particle was considered in a time-dependent electromagnetic
field, of the form typical for scattering situations. It was found that an
appropriate choice of gauge allows the existence of the asymptotic dynamics,
with no need for dressing.\footnote{What remains an open question for this
system is the asymptotic completeness. In the article we expressed the view
that the clash between the symmetry groups of the two parts of the system:
Lorentz for the Maxwell and Galilean for the Schr\"{o}dinger equations, might
make the issue more problematic. Our present discussion seems to confirm
this.} We also refer the reader to this article for more extensive
description of the context and our motivation.

Here we implement the idea in the case of the classical Dirac field evolving
in the external electromagnetic field. The system is not fully
self-interacting, but the form of the external electromagnetic field mimics
the expected properties of this field in fully interacting system. We
consider the evolution of the Dirac field as a unitary evolution in a Hilbert
space. However, similarly as in the previously considered nonrelativistic
case, the evolution does not take place in standard time, over flat pure
space sheets. It turns out, that for our purpose it is convenient to consider
a Cauchy foliation supplied by constant $\tau$ surfaces, where $(\tau,\z)$
form the coordinate system\footnote{I have first proposed the use of these
coordinates in the present context in 1999, as a natural extension of the
hyperbolic foliation of the inside of the lightcone, to a Cauchy foliation of
the whole spacetime. The evolution of the Dirac field over the hyperbolic
foliation, and its large hyperbolic time asymptotics, was analysed in
\cite{her95} in terms of a Fourier-like transformation, which transformed
this evolution to a unitary evolution in the Hilbert space of spinors on the
hyperboloid of $4$-velocities. The existence of the wave operators for the
Dirac field in external electromagnetic field was established (in appropriate
gauge, see below), but no results on asymptotic completeness were obtained.
The idea now was to formulate the evolution over the foliation \eqref{intvar}
in similar Fourier terms. Piotr Marecki, my student at that time, carried out
in his MSc Thesis \cite{mar00} calculations of this programme for the free
Dirac field (scattering was only briefly mentioned in this thesis). This
method has its assets, but it is inconvenient for the analysis of many
further questions like self-adjointness of the generators, the analysis of
the domain of validity of the Dirac equation in the differential form, or
asymptotic completeness of the interacting case, the questions that have not
been considered. Here we use another method, which could be developed much
further, and enabled the solution of all these questions.} defined by:
\begin{equation}\label{intvar}
 x^0=\tau(|\z|^2+1)^\frac{1}{2}\,,\qquad \x=(\tau^2+1)^\frac{1}{2}\z\,.
\end{equation}
This poses technical complication even at the free Dirac equation level, as
the unitary evolution is not a unitary one-parameter group (Hamiltonian
depends on $\tau$).\footnote{ One should also mention that the use of more
general, than the flat equal time hypersurfaces, places our problem as a
special case of the problem of hyperbolic equations on Lorentzian manifolds,
a question intensively studied in the past, see a recent monograph
\cite{bgp07}. However, our more specific system allows us to apply more
specific Hilbert space methods, and obtain stronger results.} This is to be
contrasted with the nonrelativistic case considered earlier, where the
Niederer transformation introducing suitable coordinates leads from the free
particle Hamiltonian to that of the harmonic oscillator (a~kind of beautiful
miracle). Nevertheless, the Schr\"{o}dinger operator case becomes steeply
difficult when the potential with space part is introduced, as then the term
$\mathbf{A}\cdot\mathbf{p}$ prevents the application of the Dyson series
method. To deal with that case we have applied the Kato theorem. Here, for
the free Dirac equation we proceed differently, and our method may be of
independent interest. The addition of the interaction with the
electromagnetic field may be then treated, in contrast to the nonrelativistic
Schr\"{o}dinger  case, by a variation of the Dyson method combined with the
relativistic causality. For this system we show that an appropriate choice of
gauge removes the asymptotic problem. We show the existence and asymptotic
completeness of the wave operators, with no need for any modification of the
asymptotic free Dirac evolution.\footnote{The Cauchy problem and the
asymptotic completeness of the interacting Maxwell-Dirac system were
considered by Flato et al.\ in \cite{fst95}. However, their analysis needs
strong smoothness and smallness assumptions, the latter not under
well-determined control, and uses methods rather not well suited for an
application in quantum case (`nonlinear representation of the Poincar\'e
group'). Also, the authors modify the asymptotic dynamics by a variation of
the Dollard method. Our aims are different, as explained above.}

This result will show that the choice of a gauge $\cA(x)$ in this classical
field setting has a decisive importance for the asymptotic identification of
the incoming/outgoing charged fields. The main qualitative feature of gauges
in this appropriate class is that the product $x\cdot\cA(x)$ vanishes in
timelike directions, see remarks in Section \ref{disc} (a property already
identified in \cite{her95}). We would like to stress that the theory is
formulated from the outset in such chosen gauge. Whether such formulation may
be carried over to the quantum field theory is a subject for future
research.\footnote{We postpone to such prospective publication a comparison
with the existing discussions of the relevance of gauge choice in QED. Here
counts the idea, originated by Dirac \cite{dir55}, of an \emph{a posteriori}
transformation to a `gauge-independent gauge', developed by many authors,
most exhaustively by Steinmann, see his monograph \cite{sta00}, Chapter 12.
Also, the expectation (not shared by everyone) that formulations of QED in
differing gauges need not be equivalent received a recent support from a
mathematical analysis of a simplified model \cite{dyb19} (where literature
account may also be found).} This prospective investigation should also find
contact with the asymptotic algebra of fields in quantum electrodynamics
postulated by the author \cite{her98} (see also \cite{her17}).

Here let us only mention, that the use of a gauge in the class anticipated
above could lead to some broadening of the scope of external classical,
time-dependent electromagnetic fields for which the scattering operator for
the classical Dirac field may be lifted to the case of the respective quantum
field. As is well-known the mixing of electrons/positrons leads to rather
severe restrictions in this respect\footnote{This is a standard example of
the fact that classical external interaction problems for quantum fields
create problems of their own, which are not expected to propagate into full
closed quantum theory. An even more restrictive question in the case of the
quantum Dirac field in an external classical field is whether a unitary
evolution operator exists for this system. As shown in \cite{rui76} (for
standard evolution over the flat foliation of the spacetime), this may be
possible only if the magnetic field vanishes. This brings to sharp light a
rather restricted physical relevance of such questions and models (note that
this condition is not even inertial observer-independent; anyway, this
problem is not related to infrared questions).} (see e.g.\ the monograph
\cite{sha95}, Sections 2.4 and~2.5). Precisely what gain would be possible is
an open question.

The outline of the article is as follows. In Section \ref{free} we formulate
the Dirac evolution as a unitary evolution, with time-dependent self-adjoint
generators, over a rather general family of Cauchy surfaces. Section
\ref{evolelmg} gives the formulation of the external field problem with the
Dyson series method, but with the important use of the relativistic
causality. In both free and interacting case the self-adjointness is achieved
with the use of the commutator theorem (see Appendix \ref{open}) and the
harmonic oscillator Hamiltonian. In Section \ref{sfnpfa} we specify our
choice of coordinates to those mentioned above and transform the evolution to
a new `picture'. In this picture the free Dirac evolution on our foliation
has well-defined limits as unitary operators. Section \ref{elmintgau}
specifies the general external problem to the electromagnetic case and
discusses the gauge transformation. Section \ref{scattering} contains our
central results formulated for a wide class of electromagnetic fields: the
existence and asymptotic completeness of the wave operators. Section
\ref{typspec} gives a theorem showing that the electromagnetic fields typical
for scattering contexts described above admit potentials in gauges satisfying
the demands of the main theorems of Section \ref{scattering}. Section
\ref{disc} offers some remarks on the implications of our results and on
further physically motivated restriction in the class of the obtained
electromagnetic gauges.

Large parts of the material are shifted to Appendix. In Appendix \ref{trK} we
discuss a spinor transformation needed in Section~\ref{free}. Appendix
\ref{open} describes our method to deal with a class of time dependent
Hamiltonians. A lemma needed for the application of this method to the system
considered here is discussed in Appendix \ref{lemma}. In Appendix \ref{spvar}
we gather geometrical facts and relations in our special coordinate system.
Appendices \ref{fDe} and \ref{fwe} recapitulate some properties of the
solutions of the free Dirac and wave equations, respectively.
Appendix~\ref{decay} contains some estimates of the decay of the advanced and
retarded solutions of the inhomogeneous wave equation and their differences
(radiation fields). The results of Appendices \ref{fDe}--\ref{decay} are
applied next in Appendix \ref{elmfields} to the case of the electromagnetic
fields typical for scattering contexts. Finally, the necessary decay
properties of the special gauge introduced in Section~\ref{typspec} are
obtained in Appendix \ref{spgava}.

\section{Free Dirac evolution}\label{free}

Throughout the article we set $\hbar=1$, $c=1$. We choose a reference point,
and then the flat spacetime is identified with the Minkowski vector space.
Let $m$ be the mass parameter in the free Dirac equation. To simplify
notation we rescale Minkowski vectors by multiplying them by $m$, and denote
the resulting space by $M$, and its dimensionless vectors by $x^a$. The flat
(covariant) derivative in the rescaled Minkowski space is denoted by $\n_a$.
Also, the electromagnetic interaction to appear later will be introduced by
the interaction term $\cA_a\ov{\psi}\ga^a\psi$, so to recover the physical
units and quantities one should replace $\cA\rightarrow (e/m)\cA$, with $e$
the elementary charge.

Let $\tau:M\mapsto\mR$ be a smooth surjective function such that the
hypersurfaces $\Sigma_\tau$ of constant $\tau$ form a Cauchy foliation of the
Minkowski spacetime, with $\tau$ increasing into the future. Consider the
Dirac equation, which we write in the form
\begin{equation}\label{fDirac}
 (\tfrac{1}{2}[\ga^a,i\n_a]_+-1)\psi=0
\end{equation}
in our dimensionless coordinates, where $\gamma^a$ are the Dirac matrices and
$[.,.]_+$ symbolizes the anticommutator. As is well-known, the Cauchy problem
for this equation with the initial data $\psi|_{\Sigma_{\tau_0}}=f$ is
explicitly solved by the formula
\begin{equation}
 \psi(x)=i^{-1}\int_{\Sigma_{\tau_0}} S(x-y)\ga^a f(y)\, d\si_a(y)\,,
\end{equation}
where $d\si_a$ is the dual integration element on $\Sigma_{\tau_0}$ and
$S(x)$ is the standard Green function of the free Dirac field, in the
dimensionless coordinates
\begin{equation}
 S(x)=(i\ga\cdot\n+1)D_1(x)\,,\quad
 D_1(x)=\frac{i}{(2\pi)^3}
 \int\sgn(v^0)\delta(v^2-1)e^{-iv\cdot x}dv\,.
\end{equation}
If $f$ is a smooth bi-spinor function on $\Sigma_{\tau_0}$, with a compact
support, then $\psi(x)$ is a~smooth function in $M$, with compact support on
each Cauchy surface $\Sigma_{\tau}$. Therefore, we obtain a bijective
evolution mapping between the spaces of smooth, compactly supported bi-spinor
functions on our family of Cauchy surfaces. Moreover, if on each
$\Sigma_\tau$ one defines the scalar product
\begin{equation}\label{prtau}
 (\psi_1,\psi_2)_\tau=\int_{\Sigma_\tau} \ov{\psi_1}\ga^a\psi_2\,d\si_a\,,
\end{equation}
then the evolution is isometric.

Let now $(\nv^\mu)=(\nv^0,\nv^i)=(\tau,z^i)\equiv(\tau,\z )\in\mR^4$
($i=1,2,3$), with $\tau$ des\-cribed above, be a smooth curvilinear
coordinate system, mapping $M$ diffeomorphically onto $\mR^4$. Denote by
$\eta_{ab}$ and $C^{a\ldots}{}_{b\ldots}$ the Minkowski spacetime metric
tensor, and any other tensor, respectively. Then the geometrical components
in the coordinate system $(\xi^\mu)$ will be denoted by\pagebreak[1]
\begin{equation}\label{compxi}
 \begin{gathered}
 \p_\mu=\frac{\p}{\p\nv^\mu}=\frac{\p x^a}{\p\nv^\mu}\n_a\,,\quad
 \p_\tau=\frac{\p}{\p\tau}\,,\quad \p_i=\frac{\p}{\p z^i}\,,\\[1ex]
 \gax^\mu=(\n_a\nv^\mu)\ga^a\,,\quad
 \hat{C}^{\mu\ldots}{}_{\nu\ldots}
 =(\n_a\zeta^\mu)\ldots C^{a\ldots}{}_{b\ldots}\frac{\p x^b}{\p\zeta^\nu}\ldots \,,\\
 g_{\mu\nu}=\frac{\p x^a}{\p\nv^\mu}\frac{\p x^b}{\p\nv^\nu}\eta_{ab}\,,\quad
 g=\det(g_{\mu\nu})\,,\quad \gz=\det[(g_{ij})_{i,j\leq3}]\,,
 \end{gathered}
\end{equation}
Coordinates restricted to the indices $1,2,3$ will be written as $\gax^i$,
$g_{ij}$ and $\hat{C}^{i\ldots}{}_{j\ldots}$, and the zeroth coordinate will
be indicated by $\tau$; thus for instance:
 $\cA_a\ga^a=\hA_\mu \gax^\mu=\hA_\tau\gax^\tau+\hA_i\gax^i$.

In~these coordinates the Dirac equation~\eqref{fDirac} and the product
\eqref{prtau} take the form
\begin{gather}
 \Big(\tfrac{i}{2}\Big[|g|^{\frac{1}{2}}\gax^\mu,|g|^{-\frac{1}{2}}\p_\mu\Big]_+ -1\Big)\psi=0\,,\label{dc}\\[1ex]
 (\psi_1,\psi_2)_\tau=\int\ov{\psi_1}\gax^\tau\psi_2\, |g|^{\frac{1}{2}}d^3z\,.\label{pc}
\end{gather}

We now choose a Minkowski reference system $(e_0,\ldots,e_3)$. Let us denote
\begin{equation}\label{n}
 n=[g^{\tau\tau}]^{-\frac{1}{2}}\n\tau
\end{equation}
and let $K$ be the Lorentz rotation of Dirac spinors in the hyperplane
spanned by the pair of timelike unit vectors $(e_0,n)$, such that
\begin{equation}\label{defK}
 K^{-1}\ga^a n_a\,K=\ga^0\,,\qquad K^\hc=\ga^0K^{-1}\ga^0\,,
\end{equation}
where dagger $\hc$ denotes the hermitian conjugation of matrices. The form
and further properties of $K$ are discussed in Appendix \ref{trK}. We define
the following transformation:
\begin{equation}\label{psich}
 \psi=T_\tau\chi\,,\qquad T_\tau=(|g|g^{\tau\tau})^{-\frac{1}{4}}K=|\gz|^{-\frac{1}{4}}K\,.
\end{equation}
The product \eqref{pc} then takes the standard $\mC^4\otimes L^2(\mR^3)$ form
\begin{equation}
 (\psi_1,\psi_2)_\tau=\int \chi_1^\hc\chi_2 d^3z\,.
\end{equation}
Setting $\psi$ into \eqref{dc} and applying from the left the transformation
$(g^{\tau\tau})^{-\frac{1}{2}}T_\tau^{-1}$ we obtain an equivalent form of
the Dirac equation
\begin{multline}\label{dcT}
 \Big(\tfrac{i}{2}\Big[(g^{\tau\tau})^{-\frac{1}{2}}\gax_K^\mu,\p_\mu\Big]_+
 -(g^{\tau\tau})^{-\frac{1}{2}}\Big)\chi\\
 +\tfrac{i}{2}(g^{\tau\tau})^{-\frac{1}{2}}
 \Big[\gax_K^\mu,(K^{-1}\p_\mu K)\Big]_+\chi=0\,,
\end{multline}
where we have introduced notation
\begin{equation}\label{gaK}
\ga_K=K^{-1}\ga K\,.
\end{equation}
We now make an
additional assumption: $\Sigma_\tau$ is rotationally symmetric in the chosen
Minkowski system, that is $\tau=\tau(x^0,|\x|)$. We show in Appendix
\ref{trK} that in this case the anticommutator in the second line of
\eqref{dcT} vanishes. We also note that
$(g^{\tau\tau})^{-\frac{1}{2}}\gax_K^\tau=\ga^0\equiv\beta$. This allows us
to write the Dirac equation in the form
\begin{equation}\label{dcTs}
 i\p_\tau\chi
 =\Big(-\tfrac{i}{2}\big[(g^{\tau\tau})^{-\frac{1}{2}}\beta\gax_K^i,\p_i\big]_+
 +(g^{\tau\tau})^{-\frac{1}{2}}\beta\Big)\chi\,.
\end{equation}

We summarize the above discussion.
\begin{thm}\label{freeDirac}
(i) Let the smooth function on the Minkowski space $\tau=\tau(x^0,|\x|)$
determine its foliation by Cauchy surfaces $\Sigma_\tau$, with $\tau$
increasing into the future. Denote by $C_0^\infty(\Sigma_\tau, \mC^4)$ the
space of smooth, compactly supported functions on $\Sigma_\tau$ with values
in $\mC^4$.

Then the free Dirac evolution of initial data in $C_0^\infty(\Sigma_\sigma,
\mC^4)$ is determined by a family of bijective linear evolution operators
\begin{equation}\label{fsec}
 U^\Sigma_0(\tau,\sigma): C_0^\infty(\Sigma_\sigma, \mC^4)\mapsto
 C_0^\infty(\Sigma_\tau, \mC^4)
\end{equation}
isometric with respect to the products \eqref{prtau}. The size of the support
in $\Sigma_\tau$ is restricted by the size of the support in $\Sigma_\si$ and
the relativistic causality. By continuity, $U_0^\Sigma(\tau,\sigma)$ extend
to unitary operators
\begin{equation}\label{fseh}
 U^\Sigma_0(\tau,\sigma): \Hc(\Sigma_\sigma)\mapsto\Hc(\Sigma_\tau)\,,
\end{equation}
where $\Hc(\Sigma_\tau)$ is the Hilbert space of spinor functions on
$\Sigma_\tau$ with the pro\-duct~\eqref{prtau}.

 (ii) Let $(\tau,z^i):M\mapsto\mR^4$ be a smooth diffeomorphism, with
notation  \eqref{compxi}, such that $\tau$ satisfies the assumptions of (i).
Denote
\begin{equation}
 \Hc=\mC^4\otimes L^2(\mR^3,d^3z)\,,
\end{equation}
\begin{equation}\label{H0}
 H_0(\tau)=\tfrac{1}{2}[\la^i(\tau),p_i]_+ +\mu(\tau)\,,\quad
 h_0=\tfrac{1}{2}(\pv^2+\z^2)\,,
\end{equation}
\begin{equation}\label{mula}
p_i=-i\p/\p z^i\,,\qquad
 \mu(\tau)=(g^{\tau\tau})^{-\frac{1}{2}}\beta\,,\qquad
 \la^i(\tau)=\mu(\tau)\gax_K^i\,.
\end{equation}
Then the transformation
\begin{equation}
 T_\tau:\Hc\mapsto \Hc(\Sigma_\tau)\,,\qquad T_\tau=|\gz|^{-1/4}K\,,
\end{equation}
with operator $K$ discussed in Appendix \ref{trK}, is a unitary operator and
the family
\begin{equation}\label{Uzero}
 U_0(\tau,\sigma)=T_\tau^{-1}U_0^\Sigma(\tau,\sigma)T_\sigma\,, \qquad
 U_0(\tau,\si):\Hc\mapsto\Hc
\end{equation}
formes a unitary, strongly continuous evolution system, such that the
following holds:
\begin{itemize}
\item[(A)] $U_0(\tau,\si)\Dh=\Dh$ and the relativistic causality is
    respected.
\item[(B)] For $\vph\in\Dh$ the maps $(\tau,\si)\mapsto
    H_0(\tau)U_0(\tau,\si)\vph$ and $(\tau,\si)\mapsto
    h_0U_0(\tau,\si)\vph$ are strongly continuous, the map
    $(\tau,\si)\mapsto U_0(\tau,\si)\vph$ is in the class $C^1$ in the
    strong sense, and the following equations are satisfied
\begin{equation}\label{freeDH}
\begin{split}
 i\p_\tau U_0(\tau,\si)\vph&=H_0(\tau)U_0(\tau,\si)\vph\,,\\
 i\p_\si U_0(\tau,\si)\vph&=-U_0(\tau,\si)H_0(\si)\vph\,.
 \end{split}
\end{equation}

\item[(C)] The operators $H_0(\tau)$ are symmetric on $\Dh$.
\end{itemize}
\end{thm}
\begin{proof}
The symmetry of $H_0(\tau)$ follows from the symmetry of the operators
$\la^i(\tau)$, $\mu(\tau)$ and~$p_i$, and invariance of $\Dh$ with respect to
them. The other statements of the theorem follow easily from the discussion
preceding the theorem.
\end{proof}

We now add assumptions which will be satisfied in the context of our
application, and which alow us to significantly extend the domains of
validity of the results in the last theorem. We use the abbreviation
\begin{equation}\label{zm}
 \zm=(|\z|^2+1)^\frac{1}{2}\,.
\end{equation}
Also, the usual multi-index notation will be used.
\begin{thm}\label{freeDiracHil}
Let all the assumptions of Theorem \ref{freeDirac} be satisfied. Suppose in
addition that the following bounds of matrix norms hold:
\begin{gather}
 |\la^i(\tau,\z)|\leq C(\tau)\zm\,,\quad
 |\p_z^\al\la^i(\tau,\z)|\leq C(\tau)\,,\quad 1\leq|\al|\leq3\,,\label{labound}\\
 |\p_z^\beta\mu(\tau,\z)|\leq C(\tau)\zm^{2-|\beta|}\,,\quad |\beta|\leq2\,,
\end{gather}
where $C(\tau)$ is a continuous function. Denote
$h_0=\tfrac{1}{2}(\pv^2+\z^2$), the self-adjoint harmonic oscillator
Hamiltonian.

Then the following holds:
\begin{itemize}
\item[(A)] $H_0(\tau)$ are essentially self-adjoint on $\Dh$ and on each
    core of $h_0$, and  $\Dc(h_0)\subseteq \Dc(H_0(\tau))$. Moreover,
    $h_0^{-1}H_0(\tau)$ and $H(\tau)h_0^{-1}$ extend to bounded operators,
    and as functions of $\tau$ are strongly continuous.
\item[(B)] The operator $h_0U_0(\tau,\si)h_0^{-1}$ is bounded, with the
    norm
    \begin{equation}
    \|h_0U_0(\tau,\si)h_0^{-1}\|
    \leq \exp\Big[\con\int_\si^\tau C(\rho)d\rho\Big]\,,
    \end{equation}
 so in particular
\begin{equation}
    U_0(\tau,\si)\Dc(h_0)=\Dc(h_0)\,,
\end{equation}
and the map $(\tau,\si)\mapsto h_0U(\tau,\si)h_0^{-1}$ is strongly
continuous.
\item[(C)] For $\vph\in\Dc(h_0)$ the vector functions
    $h_0U_0(\tau,\si)\vph$, $H_0(\tau)U_0(\tau,\si)\vph$  and
    $U_0(\tau,\si)H_0(\si)\vph $ are strongly continuous with respect to
    parameters, the vector function $U_0(\tau,\si)\vph$ is of the class
    $C^1$ with respect to parameters in the strong sense and the following
    equations hold
\begin{align}
 i\p_\tau U_0(\tau,\si)\vph&=H_0(\tau) U_0(\tau,\si)\vph\,,\\
 i\p_\si U_0(\tau,\si)\vph&=-U_0(\tau,\si)H_0(\si)\vph\,.
\end{align}
\item[(D)] Relativistic causality is respected by $U_0(\tau,\si)$.
    Therefore, if we denote
    \begin{gather}
    \Hc\supset\Hc_c\ \text{--the subspace of functions with compact essential support}\,,\\
    \Dc_c(\pv^2)=\Hc_c\cap\Dc(\pv^2)\,,
    \end{gather}
    then
\begin{equation}
 U_0(\tau,\si)\Hc_c=\Hc_c\,,\quad
 U_0(\tau,\si)\Dc_c(\pv^2)=\Dc_c(\pv^2)\,.
\end{equation}
\end{itemize}
\end{thm}
\begin{proof}
It follows from the assumptions on $\la^i$ and $\mu$ that $H_0(\tau)$ fulfill
all the conditions imposed on the operator $h$ in Lemma \ref{lembound} in
Appendix \ref{lemma}. In~Theorems \ref{op-ham} and \ref{op-ev} in
Appendix~\ref{open} we now put
\begin{equation}
 \Dc=\Dh\,,\quad h(\tau)=H_0(\tau)\,,\quad u(\tau,\si)=U_0(\tau,\si)
 \end{equation}
and $h_0$ as defined in the present assumptions. Then by the results of
Theorem~\ref{freeDirac} all the assumptions of these theorems are satisfied
and the present thesis follows.
\end{proof}

\section{Evolution in external field}\label{evolelmg}

We now consider the Dirac equation in an external field, of the form
\begin{equation}
 (\tfrac{1}{2}[\ga^a,i\n_a]_+-1-V)\psi=0\,,
\end{equation}
where $V(x)$ is a matrix function satisfying the condition
\begin{equation}\label{Vhc}
 V(x)^\hc=\ga^0V(x)\ga^0\,,
\end{equation}
which guarantees the reality of the interaction term $-\ov{\psi}V\psi$ in the
lagrangian and the conservation of the current $\ov{\psi}\gamma^a\psi$. The
addition of this interaction term leads to the modification of the free
curvilinear version to the equation of the form
\begin{equation}\label{intW}
 i\p_\tau\vph=(\tfrac{1}{2}[\la^i,p_i]_++\mu +W)\vph\,,
\end{equation}
where
\begin{equation}
 W(\tau,\z)=\mu K^{-1}V(x)K=W(\tau,\z)^\hc\,,
\end{equation}
hermicity being equivalent to the property \eqref{Vhc} (see \eqref{defK}).

Using the interaction picture technique we shall obtain the evolution
operators for which equation \eqref{intW} is satisfied. However, thanks to
the relativistic causality we can extend the applicability of the technique
to the following setting, in which the operators $W(\tau)$ do not need to be
bounded.

\begin{thm}\label{interDirac}
Let all the assumptions of Theorems \ref{freeDirac} and \ref{freeDiracHil} be
satisfied. Suppose in addition that $W(\tau,\z)$ is a hermitian matrix
function such that the mappings
\begin{equation}\label{propW1}
 \big\{\mR\ni\tau\mapsto \p_z^\al W(\tau,\z)\big\}
 \in C^0(\mR, L_\mathrm{loc}^\infty(\mR^3))\,,\qquad |\al|\leq2\,,
\end{equation}
and
\begin{equation}\label{propW2}
 \|\zm^{-2+|\al|}\p_z^\al W(\tau,\z)\|_\infty\leq C(\tau)\,,\qquad |\al|\leq1\,,
\end{equation}
where $C(\tau)$ is a continuous function.  Denote
\begin{equation}
 H(\tau)=H_0(\tau)+W(\tau)\,,
\end{equation}
with the initial domain $\Dh$, and define the formal series\pagebreak[2]
\begin{equation}\label{U}
\begin{aligned}
 U(\tau,\si)&=\sum_{n=0}^\infty U^{(n)}(\tau,\si)\,,\qquad
 U^{(0)}(\tau,\si)=\1\,,\\
  U^{(n)}(\tau,\si)
  &=(-i)^n \int\limits_{\tau\geq\tau_{n}\geq\ldots\geq\tau_1\geq\si}
  U_0(\tau,\tau_n)W(\tau_n)U_0(\tau_n,\tau_{n-1})\ldots\\
 &\hspace{10em}\ldots W(\tau_1)U_0(\tau_1,\si)d\tau_n\ldots d\tau_1\\
 &=-i\int\limits_\si^\tau U_0(\tau,\rho)W(\rho)U^{(n-1)}(\rho,\si)d\rho\\
 &=-i\int\limits_\si^\tau U^{(n-1)}(\tau,\rho)W(\rho)U_0(\rho,\si)d\rho\,,\qquad
 n\geq 1\,.
\end{aligned}
\end{equation}

Then the following is true:
\begin{itemize}
\item[(A)] $H(\tau)$ are essentially self-adjoint on $\Dh$ and on each core
    of $h_0$, and  $\Dc(h_0)\subseteq \Dc(H(\tau))$.
\item[(B)] Series $U(\tau,\si)$ and its conjugation converge strongly on
    $\Hc_c$, and the limit operator extends to a unitary propagator on
    $\Hc$, strongly continuous in its parameters.
\item[(C)] $U(\tau,\si)$ respects relativistic causality and
    \begin{equation}
    U(\tau,\si)\Dc_c(\pv^2)=\Dc_c(\pv^2)\,.
    \end{equation}
\item[(D)] For $\vph\in\Dc_c(\pv^2)$ the vector functions
    $h_0U(\tau,\si)\vph$, $H(\tau)U(\tau,\si)\vph$ and
    $U(\tau,\si)H(\si)\vph$ are strongly continuous with respect to
    parameters, the function $U(\tau,\si)\vph$ is of the class $C^1$ with
    respect to parameters in the strong sense and the following equations
    hold
\begin{equation}\label{extDirac}
\begin{aligned}
 i\p_\tau U(\tau,\si)\vph&=H(\tau) U(\tau,\si)\vph\,,\\
 i\p_\si U(\tau,\si)\vph&=-U(\tau,\si)H(\si)\vph\,.
\end{aligned}
\end{equation}
\item[(E)] Unitary operators satisfying (D), with $U(\si,\si)=\1$,  are
    unique.
\end{itemize}
\end{thm}
\begin{proof}
The proof of (A) is similar as the proof in the free case, statement (A) in
Theorem \ref{freeDiracHil}. One has to note that the assumption
\eqref{propW2} ensures the validity of items (i) and (ii) in Lemma
\ref{lembound}; the result of item (iii) is not needed.

Consider the series \eqref{U}.  As each of the operators $U_0(\tau,\tau_n)$,
$U_0(\tau_k,\tau_{k-1})$ and $U_0(\tau_1,\si)$ in the definition of
$U^{(n)}(\tau,\si)$ respects causality, and multiplication by $W(\tau_k,\z)$
does not enlarge the support of the function, each of the operators
$U^{(n)}(\tau,\si)$ respects causality. Let, for chosen $T$ and $r$,
\begin{equation}
 \tau,\si\in[-T,T]\,,\quad \vph\in\Hc\,,\quad
 \essupp\vph\subseteq\{\z\mid |\z|\leq r\}\,,
 \end{equation}
which is assumed for the rest of this proof. Then there exists $R(T,r)$ such
that the essential support of all functions $U^{(n)}(\tau,\si)\vph$ is
contained in $|\z|\leq R(T,r)$ (uniformly with respect to $\tau$, $\si$ and
$\vph$ in the assumed range). Let $J(\z)$ be a smooth function such that
$J(\z)=1$ for $|\z|\leq R$ and $J(\z)=0$ for $|\z|\geq R+1$. Denote by
$U_J(\tau,\si)$ and $U_J^{(n)}(\tau,\si)$ the series \eqref{U} and its terms
in which $W(\rho,\z)$ have been replaced by $W_J(\rho,\z)=J(\z)W(\rho,\z)$;
note that by assumption \eqref{propW1} $W_J(\rho)$ is a bounded, strongly
continuous operator, with $\|W_J(\rho)\|\leq d_1$ for some constant $d_1$,
uniformly on $\rho\in[-T,T]$. Using the first of the recursive relations in
\eqref{U} it is now easy to see that in the given setting
\begin{equation}
 U^{(n)}(\tau,\si)\vph=U_J^{(n)}(\tau,\si)\vph\,.
\end{equation}
Arguing as in the proof of the Dyson expansion (see Thm.\, X.69 in
\cite{rs79II}) one shows that $U_J(\tau,\si)$ converges uniformly to a
strongly continuous unitary propagator. As the space $\Hc_c$ is dense in
$\Hc$, the statement (B) is proved. Also, the causality is respected.

To prove (C) (the invariance of the subspace) and (D) we write
\begin{equation}
 h_0W_J(\rho) h_0^{-1}=[h_0,W_J(\rho)]h_0^{-1}+W_J(\rho)
\end{equation}
and note that by Theorem \ref{op-ham} applied to $h(\tau)=W_J(\tau)$ (with
Lemma \ref{lembound}) the first term on the rhs is a bounded operator,
strongly continuous by the assumption \eqref{propW1}. Therefore, the operator
$h_0W_J(\rho)h_0^{-1}$ is also bounded and strongly continuous, and
$\|h_0W_J(\rho)h_0^{-1}\|\leq d_2$ for some constant $d_2$, uniformly on
$\rho\in[-T,T]$. Taking also into account statement (B) in Theorem
\ref{freeDiracHil} we conclude that $h_0U_J^{(n)}(\tau,\si)h_0^{-1}$ is
bounded, strongly continuous, with the norm estimated by
\begin{equation}
 \|h_0 U_J^{(n)}(\tau,\si)h_0^{-1}\|
 \leq (n!)^{-1}(d_1+d_2)^n
 \exp\Big[\con\int_\si^\tau C(\rho)d\rho\Big]\,,
\end{equation}
This implies that the series $\sum_nh_0U_J^{(n)}(\tau,\si)h_0^{-1}$ converges
uniformly to the strongly continuous function $h_0U_J(\tau,\si)h_0^{-1}$. Let
us now impose on $\vph$ a stron\-ger assumption that $\vph\in \Dc_c(\pv^2)$.
Then
\begin{equation}
 h_0U(\tau,\si)\vph=h_0U_J(\tau,\si)\vph=h_0U_J(\tau,\si)h_0^{-1}h_0\vph\,.
\end{equation}
This proves (C), and also the strong continuity of $h_0U(\tau,\si)\vph$ in
(D). The strong continuity of $H(\tau)U(\tau,\si)\vph$ and
$U(\tau,\si)H(\si)\vph$ follows now from the strong continuity of
$H(\tau)h_0^{-1}$, similarly as in Theorem \ref{freeDiracHil}, basing on
Theorem \ref{op-ham} and assumptions \eqref{propW1} and \eqref{propW2}.

Next, we note that
\begin{equation}
 U_0(0,\tau)U^{(n)}(\tau,\si)\vph
 =-i\int_\si^\tau U_0(0,\rho)W(\rho)U^{(n-1)}(\rho,\si)\vph\, d\rho\,,
\end{equation}
so this vector function is continuously differentiable in $\tau$ in the
strong sense and
\begin{equation}
 i\p_\tau[U_0(0,\tau)U^{(n)}(\tau,\si)]\vph
 = U_0(0,\tau)W(\tau)U^{(n-1)}(\tau,\si)\vph
\end{equation}
(remember that on the rhs $W$ may be replaced by $W_J$). Again, using the
uniform convergence of $U_J(\tau,\si)$ we obtain
\begin{equation}
 i\p_\tau[U_0(0,\tau)U(\tau,\si)]\vph=U_0(0,\tau)W(\tau)U(\tau,\si)\vph\,,
\end{equation}
and the rhs is strongly continuous in $(\tau,\si)$. Finally, differentiating
\begin{equation}
 U(\tau,\si)\vph=U_0(\tau,0)[U_0(0,\tau)U(\tau,\si)]\vph
\end{equation}
by the Leibnitz rule and using (C) from Theorem \ref{freeDiracHil} we arrive
at the first equation in \eqref{extDirac}. The uniqueness (E) follows easily
from this equation. The proof of the second equation in \eqref{extDirac} is
similar, with the use of the second of the recursive relations in \eqref{U};
we omit the details.
\end{proof}

The most general matrix field $V(x)$ satisfying condition \eqref{Vhc} may be
concisely represented in the form
\begin{equation}
 V=\sum_{k=0}^4 i^{\frac{1}{2}k(k-1)}C^k_{a_1\ldots a_k}\ga^{a_1}\ldots\ga^{a_k}\,,
\end{equation}
where the tensor fields $C^k_{a_1\ldots a_k}$ are real and antisymmetric. The
corresponding form of the hermitian matrix function $W$ is
\begin{equation}
 W=\mu\sum_{k=0}^4 i^{\frac{1}{2}k(k-1)}
 \hat{C}^k_{\mu_1\ldots \mu_k}\gax^{\mu_1}_K\ldots\gax^{\mu_k}_K\,.
\end{equation}
This form encompasses the scalar ($k=0$) and pseudoscalar ($k=4$) potentials,
the electromagnetic vector potential ($k=1$) and the pseudovector potential
($k=3$), the interactions characteristic for anomalous magnetic and electric
moments, as well as the linearized gravitation ($k=2$).

\section{Special foliation, new picture and free asymptotics}\label{sfnpfa}

We now choose the foliation $\tau$ and the variables $\z$ by
\begin{equation}\label{tauzed}
 x^0=\tau\zm\equiv\tau(|\z|^2+1)^{\frac{1}{2}}\,,\qquad
 \x=\tm\z\equiv(\tau^2+1)^{\frac{1}{2}}\z\,,
\end{equation}
and we also denote
\begin{equation}
 \tmz=(\tau^2+|\z|^2+1)^{\frac{1}{2}}\,.
\end{equation}
The idea behind this choice is that for $\tau$ tending to $\pm\infty$ the
Cauchy surfaces should tend, for timelike directions in the spacetime, to the
hyperboloids $x^2=\tau^2$. All geometrical facts on these curvilinear
coordinates needed for our purposes are gathered in Appendix~\ref{spvar}.

From now on we adopt the coordinate system \eqref{tauzed}. Using the
properties of this system it is then easy to see that all the conditions of
Theorems \ref{freeDirac} and \ref{freeDiracHil} are satisfied. Also,
a~convenient property of these coordinates is that $g_{\tau i}=0$. Therefore,
we find
\begin{equation}
 [\gax_K^i,\beta]_+=(g^{\tau\tau})^{-\frac{1}{2}}K^{-1}[\gax^i,\gax^\tau]_+K
 =2(g^{\tau\tau})^{-\frac{1}{2}}g^{i\tau}=0\,,
\end{equation}
and then
\begin{gather}
 [\la^i,\beta]_+=0=[\la^i,\mu]_+\,,\label{comlabe}\\[1ex]
 \tfrac{1}{2}[\la^i,\la^j]_+=-g_{\tau\tau}g^{ij}\1\equiv \rho^{ij}\1\,;\label{defro}
\end{gather}
the symmetric form $\rho$ is positive definite.

The free evolution $U_0(\tau,\si)$ may be easily expressed in the coordinate
system \eqref{tauzed}. For this purpose it is sufficient to consider
$U_0(\tau,0)$. Using the representation of the Dirac solution \eqref{Dirini}
in Appendix \ref{fDe} and the definition~\eqref{Uzero}, for
$\vph\in\Sc(\mR^3,\mC^4)$ we obtain
\begin{multline}\label{asUfi}
 [U_0(\tau,0)\vph](\z)
 =\Big(\frac{\tmz}{\tm\zm}\Big)^{\frac{1}{2}}K(\tau,\z)^{-1}\times\\
 \times\Big(\frac{\tm}{2\pi}\Big)^{\frac{3}{2}}\int\big[e^{-ix(\tau,\z)\cdot v}P_+(v)
 -e^{ix(\tau,\z)\cdot v}P_-(v)\big][\mathcal{F}^{-1}\vph](v)d\mu(v)\,,
\end{multline}
where $x(\tau,\z)=(\tau\zm,\tm\z)$. We consider the asymptotic form of this
evolution for $\tau\to\pm\infty$. Let $\vph\in \mathcal{F}\Dh$ and restrict
$\z$ to a compact set. Then both $\vv$ (the space part of $v$) and $\z$ are
restricted to bounded sets and by the stationary phase method (see e.g.\
\cite{vai89}) the leading asymptotic behavior of the second line in
\eqref{asUfi} is given by
\begin{equation}\label{asint}
 \mp i\big[e^{-i(\tau\pm\frac{\pi}{4})}P_+(z^0,\pm\z)
 +e^{i(\tau\pm\frac{\pi}{4})}P_-(z^0,\pm\z)\big][\mathcal{F}^{-1}\vph](z^0,\pm\z)
 +O(\tm^{-1})\,,
\end{equation}
where $z^0=\zm$, the upper/lower signs $\mp$ and $\pm$ correspond to the
limits $\tau\to\pm\infty$, respectively, and the rest is bounded uniformly
for $\z$ in the given set. The limit of the term on the rhs of \eqref{asUfi}
in the first line is equal to $\zm^{-\frac{1}{2}}K(\pm\infty,\z)^{-1}$, with
\begin{equation}\label{asK}
 K(\pm\infty,\z)=2^{-\frac{1}{2}}\big[(1+\zm)^{\frac{1}{2}}
 \pm(1+\zm)^{-\frac{1}{2}}\ga^0\z\cdot\gab\big]=K(\infty,\pm\z)\,,
\end{equation}
and again the rest is bounded by $\con\tm^{-1}$, uniformly in the given set.
We now note two facts:
\begin{gather}
 K(\infty,\vv)^{-1} v\cdot\ga\, K(\infty,\vv)=\beta\,,\\
 e^{-i\si}\tfrac{1}{2}(\1+\beta)
 +e^{i\si}\tfrac{1}{2}(\1-\beta)
 =e^{-i\si\beta}\,.
\end{gather}
The first identity is the limit form of the first of relations in
\eqref{defK}, but it may also be checked directly. The second identity is
most easily evaluated on the two complementary eigenspaces of $\beta$ (with
eigenvalues $\pm1$).

Setting the asymptotic forms \eqref{asint} and \eqref{asK} into the formula
\eqref{asUfi}, and using the above identities (the second one with
$\si=\tau\pm\frac{\pi}{4}$), we obtain
\begin{multline}\label{asU0fi}
 [U_0(\tau,0)\vph](\z)
 = \mp ie^{-i(\tau\pm\frac{\pi}{4})\beta}\zm^{-\frac{1}{2}}K(\infty,\pm\z)^{-1}
 [\mathcal{F}^{-1}\vph](z^0,\pm\z)\\ +O(\tm^{-1})\,.
\end{multline}

The last formula suggests the definition of the following unitary
transformation:
\begin{equation}
 \Phi(\tau)=\exp[-i\tau\beta]\,,\label{fi1}
\end{equation}
and the associated change of the evolution `picture' (also in the interacting
case):
\begin{equation}\label{HUPhi}
\begin{gathered}
 U_{0\Phi}(\tau,\si)=\Phi(\tau)^*U_0(\tau,\si)\Phi(\si)\,,\quad
 U_\Phi(\tau,\si)=\Phi(\tau)^*U(\tau,\si)\Phi(\si)\,,\\[1ex]
  H_{0\Phi}(\tau)=\Phi(\tau)^*H_0(\tau)\Phi(\tau)-\beta\,,\\[1ex]
  H_{\Phi}(\tau)=H_{0\phi}(\tau)+W_\Phi(\tau)\,,\quad
  W_\Phi(\tau)=\Phi(\tau)^* W(\tau)\Phi(\tau)\,.
\end{gathered}
\end{equation}

\begin{rem}\label{interDiracPhi}
Under the conditions of Theorem~\ref{interDirac} all statements of its thesis
remain valid with the replacements defined by the equations \eqref{HUPhi}.
\end{rem}
\noindent This follows quite trivially, as $\Phi(\tau)$ acts only on the
factor $\mC^4$ in the Hilbert space $\Hc=\mC^4\otimes L^2(\mR^3)$.

After that, we go back to the asymptotics of the free evolution, where we
shall need the parity operator
\begin{equation}
 [\mathcal{P}\vph](\z)=\vph(-\z)\,.
\end{equation}
Moreover, we observe that the map
\begin{equation}
 [\mathcal{K}\vph](\z)=\zm^{\frac{1}{2}}K(\infty,\z)\vph(\z)
\end{equation}
is a unitary operator $\Hc\mapsto L^2_\ga(H)$ (with the latter space defined
in Appendix~\ref{fDe}).

\begin{thm}\label{UF0}
The following strong limits exist as unitary operators in $\Hc$:
\begin{gather}
 U_{0\Phi}(+\infty,0)=\slim_{\tau\to+\infty}U_{0\Phi}(\tau,0)
 =-ie^{-i\frac{\pi}{4}\beta}\mathcal{K}^{-1}\mathcal{F}^{-1}\,,\\
 U_{0\Phi}(-\infty,0)=\slim_{\tau\to-\infty}U_{0\Phi}(\tau,0)
 =ie^{+i\frac{\pi}{4}\beta}\mathcal{P} \mathcal{K}^{-1}\mathcal{F}^{-1}\,,
\end{gather}
\begin{gather}
 U_{0\Phi}(0,+\infty)=\slim_{\tau\to+\infty}U_{0\Phi}(0,\tau)
 =i\mathcal{F}\mathcal{K}e^{+i\frac{\pi}{4}\beta}\,,\\
 U_{0\Phi}(0,-\infty)=\slim_{\tau\to-\infty}U_{0\Phi}(0,\tau)
 =-i\mathcal{F}\mathcal{K}\mathcal{P}e^{-i\frac{\pi}{4}\beta}\,,\\
 U_{0\Phi}(0,\pm\infty)=U_{0\Phi}(\pm\infty,0)^*\,.
\end{gather}
Therefore,\footnote{Note an analogy with the nonrelativistic case, see
Section 2 in \cite{her19}.}
\begin{equation}
 U_{0\Phi}(+\infty,-\infty)=i\beta\mathcal{P}\,.
\end{equation}
\end{thm}
\begin{proof}
For $\psi\in \Dh$ and $\vph\in \mathcal{F}\Dh$ the formula \eqref{asU0fi}
gives
\begin{gather}
  \lim_{\tau\to+\infty}(\psi,U_{0\Phi}(\tau,0)\vph)
  =(\psi,(-i)e^{-i\frac{\pi}{4}\beta}\mathcal{K}^{-1}\mathcal{F}^{-1}\vph)\,,\\
 \lim_{\tau\to-\infty}(\psi,U_{0\Phi}(\tau,0)\vph)
 =(\psi,ie^{+i\frac{\pi}{4}\beta}\mathcal{P} \mathcal{K}^{-1}\mathcal{F}^{-1}\vph)\,.
\end{gather}
Both subspaces are dense in $\Hc$, so the weak operator limits result. But
the limit operators are evidently unitary, so the weak limits imply the
strong limits of $U_{0\Phi}(\tau,0)$, as well as its conjugate.
\end{proof}

\section{Electromagnetic interaction and gauge
transformation}\label{elmintgau}

In the rest of this article we are interested in the standard, minimal
coupling electromagnetic interaction. For the electromagnetic field $F_{ab}$,
we reserve notation $A_a$ for the Lorenz gauge potential (fully specified in
what follows). We write $\cA_a$ for the potential in a general gauge to be
used in the Dirac equation. Also, we recall our conventions defined in
formulas \eqref{compxi} and the following remarks, so $\hF_{\mu\nu}$,
$\hat{A}_\mu$ and $\hA_\mu$ are components of these fields in our coordinate
system, and for $\mu,\nu=i,j$ etc. the range is restricted to values $1,2,3$.

Therefore, in the electromagnetic case the field $V$ and its transformed
version $W$ are
\begin{equation}\label{WA}
 V(x)=\cA_a(x)\ga^a\,,\qquad W(\tau,\z)=\hA_\tau(\tau,\z)+\hA_i(\tau,\z)\la^i(\tau,\z)\,,
\end{equation}
and we write the Hamiltonian as
\begin{equation}
 H=\tfrac{1}{2}[\la^i,\pi_i]_++\mu+\hA_\tau\,,\qquad \pi_i=p_i+\hA_i\,.
\end{equation}
\begin{thm}\label{gengau}
 The coordinate system $(\tau,\z)$ given by \eqref{tauzed} is assumed.

 (i) Let the electromagnetic potential $\cA(x)$ have components $\hA_\mu(\tau,\z)$ such that
for all indices $\mu=\tau,1,2,3$  the mappings
 \begin{equation}\label{pot}
 \big\{\mR\ni\tau\mapsto \p_z^\al\hA_\mu(\tau,\z)\big\}
 \in C^0(\mR,L^\infty_\mathrm{loc}(\mR^3))\,,\quad |\al|\leq2\,,
 \end{equation}
and
\begin{equation}\label{pot2}
 \left.\begin{aligned}
 &\|\zm^{-2+|\al|}\p_z^\al\hA_\tau(\tau,\z)\|_\infty\\
 &\|\zm^{-1+|\al|}\p_z^\al\hA_i(\tau,\z)\|_\infty
 \end{aligned}
 \right\}\leq C(\tau)\,,\quad
 |\al|\leq1\,,
\end{equation}
where $C(\tau)$ is a continuous function. Then the conditions of Theorem
\ref{interDirac} are fulfilled and the
 unitary propagator $U(\tau,\si)$ with the listed properties is obtained.

(ii) Let the potential $\cA(x)$ and the corresponding propagator
$U(\tau,\si)$ be as in (i). Define a new gauge
\begin{equation}
  \cA^\gau=\cA-\n\gau\,,
\end{equation}
where $\gau(x)$ is a gauge function such that for all indices
$\mu=\tau,1,2,3$ the mappings
 \begin{equation}\label{potgau}
 \big\{\mR\ni\tau\mapsto \p_z^\al\p_\mu\gau(\tau,\z)\big\}
 \in C^0(\mR,L^\infty_\mathrm{loc}(\mR^3))\,,\quad |\al|\leq1\,,
 \end{equation}
and
\begin{equation}\label{potgau2}
 \left.\begin{aligned}
 &\|\zm^{-2+|\al|}\p_z^\al\p_\tau\gau(\tau,\z)\|_\infty\\
 &\|\zm^{-1+|\al|}\p_z^\al\p_i\gau(\tau,\z)\|_\infty
 \end{aligned}
 \right\}\leq C(\tau)\,,\quad
 |\al|\leq1\,,
\end{equation}
where $C(\tau)$ is a continuous function. Denote by $H^\gau(\tau)$ the
interacting Hamiltonian with $\hA$ replaced by $\hA^\gau$, and
\begin{equation}\label{US}
 U^\gau(\tau,\si)= e^{i \gau(\tau)}U(\tau,\si)e^{-i \gau(\si)}\,.
\end{equation}
Then for such modified operators the statements (A) and (C)--(E) of
Theorem~\ref{interDirac} are satisfied.
\end{thm}
\begin{proof}
(i) The bounds \eqref{labound} are satisfied in our coordinate system.
Together with the assumed properties of $\hA_\mu$ this ensures that the
interaction term \eqref{WA} satisfies the assumptions of Theorem
\ref{interDirac}, so the thesis follows.

(ii) For (A) we note that the assumption \eqref{potgau2} implies that the
interaction term $W^\gau=\hA^\gau_\tau+\hA^\gau_i\la^i$ satisfies assumptions
imposed on $M$ in Lemma \ref{lembound} (i) and (ii) (not necessarily (iii)).
This is sufficient for the conclusion on self-adjointness of $H^\gau$
obtained as in Theorem \ref{interDirac}. Moreover, it follows from
\eqref{potgau} that $e^{i\gau(\tau)}\Dc_c(\pv^2)=\Dc_c(\pv^2)$, so (C) is
satisfied. Finally, for $\vph\in\Dc_c(\pv^2)$ we have
\begin{gather}
 i\p_\tau e^{i\gau(\tau)}\vph=-e^{i\gau(\tau)}\p_\tau \gau(\tau)\vph\,,\\
 e^{i\gau(\tau)}H(\tau)e^{-i\gau(\tau)}\vph=[H(\tau)-\la^i(\tau)\p_i\gau(\tau)]\vph\,,
\end{gather}
so the remaining statements easily follow with the use of \eqref{potgau}.
\end{proof}

\section{Scattering}\label{scattering}

We come here to our main objective in this article, scattering of the Dirac
field in an external time-dependent electromagnetic field. With the
assumptions of Theorem \ref{gengau}, augmented by some decay conditions
formulated in the two assumptions below, we shall obtain the complete
description of scattering in terms of the Cauchy surfaces of constant $\tau$,
as to be found in Theorem~\ref{scat}. The existence of the wave operators
needs only a rather simple additional Assumption \ref{asum0}, but their
unitarity and  completeness are more demanding, and they follow from
Assumption \ref{asum}. It is with regard to this latter question that we need
to discuss some further notation and properties.

It is easy to see that the operators
\begin{equation}\label{hatH}
 \mH=\tfrac{1}{2}[\la^i,\pi_i]_+=H-\mu-\hA_\tau\,,\quad
 \mH_0=\tfrac{1}{2}[\la^i,p_i]_+=H_0-\mu
\end{equation}
have similar properties as $H$ and $H_0$, in particular they are essentially
self-adjoint on $\Dh$, so $\Dc_c(\pv)=\Hc_c\cap\Dc(\pv)$ is contained in
their domains. Therefore, $\Dc_c(\pv^2)$ is contained in the domains of
$\mH^2$ and $\mH_0^2$ (as well as $H^2$ and $H_0^2$). Moreover, we note for
later use that
\begin{equation}\label{antHbe}
 [\mH,\beta]_+=[\mH_0,\beta]_+=0\,,
\end{equation}
which is a consequence of \eqref{comlabe}.

We shall need a more explicit form of $\mH^2$ below. We calculate (with
$\rho$ defined in \eqref{defro})
\begin{multline}\label{Hti2}
 \mH^2=(\pi_i\la^i+\tfrac{i}{2}\p\cdot\la)(\la^j\pi_j-\tfrac{i}2{}\p\cdot\la)\\
 =\pi_i\rho^{ij}\pi_j+\tfrac{1}{2}\pi_i[\la^i,\la^j]\pi_j
 +\tfrac{i}{2}(\p\cdot\la)\la^j\pi_j-\tfrac{i}{2}\pi_j\la^j(\p\cdot\la)
 +\tfrac{1}{4}(\p\cdot\la)^2\,.
\end{multline}
The second term on the rhs above may be written in two alternative ways:
\begin{equation}\label{pillpi}
\begin{aligned}
  \tfrac{1}{2}\pi_i[\la^i,\la^j]\pi_j&=-\tfrac{i}{2}(\p_i[\la^i,\la^j])\pi_j
  +\tfrac{1}{2}[\la^i,\la^j]\pi_i\pi_j\\
  &=\tfrac{i}{2}\pi_i(\p_j[\la^i,\la^j])+
  \tfrac{1}{2}\pi_i\pi_j[\la^i,\la^j]\,.
\end{aligned}
\end{equation}
Taking into account that $[\la^i,\la^j]$ is antisymmetric in the indices, one
can replace the product $\pi_i\pi_j$ by
$\frac{1}{2}[\pi_i,\pi_j]=-\frac{i}{2}\hF_{ij}$, and then replace
$[\la^i,\la^j]$ multiplying this expression by $2\la^i\la^j$.  Taking now one
half of the sum of the two expressions in \eqref{pillpi}, we find
\begin{equation}\label{pillpiF}
 \tfrac{1}{2}\pi_i[\la^i,\la^j]\pi_j=\tfrac{i}{4}\big[\pi_j,\p_i[\la^j,\la^i]\big]_+
 -\tfrac{i}{2}\la^i\la^j\hF_{ij}\,.
\end{equation}
For the next two terms on the rhs of \eqref{Hti2} we note
\begin{equation}\label{dll}
 \tfrac{i}{2}(\p\cdot\la)\la^j\pi_j-\tfrac{i}{2}\pi_j\la^j(\p\cdot\la)
 =-\tfrac{i}{4}[\pi_j,[\la^j,\p\cdot\la]]_+-\tfrac{1}{4}\p_j[\p\cdot\la,\la^j]_+\,,
\end{equation}
which is shown in a somewhat similar way as the former identity. The sum of
the first terms on the rhs of \eqref{pillpiF} and \eqref{dll} gives
$\tfrac{i}{4}[\pi_j,[\p_i\la^j,\la^i]]_+$. In this way we obtain
\begin{multline}
 \mH^2=\pi_i\rho^{ij}\pi_j+\tfrac{i}{4}[\p_i\la^j,\la^i]\pi_j
 +\tfrac{i}{4}\pi_j[\p_i\la^j,\la^i]\\
 -\tfrac{1}{4}\p_j([\la^j,\p\cdot\la]_+)+\tfrac{1}{4}(\p\cdot\la)^2
 -\tfrac{i}{2}\la^i\la^j\hF_{ij}\,.
\end{multline}
With the use of further notation
\begin{gather}
 \La_j=\tfrac{i}{4}(\rho^{-1})_{jk}[\p_i\la^k,\la^i]\,,\quad
 \pi_{\La i}=\pi_i+\La_i\,,\label{lapila}\\[1ex]
 Q=\tfrac{1}{4}\p_j([\la^j,\p\cdot\la]_+)-\tfrac{1}{4}(\p\cdot\la)^2\,,\quad
 N=Q+\La_i\rho^{ij}\La_j\,,\label{QN}\\[1ex]
 \Bc=\tfrac{i}{2}\la^i\la^j\hF_{ij}\,,\label{B}
\end{gather}
we can write
\begin{equation}\label{H2}
 \mH^2=\pi_{\La}\rho\pi_{\La}-N-\Bc\,,
\end{equation}
where in the first term on the rhs a symbolic notation for summation over
indices is used. A straightforward calculation shows that both $Q$ and $N$
are positive numerical functions (times the unit matrix; see Appendix
\ref{spvar}, formula \eqref{estQ} for $Q$, and then for $N$ this is obvious).
Therefore, if we further denote
\begin{equation}\label{defs}
 s=\sqrt{\rho}\,,
\end{equation}
then for $\vph\in\Dc_c(\pv^2)$ we have
\begin{equation}\label{hatHbo}
 \|\mH\vph\|\leq\big(\|s\pi_\La\vph\|^2-(\vph,\Bc\vph)\big)^{\frac{1}{2}}
 \leq \|s\pi_\La\vph\|+|(\vph,\Bc\vph)|^{\frac{1}{2}}
\end{equation}

Next denote
\begin{equation}\label{Xa}
 X=\tfrac{1}{2}\mu^{-1} \hA_i\la^i=\tfrac{1}{2}\beta a_i\la^i\,,\qquad
 a_i=(g^{\tau\tau})^{\frac{1}{2}}\hA_i\,.
\end{equation}
Below we shall need the following identity valid, with our assumptions on
$\hA_i$, on $\Dc_c(\pv)$:
\begin{equation}\label{XhatH}
 \beta[X,\mH]
 =a_i\rho^{ij}\pi_{\La j}-\tfrac{i}{2}\la^i\la^j\p_i a_j
 -\tfrac{i}{2}(\p_i\rho^{ij})a_j\,.
\end{equation}
To show this we note that $[X,\mH]=\tfrac{1}{2}\beta[\la^ia_i,\mH]_+$, write
$\mH=\la^j\pi_j-\tfrac{i}{2}\p\cdot\la$, and then the lhs of \eqref{XhatH}
takes the form
\begin{multline}
 \tfrac{1}{2}a_i[\la^i,\la^j]_+\pi_j-\tfrac{i}{2}\la^j\p_j(\la^i a_i)
 -\tfrac{i}{4}a_i [\la^i,\p_j\la^j]_+\\
 =a_i\rho^{ij}\pi_{\La j}-\tfrac{i}{2}\la^j\la^i\p_j a_i
 -\tfrac{i}{2}a_i\big(\la^j\p_j\la^i+\tfrac{1}{2}[\la^i,\p_j\la^j]_+
 +\tfrac{1}{2}[\p_j\la^i,\la^j]\big)\,,
\end{multline}
where after the equality sign we have added and subtracted the term
$a_i\rho^{ij}\La_j$. It is now easy to show that the terms in parentheses
multiplying $a_i$ sum up to $\tfrac{1}{2}\p_j[\la^j,\la^i]_+$, which ends the
proof of \eqref{XhatH}.

The spreading of the past and future is characterized in our coordinate
system by Lemma \ref{pastfut} in Appendix \ref{spvar}. Denote
\begin{equation}
 \Dc(r,\pv^2)=\{\psi\in \Dc(\pv^2)\mid \psi(\z)=0\ \text{for}\ |\z|\geq r\}\,,
\end{equation}
so that
\begin{equation}
 \Dc_c(\pv^2)=\bigcup_{r>0}\Dc(r,\pv^2)\,.
\end{equation}
We set $\tau_0=0$ in this lemma, and replace $r_0$ and $r$ by $r$ and
$r(\tau)$, respectively, so that
\begin{equation}
 r(\tau)=\rrm|\tau|+r\tm\,.
\end{equation}
Then according to this lemma, and statement (C) of Theorem \ref{interDirac},
we have
\begin{equation}\label{spread}
 \psi\in\Dc(r,\pv^2)\quad \Longrightarrow\quad
 U(\tau,0)\psi\in \Dc(r(\tau),\pv^2)\,.
\end{equation}
For a measurable function $f(\tau,\z)$ we define a semi-norm function
\begin{equation}
 \tau\mapsto\|f\|_{r,\tau}=\essup_{|\z|\leq r(\tau)}|f(\tau,\z)|\,,
\end{equation}
and then for $\psi$ as in \eqref{spread} we find
\begin{equation}\label{estsemi}
 \|f(\tau,.)U(\tau,0)\psi\|\leq\|f\|_{r,\tau}\|\psi\|\,.
\end{equation}
 Note that according to Lemma \ref{pastfut} we have
\begin{equation}
 |\z|\leq r(\tau)\ \Longrightarrow\ \zm\leq \rrm\tm+r|\tau|\,,
\end{equation}
so in that case
\begin{equation}
 |\z|\leq\zm\leq 2\rrm\tm\,.
\end{equation}

For any measurable function $k(\z)$ and $r>0$ we shall denote
\begin{equation}
 \|k\|_r=\essup_{|\z|\leq r}|k(\z)|\,.
\end{equation}
For a function $f(\tau,\z)$ we obviously have
$\|f(\tau,.)\|_r\leq\|f(\tau,.)\|_{r,\tau}$.

Our scattering theorem will apply to potentials satisfying the following
conditions of increasing restrictiveness.

\begin{asu}\label{asum0}
Potential $\hA_\mu(\tau,\z)$ satisfies the assumptions of Theorem
\ref{gengau}, and in addition the mapping
$
 \tau\mapsto \hA_i(\tau,\z)
$ is in $C^1(\mR,L^\infty_\mathrm{loc}(\mR^3))$. For each $r>0$ the following
expressions are integrable on $\mR$ with respect to $\tau$:
\begin{equation}
 \|\hA_\tau\|_r\,,\quad \tm^{-2}\|\hA_i\|_r^2\,,\quad
 \tm^{-2}\|\p_i\hA_j\|_r\,,\quad \tm^{-1}\|\p_\tau\hA_i\|_r\,.
\end{equation}
\end{asu}

\begin{asu}\label{asum}
Potential $\hA_\mu(\tau,\z)$ satisfies the assumptions of Theorem
\ref{gengau}, and in addition the mapping
$
 \tau\mapsto \hA_i(\tau,\z)
$
is in $C^1(\mR,L^\infty_\mathrm{loc}(\mR^3))$. For each $r>0$ the following
expressions are integrable on~$\mR$ with respect to $\tau$:
\begin{gather}
 \|\hA_\tau\|_{r,\tau}\,,\\
 \tm^{-2}\big(\|\hA_i\|_{r,\tau}^2+\|z^i\hA_i\|_{r,\tau}^2\big)\,,\\
 \tm^{-2}\big(\|\p_i\hA_j\|_{r,\tau}+\|z^i\p_i\hA_j\|_{r,\tau}
 +\|z^j\p_i\hA_j\|_{r,\tau}+\|z^iz^j\p_i\hA_j\|_{r,\tau}\big)\,,\\
 \tm^{-1}\big(\|\p_\tau\hA_i\|_{r,\tau}+\|z^i\p_\tau\hA_i\|_{r,\tau}\big)\,.
\end{gather}
Moreover, let $\xi:[1,\infty)\mapsto \mR$ be an appropriately chosen smooth,
positive, nondecreasing function, such that $\xi(1)=1$ and
\begin{equation}\label{xib}
 \xi'(u)\leq \frac{\kappa}{u}\xi(u)\,,\ \kappa\in(0,\tfrac{1}{2})\,,\qquad
 |\xi''(u)|\leq\frac{\con}{u}\xi(u)\,,
\end{equation}
where in the second bound the constant is arbitrary.  The following bounds
are satisfied
\begin{align}
 \|\hF_{ij}\|_{r,\tau}+\|z^i\hF_{ij}\|_{r,\tau}
 &\leq \con(r)\frac{\tm}{\xi(\tm)}\,,\\
 \|\hF_{i\tau}\|_{r,\tau}+\|z^i\hF_{i\tau}\|_{r,\tau}
 &\leq \frac{\con(r)}{\xi(\tm)}\,,
\end{align}
and for each $r>0$ the following expressions are integrable on $\mR$
\begin{equation}
 \frac{\|\xi(\zm)\hA_i\|_{r,\tau}}{\tm\xi(\tm)}\,,\qquad
 \frac{\|\xi(\zm)z^i\hA_i\|_{r,\tau}}{\tm\xi(\tm)}\,.
\end{equation}
\end{asu}
Before stating the theorem we make a comment on the function $\xi$ and
consider some additional consequences of Assumption \ref{asum}.

If one sets $\xi(u)=u^\kappa\xi_0(u)$, then the first bound in \eqref{xib} is
equivalent to $\xi_0'(u)\leq0$. It follows that
\begin{equation}\label{xiukap}
 \xi(u)\leq u^\kappa\,.
\end{equation}
Therefore, $\xi$ is a slowly increasing, non-oscillating function. Examples
include:
\begin{align}
 &\xi_1(u)=u^\kappa\,,\label{uep}\\
 &\xi_2(u)=\Big[1+\frac{\kappa}{m}\log (u)\Big]^m\,,\quad m>0\,,\label{logep}
\end{align}
where $\kappa$ is as assumed in \eqref{xib}. A particular choice of $\xi$
must guarantee the validity of the assumptions (if this is possible).

By the Schwarz inequality also the following integral is finite:
\begin{equation}
 \int_{\mR}\frac{\|\hA_i\|_{r,\tau}}{\tm^2}d\tau
 \leq \sqrt{\pi}\Big(\int_{\mR}\frac{\|\hA_i\|_{r,\tau}^2}{\tm^2}d\tau\Big)^{\frac{1}{2}}<\infty
\end{equation}
and similarly for $z^i\hA_i$. Moreover, for $\tau_2\geq\tau_1$ we have
\begin{equation}
 \frac{\hA_i(\tau_2,\z)}{\langle\tau_2\rangle}
 -\frac{\hA_i(\tau_1,\z)}{\langle\tau_1\rangle}=
 \int_{\tau_1}^{\tau_2}\Big(\frac{\p_\si\hA_i(\si,\z)}{\sm}
 -\frac{\si\hA_i(\si,\z)}{\sm^3}\Big)d\si\,,
\end{equation}
and similarly for $z^i\hA_i$. Therefore, by Assumption \ref{asum} the
functions
\begin{equation}
 \tm^{-1}\|\hA_i\|_{r,\tau}\quad \text{and}\quad
 \tm^{-1}\|z^i\hA_i\|_{r,\tau}
\end{equation}
have limits for $\tau\to\pm\infty$. Limits different from zero would
contradict other assumptions, so
\begin{equation}
 \lim_{\tau\to\pm\infty}\frac{\|\hA_i\|_{r,\tau}}{\tm}
 =\lim_{\tau\to\pm\infty}\frac{\|z^i\hA_i\|_{r,\tau}}{\tm}=0\,.
\end{equation}

For the sake of the proof of the coming theorem we note the following
estimates easily obtained with the use of formula \eqref{laspec}:
\begin{equation}\label{laA}
 |\la^i\hA_i|\leq \frac{\con}{\tm}\bigg(\frac{\tmz}{\tm}|\hA_i|+|z^i\hA_i|\bigg)\,,
\end{equation}
\begin{equation}\label{laladA}
 |\la^i\la^j\p_i\hA_j|
 \leq\frac{\con}{\tm^2}\Big[\frac{\tmz^2}{\tm^2}|\p_i\hA_j|
 +\frac{\tmz}{\tm}\big(|z^i\p_i\hA_j|+|z^j\p_i\hA_j|\big)
 +|z^iz^j\p_i\hA_j|\Big]\,.
\end{equation}

\begin{thm}\label{scat}
(i) Let $\hA_\mu$ satisfy Assumption \ref{asum0}. Then the following strong
limits exist:
\begin{gather}
 \slim_{\tau\to\pm\infty}U(0,\tau)U_0(\tau,0)=\W_\mp\,,\label{waveop}\\
 \slim_{\tau\to\pm\infty}U_\Phi(0,\tau)=U_\Phi(0,\pm\infty)
 =\Omega_{\mp}U_{0\Phi}(0,\pm\infty)\,.\label{waveopphi}
\end{gather}
(ii) Let $\hA_\mu$ satisfy Assumption \ref{asum}. Then the operators $\W_\mp$
are unitary and also the following strong limits exist:
\begin{gather}
 \slim_{\tau\to\pm\infty}U_0(0,\tau)U(\tau,0)=\W_\mp^*\,, \label{waveopad}\\
 \slim_{\tau\to\pm\infty}U_\Phi(\tau,0)=U_\Phi(\pm\infty,0)
 =U_{0\Phi}(\pm\infty,0)\Omega^*_{\mp}\,.\label{waveopphiad}
\end{gather}
\end{thm}
\begin{proof}For the sake of the whole proof we assume that $\psi\in\Dc(r,\pv^2)$
and $\|\psi\|=1$.

(i) For $\tau\to\infty$, we prove the existence of the limit
\eqref{waveopphi}, from which the limit \eqref{waveop} follows with the use
of Theorem \ref{UF0}. The case $\tau\to-\infty$ is analogous.

We note the identity
\begin{equation}
 \big(1-\tfrac{1}{2}\beta\mH\big)\beta-H\big(1-\tfrac{1}{2}\beta\mH\big)
 =(\mu-\beta+\hA_\tau)\big(\tfrac{1}{2}\beta\mH-1\big)-\tfrac{1}{2}\beta\mH^2\,,
\end{equation}
where $\mH$ is as defined in \eqref{hatH} (to show this one eliminates $H$
with the use of \eqref{hatH} and takes into account the anticommutation
relation \eqref{antHbe}). Using it, we obtain the evolution equation
\begin{multline}
 i\p_\tau\big[U(0,\tau)\big(1-\tfrac{1}{2}\beta\mH\big)\Phi(\tau)\big]\psi\\
 =U(0,\tau)\big[-\tfrac{i}{2}\beta(\p_\tau\mH)
 +(\mu-\beta+\hA_\tau)\big(\tfrac{1}{2}\beta\mH-1\big)
 -\tfrac{1}{2}\beta\mH^2\big]\Phi(\tau)\psi\,.
\end{multline}
Taking into account the anticommutation relation \eqref{antHbe} we can write
this in the form
\begin{equation}\label{UU0}
\begin{aligned}
  i\p_\tau\big[U_\Phi(0,\tau)-\tfrac{1}{2}&U(0,\tau)\Phi(\tau)^*\beta\mH\big]\psi\\
 =-&U(0,\tau)\Phi(\tau)\big[\mu-\beta+\hA_\tau+\tfrac{1}{2}\beta\mH^2\big]\psi\\
 +\tfrac{1}{2}&U(0,\tau)\Phi(\tau)^*\beta\big[-i\dot{\mH}
 +(\mu-\beta+\hA_\tau)\mH\big]\psi\,.
\end{aligned}
\end{equation}
We estimate the terms in this equation, starting with the second term in
brackets on the lhs, and then going to the successive terms on the rhs. As
$|\z|\leq r$ on the support of $\psi$, each $\la^i$, $\p_i\la^j$ and
$\p_i\p_j\la^k$ give a bounding factor $\con(r)\tm^{-1}$, and each
$\p_\tau\la^i$ a factor $\con(r)\tm^{-2}$, which leads to an easy
straightforward estimation:
\begin{gather}
 \|\mH\psi\|\leq\frac{\con(r)}{\tm}\big(\|p_i\psi\|+\|\hA_i\|_r+1\big)\,,\\[1ex]
 \|(\mu-\beta)\psi\|
 \leq\frac{r^2}{2\tm^2}\,,\qquad \|\hA_\tau\psi\|\leq\|\hA_\tau\|_r\,,\\[1ex]
 \|\mH^2\psi\|\leq \frac{\con(r)}{\tm^2}\Big[\|p_ip_j\psi\|
 +\big(\|\hA_i\|_r+1\big)\big(\|\hA_i\|_r+1+\|p_i\psi\|\big)\Big]\,,
\end{gather}
\begin{gather}
 \|\dot{\mH}\psi\|
 \leq\frac{\con(r)}{\tm^2}\big(\|p_i\psi\|+\|\hA_i\|_r+1\big)
 +\frac{\con(r)}{\tm}\|\p_\tau\hA_i\|_r\,,\\[1ex]
 \|(\mu-\beta+\hA_\tau)\mH\psi\|
 \leq\frac{\con(r)}{\tm}\Big(\frac{r^2}{\tm^2}+\|\hA_\tau\|_r\Big)
 \big(\|p_i\psi\|+\|\hA_i\|_r+1\big)\,.
\end{gather}
Therefore, with the conditions of Assumption \ref{asum0} all the terms on the
rhs of \eqref{UU0} are integrable on $\mR$, so the strong limit of
$U_\Phi(0,\tau)\psi- \tfrac{1}{2}U(0,\tau)\Phi(\tau)^*\beta\mH\psi$ exists.
But the second term vanishes in the limit, so the thesis follows for
$\psi\in\Dc_c(\pv^2)$, and then by isometry for all $\psi\in\Hc$. (For the
integrability of $\tm^{-2}\|\hA_i\|_r$ and for vanishing of
$\tm^{-1}\|\hA_i\|_r$ one argues similarly as in the remarks following
Assumption \ref{asum}.)

(ii) We prove the existence of the limit \eqref{waveopad} for
$\tau\to\infty$, from which the limit \eqref{waveopphiad} follows. Combined
with the existence of the limit \eqref{waveop}, this also leads to unitarity
and the conjugation relation. The case $\tau\to-\infty$ is similar.

We note the identity
\begin{equation}
 (1+X)H-H_0(1+X)=\hA_\tau+WX
 +[X,\mH]\,,
\end{equation}
where $W$ and $X$ are defined in \eqref{WA} and \eqref{Xa}, respectively, and
we used the fact that
\begin{equation}
 [X,H-\mH]=[X,\mu]=-2\mu X=-\hA_i\la^i\,.
\end{equation}
Thus
\begin{multline}\label{U0U}
 i\p_\tau\big[U_0(0,\tau)(1+X)U(\tau,0)\big]\psi\\
 =U_0(0,\tau)\Big(i\p_\tau X+\hA_\tau+WX
 +[X,\mH]\Big)U(\tau,0)\psi\,.
\end{multline}
 If we can show that the norm of the rhs of \eqref{U0U} is
integrable over $[0,+\infty)$, then the strong limit of
$U_0(0,\tau)(1+X)U(\tau,0)\psi$ for $\tau\to\infty$ exists. But with the use
of formula \eqref{estsemi}, and taking into account \eqref{laA} and the value
of $g^{\tau\tau}$ to be found in Appendix \ref{spvar}, we obtain
\begin{equation}\label{Xnorm}
 \|X(\tau)U(\tau,0)\psi\|\leq\|X\|_{r,\tau}
 \leq \frac{\con}{\tm}\Big(\|\hA_i\|_{r,\tau}+\|z^i\hA_i\|_{r,\tau}\Big)\,,
\end{equation}
so this norm vanishes in the limit, which then implies the desired result.

We estimate the norms of the successive terms on the rhs of \eqref{U0U},
again with the use of \eqref{estsemi}. Differentiating formula \eqref{gla}
with respect to $\tau$ we find
\begin{multline}\label{dXnorm}
 \|(\p_\tau X)U(\tau,0)\psi\|\leq\|\p_\tau X\|_{r,\tau}
 \leq \frac{\con}{\tm^2}\Big(\|\hA_i\|_{r,\tau}+\|z^i\hA_i\|_{r,\tau}\Big)\\
 +\frac{\con}{\tm}\Big(\|\p_\tau\hA_i\|_{r,\tau}+\|z^i\p_\tau\hA_i\|_{r,\tau}\Big)\,,
\end{multline}
which is integrable. The norm of the second term is bounded by
$\|\hA_\tau(\tau)\|_{r,\tau}$, which is integrable by assumption. Next we
note that
\begin{equation}
 \beta WX=\hA_\tau\beta X-\tfrac{1}{2}(g^{\tau\tau})^{\frac{1}{2}}\hA_i\rho^{ij}\hA_j\,,
\end{equation}
so using the explicit form of $(g^{\tau\tau})^{\frac{1}{2}}\rho^{ij}$, see
\eqref{rhoex}, we estimate the norm of the third term on the rhs of
\eqref{U0U} by
\begin{equation}
 \|WX\|_{r,\tau}\leq \con\|\hA_\tau\|_{r,\tau}\|X\|_{r,\tau}
 +\frac{\con\rrm}{\tm^2}\Big(\|\hA_i\|^2_{r,\tau}+\|z^i\hA_i\|^2_{r,\tau}\Big)\,,
\end{equation}
which ensures integrability.

To estimate the norm of the fourth term we use \eqref{XhatH}, \eqref{laladA}
and \eqref{divro} to find
\begin{multline}
 \|\la^i\la^j\p_i a_j+(\p_i\rho^{ij})a_j\|_{r,\tau}
 \leq \con\frac{\rrm}{\tm^2}\|\p_i\hA_j\|_{r,\tau}\\
 +\frac{\con}{\tm^2}\Big(\|\hA_i\|_{r,\tau}+\|z^i\hA_i\|_{r,\tau}
 +\|z^i\p_i\hA_j\|_{r,\tau}+\|z^j\p_i\hA_j\|_{r,\tau}
 +\|z^iz^j\p_i\hA_j\|_{r,\tau}\Big)\,,
\end{multline}
which again is integrable.

 We are now left with the single term $a\rho\pi_\La
U(\tau,0)\psi$. As it turns out, in this case the methods applied up to now
are insufficient and it is here that we make use of the function $\xi$. We
write $\psi_\tau= U(\tau,0)\psi$ and  note that
\begin{align}
 \|a\rho\pi_{\La}\psi_\tau\|
 &\leq \|sa\xi(\zm)\|_{r,\tau}\|\xi(\zm)^{-1}s\pi_\La\psi_\tau\|\\
 &\leq \tm^{-1}\Big(\|\xi(\zm)\hA_i\|_{r,\tau}
 +\|\xi(\zm) z^i\hA_i\|_{r,\tau}\Big)\|\xi(\zm)^{-1}s\pi_\La\psi_\tau\|\,.
\end{align}
Now, the estimation of $\|\xi(\zm)^{-1}s\pi_\La\psi_\tau\|$ is the most
difficult part of the proof and we shift it to the lemma below. Substituting
its result in the above estimate and using Assumption \ref{asum} one
completes the proof of the existence of the limit $\tau\to+\infty$.
\end{proof}

\begin{lem}\label{scatlem}
Under the conditions of Assumption \ref{asum}, for $\psi\in \Dc_c(\pv^2)$ the
following estimate holds
 \begin{equation}
 \|\xi(\zm)^{-1}s\pi_\La U(\tau,0)\psi\|\leq \con(\psi)\xi(\tm)^{-1}\,.
 \end{equation}
\end{lem}

\begin{proof}
We assume again that $\psi\in\Dc(r,\pv^2)$ and $\|\psi\|=1$. We observe that
the lhs of the inequality may be equivalently replaced by $\|s\pi_\La
\xi(\zm)^{-1}U(\tau,0)\psi\|$. Indeed, we have
\begin{equation}
 \big\|\big[s\pi_\La,\xi(\zm)^{-1}\big]\big\|_\infty
 = \Big\|\frac{\z\xi'(\zm)}{\tm\xi(\zm)^2}\Big\|_\infty
 \leq\Big\|\frac{\kappa\,\z}{\tm\zm\xi(\zm)}\Big\|_\infty
 \leq\frac{\kappa}{\tm}\,.
\end{equation}
Now we shall estimate the norm squared
\begin{equation}
 \|s\pi_\La\xi(\zm)^{-1}\psi_\tau\|^2
 =(\pi_\La\xi(\zm)^{-1}\psi_\tau,\rho\pi_\La\xi(\zm)^{-1}\psi_\tau)
\end{equation}
by first finding a differential inequality, and then integrating. Preparing
for that, we denote
\begin{equation}\label{defnu}
 \nu=i\xi(\zm)^{-1}[H,\xi(\zm)]=\frac{\la^i\p_i\xi(\zm)}{\xi(\zm)}
 =\frac{\xi'(\zm)}{\tm\xi(\zm)}z^i\al^i\,,
\end{equation}
with the standard notation $\al^i=\beta\gamma^i$, and observe that
\begin{equation}\label{xiH}
 \xi(\zm)^{-1} H=\mH\xi(\zm)^{-1}+(\hA_\tau+\mu-i\nu)\xi(\zm)^{-1}\,.
\end{equation}
To shorten notation we shall write $\psi_\tau^\xi=\xi(\zm)^{-1}\psi_\tau$.
Looking at the explicit form of $\Lambda_i$ \eqref{Laspec} we note that
\begin{equation}
 [\Lambda_i,\mu]=0\,,\qquad [\Lambda_i,\nu]_+=0\,.
\end{equation}
Now calculate
\begin{multline}
 \p_\tau[\pi_\La\psi_\tau^\xi]
 =-i\pi_\La\xi(\zm)^{-1} H\psi_\tau+(\p_\tau\La+\p_\tau\hA)\psi_\tau^\xi\\
 =-i\pi_\La \mH\psi_\tau^\xi
 +(\dot{\La}+\hF_{\tau .}+\p(i\nu-\mu)+2\nu\La)\psi_\tau^\xi
 -(\nu+i\hA_\tau+i\mu)\pi_\La\psi_\tau^\xi\,,
\end{multline}
and
\begin{equation}\label{difpixipsi}
\begin{aligned}
 \p_\tau(\pi_\La\psi_\tau^\xi,\rho\pi_\La\psi_\tau^\xi)
 &-(\pi_\La\psi_\tau^\xi,\dot{\rho}\pi_\La\psi_\tau^\xi)
 =-2(\pi_\La\psi_\tau^\xi,\rho\nu\pi_\La\psi_\tau^\xi)\\
 &-i(\pi_\La\rho\pi_\La\psi_\tau^\xi,\mH\psi_\tau^\xi)
 +i(\mH\psi_\tau^\xi,\pi_\La\rho\pi_\La\psi_\tau^\xi)\\
 &+2\Rp(\pi_\La\psi_\tau^\xi,\rho[\dot{\La}
 +\hF_{\tau .}+\p(i\nu-\mu)+2\nu\La]\psi_\tau^\xi)\,,
\end{aligned}
\end{equation}
where in both identities $\hF_{\tau .}$ denotes $\hF_{\tau i}$ with the index
$i$ suppressed (in the second identity summation over this index is implied).
In the first identity we have used \eqref{xiH} and commuted $\pi_\La$ with
the term $(\hA_\tau+\mu-i\nu)$.

From now on we continue the proof for $\tau\geq0$; for $\tau\leq0$ the proof
is analogous, but equation \eqref{difpixipsi} has to be multiplied by $-1$
before continuing. The first term on the rhs of \eqref{difpixipsi} is bounded
in absolute value by
\begin{equation}
 2\|\nu\|_\infty(\pi_\La\psi_\tau^\xi,\rho\pi_\La\psi_\tau^\xi)
 \leq 2\kappa\tm^{-1}\|s\pi_\La\psi_\tau^\xi\|^2\,,
\end{equation}
where we used the estimate given in \eqref{nusnu}. With the use of formula
\eqref{H2} the second line of equation \eqref{difpixipsi} takes the form
$2\Rp i(\mH\psi_\tau^\xi,(N+\Bc)\psi_\tau^\xi)$, and thanks to the estimate
\eqref{hatHbo} is bounded in absolute value by
\begin{equation}
 2\big(\|s\pi_\La\psi_\tau^\xi\|+|(\psi_\tau^\xi,\Bc\psi_\tau^\xi)|^{\frac{1}{2}}\big)\|(N+\Bc)\psi_\tau^\xi\|\,.
\end{equation}
The third line in \eqref{difpixipsi} is bounded by
\begin{equation}
 2\|s\pi_\La\psi_\tau^\xi\|
 \|s(\dot{\La}+\hF_{\tau\,.}+\p(i\nu-\mu)+2\nu\La)\psi_\tau^\xi\|\,.
\end{equation}
Finally, we observe that for $\tau>0$ we have (see \eqref{dtro})
\begin{equation}
 \dot{\rho}\leq-\frac{2\tau}{\tm^2}\rho\,,
\end{equation}
which allows us to use \eqref{difpixipsi} for the following estimate:
\begin{equation}\label{eqtri}
 \p_\tau\|s\pi_\La\psi_\tau^\xi\|^2\leq -b\|s\pi_\La\psi_\tau^\xi\|^2+2c\|s\pi_\La\psi_\tau^\xi\|+d\,,
\end{equation}
where
\begin{gather}
 b=\frac{2\tau}{\tm^2}-\frac{2\kappa}{\tm}\,,\label{b}\\[1ex]
 c=\|(N+\Bc)\psi_\tau^\xi\|+\|s(\dot{\La}+\hF_{\tau\,.}+\p(i\nu-\mu)
 +2\nu\La)\psi_\tau^\xi\|\,,\label{c}\\[2ex]
 d=2|(\psi_\tau^\xi,\Bc\psi_\tau^\xi)|^{\frac{1}{2}}\|(N+\Bc)\psi_\tau^\xi\|\,.\label{d}
\end{gather}
The second term on the rhs of \eqref{eqtri} may be estimated as follows
\begin{align}
 2c\|s\pi_\La\psi_\tau^\xi\|
 &=2\Big[\frac{1-2\kappa}{\tm}\Big]^{\frac{1}{2}}
 \|s\pi_\La\psi_\tau^\xi\|\ \Big[\frac{\tm}{1-2\kappa}\Big]^{\frac{1}{2}}c\\
 &\leq\frac{1-2\kappa}{\tm}\|s\pi_\La\psi_\tau^\xi\|^2+\frac{\tm c^2}{1-2\kappa}\,,
\end{align}
 which results in the inequality
\begin{equation}\label{eqtri1}
  \p_\tau\|s\pi_\La\psi_\tau^\xi\|^2\leq -b_0\|s\pi_\La\psi_\tau^\xi\|^2+d_0\,,
\end{equation}
where
\begin{equation}\label{b0d0}
 b_0=\frac{2\tau}{\tm^2}-\frac{1}{\tm}\,,\quad
 d_0=\frac{\tm c^2}{1-2\kappa}+d\,.
\end{equation}
We set
\begin{equation}
 \|s\pi_\La\psi_\tau^\xi\|^2=\exp\Big(-\int_0^\tau b_0(\si)d\si\Big)f(\tau)
 =\frac{\tau+\tm}{\tm^2}f(\tau)
\end{equation}
and then \eqref{eqtri1} takes the form
\begin{equation}
 \p_\tau f\leq \frac{\tm^2}{\tau+\tm}d_0(\tau)\leq \tm d_0(\tau)\,.
\end{equation}
We note that $f(0)=\|[s\pi_\La\psi_\tau^\xi]_{\tau=0}\|^2$ and find
\begin{equation}\label{spb}
 \|s\pi_\La\psi_\tau^\xi\|^2
 \leq \frac{2}{\tm}\bigg[\|[s\pi_\La\psi_\tau^\xi]_{\tau=0}\|^2
 +\int_0^\tau \sm d_0(\si)d\si\bigg]
\end{equation}
We have to estimate $d_0$ defined in \eqref{b0d0}. We start by estimating
$c$. For the terms depending on the electromagnetic field we have (we use the
form of $s$ and estimates of $\la^i$ given in Appendix \ref{spvar})
\begin{equation}\label{Best}
 \|\Bc\|_{r,\tau}\leq\frac{\con\rrm^2}{\tm^2}
 \big(\|\hF_{ij}\|_{r,\tau}+\|z^i\hF_{ij}\|_{r,\tau}\big)
 \leq \frac{\con(r)}{\tm\xi(\tm)}\,,
\end{equation}
\begin{equation}
 \|s\hF_{\tau\,.}\|_{r,\tau}
 \leq \frac{\con\rrm}{\tm}
 \big(\|\hF_{\tau i}\|_{r,\tau}+\|z^i\hF_{\tau i}\|_{r,\tau}\big)
 \leq \frac{\con(r)}{\tm\xi(\tm)}\,.
\end{equation}
For the term $s^{ij}\p_j(i\nu-\mu)$, using the estimates \eqref{esdemu} and
\eqref{esdenu} in Appendix \ref{spvar} we find\footnote{The problem of
estimation of the term $s^{ij}\p_j\mu$ is the ultimate reason for our
introduction of the function $\xi$. Without it, the bound in \eqref{estpmu}
would have the form $\con(r)\tm^{-1}$, which would be insufficient for our
application.}
\begin{equation}\label{estpmu}
 \|s\p(i\nu-\mu)\psi_\tau^\xi\|
 \leq\Big\|\frac{s\p(i\nu-\mu)}{\xi(\zm)}\Big\|_{r,\tau}
 \leq\frac{\con}{\tm^2}\Big\|\frac{\zm}{\xi(\zm)}\Big\|_{r,\tau}\leq\frac{\con(r)}{\tm\xi(\tm)}\,,
\end{equation}
where we used the fact, that both $u/\xi(u)$ as well as $\xi(u)$ are
increasing, so
\begin{equation}
 \frac{\zm}{\xi(\zm)}\leq\frac{2\rrm\tm}{\xi(2\rrm\tm)}\leq2\rrm\frac{\tm}{\xi(\tm)}\,.
\end{equation}
The estimation of the other terms in \eqref{c} uses the bounds \eqref{LaLa},
\eqref{estN} and \eqref{nusnu} and gives
\begin{equation}\label{NLaLa}
 \|N\psi_\tau^\xi\|+\|s\dot{\La}\psi_\tau^\xi\|
 +\|2\nu\La\psi_\tau^\xi\|\leq\frac{\con}{\tm^2}\,,
\end{equation}
so summing up we have
\begin{equation}
 c\leq\frac{\con(r)}{\tm\xi(\tm)}\,,\qquad \tm\, c^2\leq\frac{\con(r)}{\tm\xi(\tm)^2}\,.
\end{equation}
The use of \eqref{Best} and \eqref{NLaLa} shows that also
\begin{equation}
  d(\tau)\leq \frac{\con(r)}{[\tm\xi(\tm)]^{\frac{3}{2}}}
 \leq\frac{\con(r)}{\tm\xi(\tm)^2}\,,
\end{equation}
where we used the bound $u^\frac{1}{2}\geq\xi(u)\geq1$, see \eqref{xiukap}.
Summing up, we obtain
\begin{equation}
 d_0(\tau)\leq \frac{\con(r)}{\tm\xi(\tm)^2}\,.
\end{equation}
Now, it follows from \eqref{xib} that $u^\kappa/\xi(u)$ is an increasing
function. Therefore,
\begin{align}
 \int_0^\tau \sm d_0(\si)d\si
 &\leq  \con(r)\int_0^\tau \frac{d\si}{\xi(\sm)^2}
 \leq \con(r) \frac{\tm^{2\kappa}}{\xi(\tm)^2}\int_0^\tau\frac{d\si}{\sm^{2\kappa}}\\
 &\leq \con(r)\frac{\tm}{\xi(\tm)^2}\,.
\end{align}
This, when used in \eqref{spb}, gives
\begin{equation}
 \|s\pi_\La\psi_\tau^\xi\|^2
 \leq\frac{2\|[s\pi_\La\psi_\tau^\xi]_{\tau=0}\|^2}{\tm}
 +\frac{\con(r)}{\xi(\tm)^2}
 \leq \frac{\con(\psi)}{\xi(\tm)^2}\,,
\end{equation}
where for the second inequality we used \eqref{xiukap}.
\end{proof}

\section{Typical electromagnetic field and its special gauges}\label{typspec}

Assumption \ref{asum}, on which our main theorem on scattering \ref{scat} is
based, is rather technical and not easy for interpretation. Here we formulate
a rather typical situation met in scattering processes.\footnote{Possible
oscillating terms in the asymptotic behavior of charged currents are not
taken into account. One can expect that fields produced by such terms decay
more rapidly than those considered here; see also Discussion.} We show that
the electromagnetic field thus identified admits a gauge in which Assumption
\ref{asum} is satisfied. Moreover, there is a class of gauges which need not
satisfy this assumption, but still assure a similar asymptotic structure.

The retarded and advanced potentials are defined in terms of the source
current $J$ in standard way (as $\varphi_{\ret/\adv}$ is defined in terms of
$\rho$ in \eqref{adv} in Appendix \ref{decay}). Also, the radiated field of
the current $J$ is defined in standard way, $A_\rad=A_\ret-A_\adv$. The
Heaviside step function is denoted by $\theta$.

\begin{asu}\label{FAJ}
The Lorenz potential $A_a$ of the electromagnetic field $F_{ab}$ is given by
\begin{equation}
 A=A_\ret+A_\inc=A_\adv+A_\out\,,
\end{equation}
where $A_\ret$ is the retarded potential of a current $J$ satisfying the
assumptions listed below, and $A_\inc$ is the radiated potential of another
current $J_\inc$ with similar properties as $J$. Then also $A_\adv$ is the
advanced potential of the current~$J$, and $A_\out$ is the radiated potential
of the current $J_\out=J+J_\inc$, which has similar properties as $J$ and
$J_\inc$.

The conserved current $J(x)$ is of class $C^3$ and for some
$0<\vep<\frac{1}{2}$ satisfies the following estimates:
\begin{align}
 &|\n^\al J(x)|
 \leq\frac{\con}{(|x|+1)^{3+|\al|}}\bigg[\theta(x^2)+\frac{1}{(|x|+1)^\vep}\bigg]\,,
 & &\text{for}\quad |\al|\leq3,\label{estJ}\\[1ex]
 &|\n^\al(x\cdot\n+3) J(x)|\leq\frac{\con}{(|x|+1)^{3+|\al|+\vep}}\,,\label{esthomJ}
 & &\text{for}\quad |\al|\leq2\,.
\end{align}
The same is assumed for $J_\inc$, and then the same follows for $J_\out$. The
potential $A$ is then of class $C^3$ on the Minkowski spacetime.
\end{asu}
Note that $A$ is a linear combination of retarded and advanced potentials of
currents satisfying \eqref{estJ}. Therefore, the last statement of Assumption
\ref{FAJ} is a consequence of Lemma \ref{estlemrv} (i) in Appendix
\ref{decay}.
\begin{thm}\label{totelmga}
Let the electromagnetic field $F$ and its Lorenz potential $A$ satisfy
Assumption \ref{FAJ}. Define a new gauge by
\begin{equation}\label{gg}
 \cA(x)=A(x)-\n S(x)\,.
\end{equation}
Then the following holds.\\
 (i) For the choice
\begin{equation}
 S(x)=\gf(x)\equiv\log(\tm\zm)\,x\cdot A(x)\,,\label{specgauge}
\end{equation}
the potential \eqref{gg} fulfills Assumption \ref{asum}, so Theorems
\ref{gengau} and \ref{scat} are satisfied. \\
 (ii) Let $S(x)$ be another gauge function, such that the difference
\begin{equation}
 \gau(x)=S(x)-\gf(x)
\end{equation}
satisfies the assumptions of Theorem \ref{gengau} (ii), so that the thesis of
this theorem is true. Suppose, in addition, that there exist point-wise
limits
\begin{equation}\label{limgau}
 \lim_{\tau\to\pm\infty}\gau(\tau,\z)\equiv \gau_\pm(\z)\,.
\end{equation}
Then the potential \eqref{gg} satisfies the thesis of Theorem \ref{scat}
(but not necessarily the assumption of this theorem).\\
 (iii) In particular, the gauges defined by:
\begin{align}
 &\text{(a)}\qquad \glo(x)=\log\tm \, x\cdot A(x)\,,\label{consprop}\\
 &\text{(b)}\qquad \gtr(x)=\int_0^\tau \p_\si x(\si,\z)\cdot A(\si,\z)d\si\,,\label{spacegau}
\end{align}
are in the class defined in (ii). In case (b) one has $\hA_\tau=0$ and
\begin{equation}\label{Aperp}
 \hA_i(\tau,\z)=\hat{A}_i(0,\z)-\int_0^\tau\hF_{i\si}(\si,\z)d\si\,.
\end{equation}
\end{thm}
\begin{rem}\label{lorinv}
Assumption III, and consequently the validity of Theorem~\ref{totelmga}, is
independent of the choice of the time axis for the definition of the
foliation~\eqref{tauzed}.
\end{rem}
\begin{proof}[Proof of Thm.\ \ref{totelmga}]
(i) This potential obeys Theorem \ref{elmbo} in Appendix \ref{spgava}. As
$\cA$ is of class $C^2$ on the Min\-kow\-ski spacetime, so the assumption
\eqref{pot} of Theorem \ref{gengau} and the assumption on continuous
differentiability of the mapping $\tau\mapsto \hA_i(\tau,\z)$ in Assumption
\ref{asum} are clearly satisfied. Denote
\begin{equation}
 \xi(u)=u^\kappa\,,\qquad \kappa<\vep<\tfrac{1}{2}\,.
\end{equation}
The estimates of Theorem \ref{elmbo} imply then the following norm bounds:
\begin{gather}
 \|\hF_{i\tau}\|_\infty\,,\
 \|z^i\hF_{i\tau}\|_\infty\leq\frac{\con}{\tm}\,,\\[1ex]
 \|\hF_{ij}\|_\infty\,,\
 \|z^i\hF_{ij}\|_\infty\leq\con\,,\\[1ex]
 \|\hA_\tau\|_{r,\tau}\leq\con\Big(\frac{1+\log\tm}{\tm^{1+\vep}}
 +\frac{\log\rrm}{\tm^3}\Big)\,,\\[1ex]
 \|(1+\log\zm)^{-1}\hA_\tau\|_\infty+\|\p_i\hA_\tau\|_\infty+\|z^i\p_i\hA_\tau\|_\infty
 \leq\con\frac{1+\log\tm}{\tm^{1+\vep}}\,,\\
 \|\p_i\p_j\hA_\tau\|_\infty\leq\con(1+\log\tm)\,,\\[1ex]
  \|\hA_i\|_\infty+\|z^i\hA_i\|_\infty
  \leq\|\xi(\zm)\hA_i\|_\infty+\|\xi(\zm)z^i\hA_i\|_\infty
  \leq\con(1+\log\tm)\,,
\end{gather}
\begin{gather}
 \|\p_\tau\hA_i\|_\infty+\|z^i\p_\tau\hA_i\|_\infty
 \leq\frac{\con}{\tm}\,,\\[1ex]
 \|\p_i\hA_j\|_\infty+\|z^i\p_i\hA_j\|_\infty
 +\|z^j\p_i\hA_j\|_\infty+\|z^iz^j\p_i\hA_j\|_\infty
 \leq\con(1+\log\tm)\,.
\end{gather}
It is now easily checked that the potential satisfies assumption \eqref{pot2}
in Theorem \ref{gengau} and all the remaining bounds in Assumption
\ref{asum}, from which the thesis follows. (Note that in this case all
expressions are bounded in $L^\infty$-norm, except for $\hA_\tau$.)

(ii) Let the evolution operator $U(\tau,\si)$ refer to the potential defined
in~(i), and denote the new potential now considered by $\hA^\gau$. Following
further notation used in Theorem \ref{gengau} we have
\begin{multline}
 \lim_{\tau\to\pm\infty}U^\gau_\Phi(\tau,0)
 =\lim_{\tau\pm\infty}\Phi^*(\tau)e^{i \gau(\tau)}U(\tau,0)e^{-i\gau(0)}\\
 =\lim_{\tau\to\pm\infty}e^{i\gau(\tau)}U_\Phi(\tau,0)e^{-i\gau(0)}
 =e^{i\gau_\pm}U_\Phi(\pm\infty,0)e^{-i\gau(0)}\,,
\end{multline}
so
\begin{multline}
 \slim_{\tau\to\pm\infty}U_0(0,\tau)U^\gau(\tau,0)
 =U_{0\Phi}(0,\pm\infty)U^\gau_\Phi(\pm\infty,0)\\
 =U_{0\Phi}(0,\pm\infty)e^{i\gau_\pm}U_\Phi(\pm\infty,0)\W_\mp e^{-i\gau(0)}
 \equiv \W^\gau_\mp\,.
\end{multline}
Similarly for the limits of the conjugated operators.

(iii) Both gauges are easily seen to satisfy condition \eqref{potgau} of
Theorem \ref{gengau}; we turn to the estimates \eqref{potgau2}. In the case
(a), with $C(x)=x\cdot A(x)$, we have
\begin{equation}
 \gau(\tau,\z)=-\log\zm\,C(\tau,\z)\,,
\end{equation}
and the estimates are easily checked with the use of the results of the proof
of Theorem \ref{elmbo}. Also, it follows from the estimate of $|\p_\tau C|$
given there that $C(\tau,\z)$ has limits for $\tau\to\pm\infty$.

In the case (b), it is now sufficient to investigate the difference of the
gauge function $\gtr(x)$ as compared to the gauge function of case (a):
\begin{equation}
 \gau'(x)=\gtr(x)-\glo(x)\,.
\end{equation}
Differentiating and using the form of $\p_\tau x$ given in the proof of
Theorem \ref{elmbo} we find
\begin{gather}
 \p_\tau\gau'(\tau,\z)=\tm^{-2}\zm A_0-\log\tm\p_\tau C\,,\\
 \p_i\p_\tau\gau'(\tau,\z)=\frac{1}{\tm^2}\Big[\frac{z^i}{\zm}A_0+\zm\p_i A_0\Big]
 -\log\tm\p_i\p_\tau C\,,
\end{gather}
Now noting that $\p_i\gau(0,\z)=0$, integrating (the last term by parts) and
applying the derivative $\p_j$ we obtain (fields in the integrand depend on
$(\si,\z)$)
\begin{multline}
 \p_i\p_j\gau'(\tau,\z)
 =\int_0^\tau\Big[\frac{d^{ij}}{\zm}A_0+\frac{1}{\zm}(z^i\p_j A_0+z^j\p_i A_0)
 +\zm\p_i\p_jA_0+\si\p_i\p_j C\Big]\frac{d\si}{\sm^2}\\
 -\log\tm\p_i\p_j C\,,
\end{multline}
with $d^{ij}$ defined in \eqref{dd}.  With the use of the estimates listed in
the proof of Theorem \ref{elmbo} one finds
\begin{gather}
 |\p_\tau\gau'|\leq\con\frac{1+\log\tm}{\tm^{1+\vep}}\,,\qquad
 |\p_i\p_\tau\gau'|\leq\con\frac{1+\log\tm}{\tm^{1+\vep}}\,,\\
 |\p_i\p_j\gau'|\leq\con(1+\log\tm)\,,\qquad
 |\p_i\gau'|\leq\con\,,
\end{gather}
the fourth estimate by the integration of the second one. Thus the estimates
\eqref{potgau2} are satisfied. Finally, $\p_\tau\gau'$ is integrable on
$\mR$, so the thesis follows.
\end{proof}

\section{Discussion}\label{disc}

There are three questions we want to address in this section:
\begin{itemize}
\item[(i)] How far is the present analysis from a complete treatment of the
    Max\-well-Dirac system?
\item[(ii)] Is there a further physical selection criterion to choose a
    gauge from the class of gauges obtained in Theorem \ref{totelmga}?
\item[(iii)] Open problems.
\end{itemize}

With regard to the first of these questions we note that the form of the
charged currents producing  electromagnetic fields in Assumption \ref{FAJ}
mimics what one should expect in fully interacting theory. A possible
shortage of this assumption rests in the estimates of derivatives of the
currents, which would not be satisfied for oscillating terms in asymptotic
behavior. The Dirac field current does have such asymptotic terms, but it is
quite plausible to predict that oscillations dump the asymptotic behavior of
fields produced by them. What would be needed is an appropriately fast
vanishing of the leading oscillating asymptotic terms in the neighborhood of
the lightcone (which is a~quite reasonable prediction). This would presumably
lead to similar behavior of electromagnetic potentials as that following from
Assumption \ref{FAJ}. Moreover, we note that our Assumption \ref{asum}, on
which our analysis is based, leaves much more room for the types of
potentials than Assumption \ref{FAJ}.

Let us once more at this point recall the work by Flato et al.\cite{fst95}.
These authors do have a theorem on the evolution of the complete system, but
in a rather restricted setting and with not much control over the range of
validity. The theorem states that in the space of smooth (i.e.\ $C^\infty$)
initial data there exists a neighborhood of zero, which gives rise to the
Cauchy evolution and completeness. Our analysis is not fully developed to the
full interacting case, but gives complete results for the Dirac part of the
system, with the electromagnetic fields plausibly guessed.

Our main motivation for the present work was the expectation that the choice
of gauge does matter for asymptotic behavior and its interpretation. We have
shown that for a class of gauges the asymptotic behavior of the Dirac field
approaches that of a free field, without the need for any dynamical
corrections.

This brings us to the second question mentioned at the beginning. Theorem
\ref{totelmga} identifies a class of gauges for which the situation described
above takes place. The concrete gauge defined by \eqref{specgauge} was used
for technical reasons: this gauge satisfies Assumption \ref{asum}, and it is
the only gauge with that property among the gauges explicitly mentioned in
Theorem \ref{totelmga}. However, we want to argue now, in less precise terms
then those of the preceding sections, that a~different choice of gauges in
the given class has a definite physical interpretation when an extension, as
mentioned above, to fully interacting system is considered. Such extension,
partly based on conjectures, was considered in \cite{her95}, and we have to
briefly recall a few results of this analysis. It was found that if one uses
inside the lightcone the gauge related to the Lorenz gauge $A$ (of similar
properties as those of Assumption III) by
\begin{equation}
 \cA_\mathrm{cone}(x)=A(x)-\n\gco(x)\,,\qquad
 \gco(x)=\log\sqrt{x^2}\,x\cdot A(x)\,,
\end{equation}
or another with similar timelike asymptotic behavior, then the asymptotic
total four-mo\-men\-tum and angular momentum taken away into timelike
infinity have the same functional form as those for the free Dirac field.
Moreover, the energy-momentum radiated into null infinity is fully due to the
free outgoing electromagnetic field. However, the angular momentum going out
into the null infinity, in addition to the free radiated contribution, has a
mixed adv-out electromagnetic terms. These latter terms may be incorporated
into the free Dirac field by a change of the asymptotic limit addition to the
gauge function. As a result, both energy-momentum as well as angular momentum
are clearly separated into electromagnetic and Dirac parts.\footnote{For
explicit expressions and more extensive discussion we refer the reader to
\cite{her95}, Section~V. Here we would only like to reassure the reader that
the problems met for angular momentum in electrodynamics are accounted for in
that discussion.}\pagebreak[2] Now we would like to identify in our class of
global gauges those for which this separation may be expected. It is not
difficult to show that for points inside the future lightcone, represented by
$\la v$, with $\la>0$ and $v$ on the future unit hyperboloid, we have
\begin{equation}
 \lim_{\la\to\infty}[\glo(\la v)-\gco(\la v)]=0\
\end{equation}
for the gauge function $S_{\log}$ defined in Theorem \ref{totelmga} (iii).
Therefore, we put forward the selection criterion for the gauge functions $S$
to be used in \eqref{gg} for the definition of $\cA$:
\begin{equation}
 S(\tau,\z)=S_{\log}(\tau,\z)+\Delta S(\tau,\z)\,,
\end{equation}
where $\Delta S$ has the limit
\begin{equation}
 \Delta S(\infty,\z)=\lim_{\tau\to\infty}\Delta S(\tau,\z)
\end{equation}
which satisfies the condition of the separation of angular momentum as
described above (it is easy to see that for $\Delta S$ sufficiently regular
the rhs is equal to $\dsp\lim_{\la\to\infty}\Delta S(\la(\zm,\z))$). Similar
conditions should be applied for past infinity.

Gauges in the class thus selected have the property announced in the
introduction: $x\cdot\cA(x)$ vanishes asymptotically in timelike directions.
This may be checked easily for $S_{\log}$, and if $\Delta S(\tau,\z)$ is not
oscillating in $\tau$, the same is true for this part.

The existence and completeness of the wave operators, as indicated in Remark
\ref{lorinv}, does not depend on the choice of the time axis for the
definition of our spacetime foliation. On the other hand, and this is the
first of open problems, the precise transformation law from one inertial
observer to another needs further investigation. More generally, one can ask
what is a general class of Cauchy foliations for which the results could be
repeated, and how the results would depend on the choice in the class.

Let us end this discussion by expressing the belief that our scheme could be
applied to analogous problems on at least some of the smooth curved spacetime
backgrounds. Whether, as inquired by one of the Referees, further extension
to black hole type spacetimes would be possible, is a more speculative
question.

\appendix

\section{Transformation $K$}\label{trK}

The operator of the Lorentz rotation of Dirac spinors in the hyperplane
spanned by the timelike, future-pointing unit vectors $t$ and $n$ is given by
\begin{equation}
 K(\kappa)=\exp(\tfrac{\kappa}{2}[\ga\cdot n,\ga\cdot t])=
 \cosh(\kappa\sinh\zeta)+
 \frac{\sinh(\kappa\sinh\zeta)}{2\sinh\zeta}[\ga\cdot n, \ga\cdot t]\,,
\end{equation}
where $\kappa$ is the parameter of the rotation, $\zeta>0$ is such that
$n\cdot t=\cosh\zeta$, and the rhs is obtained by a simple calculation with
the use of the relation \mbox{$([\ga\cdot n,\ga\cdot t])^2=4(\sinh\zeta)^2$}.
Rotation of $\ga\cdot n$ gives
\begin{equation}
 K(-\kappa)\ga\cdot n K(\kappa)
 =\frac{\sinh(\zeta-2\kappa\sinh\zeta) \ga\cdot n
 +\sinh(2\kappa\sinh\zeta)\ga\cdot t}{\sinh\zeta}\,.
\end{equation}
Demanding that the coefficient at $\ga\cdot n$ vanishes we find
$\kappa_0=\zeta/(2\sinh\zeta)$ and then for $K=K(\kappa_0)$ we have
\begin{gather}
 K=\frac{1+\ga\cdot n \ga\cdot t}{\sqrt{2(1+n\cdot t)}}
 =\ga\cdot w\,\beta=K^\dagger \,,\qquad
 w=\frac{t+n}{\sqrt{(t+n)^2}}\,,\label{Kw}\\
 K^{-1}\ga\cdot n K=\ga\cdot t\,.\label{KnK}
\end{gather}

We now put $t=e_0$, and $n$ as given in \eqref{n}. With notation introduced
in Section \ref{free} we now have:
\begin{pr}\label{gK}
If $\Sigma_\tau$ are rotationally symmetric, that is $\tau=\tau(x^0,|\x|)$,
then
\begin{equation}\label{gKv}
 [\gax_K^\mu,(K^{-1}\p_\mu K)]_+=0\,.
\end{equation}
\end{pr}
\begin{proof}
The lhs of \eqref{gKv} is coordinate independent, so we can use Minkowski
coordinates. If we use $K$ in the form \eqref{Kw} then we find
\begin{equation}
 \beta[\ga_K^a,(K^{-1}\n_a K)]_+\beta
 =(\ga^b\ga^a\ga^c-\ga^c\ga^a\ga^b)w_b\n_a w_c
 =2\ga^b\ga^a\ga^c w_{[b}\n_a w_{c]}\,,
\end{equation}
where the second equality follows from the easily seen complete antisymmetry
of the expression in gamma matrices following the first equality sign. Using
the form of $w$ given in \eqref{Kw} we find
\begin{equation}\label{wdw}
 w_{[b}\n_a w_{c]}=[(t+n)^2]^{-1}(t_{[b}+n_{[b})\n_a n_{c]}\,.
\end{equation}
If the assumption of the proposition holds, then $n_a=(k_1,k_2x^i)$, where
$k_1$, $k_2$ are functions of $x^0$ and $|\x|$. In consequence, the rhs of
\eqref{wdw} vanishes and the thesis follows.
\end{proof}

\section{An open system}\label{open}

We discuss here a class of open systems for which the (general) Schr\"odinger
equation may be defined. The main tool for the discussion of this Section is
the commutator theorem, as formulated in Thm.\,X.37 in
\cite{rs79II}.\footnote{Unlike in \cite{her19}, here we make no use of the
Kato theorem.}

\subsection{Hamiltonians}\label{open-hamilt}

\begin{thm}\label{op-ham}
 Let $h_0\geq1$ be a self-adjoint operator, for which the dense subspace
     $\Dc\subseteq\Dc(h_0)$ is a~core. Suppose that for $\tau\in\mR$ the
     following holds:
 \begin{itemize}
 \item[(a)]  $h(\tau)$ is a family of symmetric operators on $\Dc$,

 \item[(b)] the commutator $[h(\tau),h_0]$, defined originally as a
     quadratic form on $\Dc$, extends to an operator on $\Dc$,
 \item[(c)] for $\psi\in\Dc$ the maps $\tau\mapsto h(\tau)\psi$,
     $\tau\mapsto [h(\tau),h_0]\psi$ are strongly continuous and the
     following bounds are satisfied:
\begin{align}
 \|h(\tau)\psi\|&\leq c_1(\tau)\|h_0\psi\|\,,\label{hbound}\\
 |(\psi,[h(\tau),h_0]\psi)|&\leq c_2(\tau) (\psi,h_0\psi)\,,\label{hhbounds1}\\
 \|[h(\tau),h_0]\psi\|&\leq c_3(\tau)\|h_0\psi\|\,,\label{hhbounds2}
\end{align}
where $c_i(\tau)$ are continuous functions.
\end{itemize}
Denote by the same symbols the closures of the operators $h(\tau)$ and
$[h(\tau),h_0]$.

Then the following holds:
\begin{itemize}
\item[(A)] $\Dc(h_0)\subseteq\Dc(h(\tau))$,
    $\Dc(h_0)\subseteq\Dc([h(\tau),h_0])$, and $h(\tau)$ are essentially
    self-adjoint on $\Dc$ and on each core of $h_0$,
\item[(B)] operators $h(\tau)h_0^{-1}$, $h_0^{-1}h(\tau)$,
    $[h(\tau),h_0]h_0^{-1}$ and $h_0^{-1}[h(\tau),h_0]$ extend to bounded
    operators, strongly continuous with respect to $\tau$, and
\begin{equation}
 \begin{aligned}\label{hhinv}
 \|h(\tau)h_0^{-1}\|&=\|h_0^{-1}h(\tau)\|\leq c_1(\tau)\,,\\
 \|[h(\tau),h_0]h_0^{-1}\|
 &=\|h_0^{-1}[h(\tau),h_0]\|\leq c_3(\tau)\,,
\end{aligned}
\end{equation}
\item[(C)] statement (c) remains true for all $\psi\in\Dc(h_0)$.
\end{itemize}
\end{thm}
\begin{proof}
As $\Dc$ is a core for $h_0$, the domains inclusions in (A) and the extension
of inequalities \eqref{hbound}--\eqref{hhbounds2} to $\psi\in\Dc(h_0)$ follow
immediately, if we denote the closures of the operators $h(\tau)$ and
$[h(\tau),h_0]$ by the same symbols. In addition, the essential
self-adjointness of $h(\tau)$ follows from \eqref{hbound} and
\eqref{hhbounds1} by the commutator theorem, Thm.\,X.37 in \cite{rs79II}.

The bounds \eqref{hhinv} for the operators $h(\tau)h_0^{-1}$ and
$[h(\tau),h_0]h_0^{-1}$ follow now immediately. The operators
$h_0^{-1}h(\tau)$ and $h_0^{-1}i[h(\tau),h_0]$, initially densely defined on
$\Dc(h_0)$, extend to conjugates of $h(\tau)h_0^{-1}$ and
$i[h(\tau),h_0]h_0^{-1}$, respectively, so all the bounds \eqref{hhinv} are
satisfied.

Strong continuity in $\tau$ of the operators $h_0^{-1}h(\tau)$ and
$h_0^{-1}[h(\tau),h_0]$ follows immediately from the bounds \eqref{hhinv} and
the assumption (c). For the continuity of $h(\tau)h_0^{-1}$ we note what
follows. For $\psi\in\Hc$ the vector $h_0^{-1}\psi$ is in $\Dc(h_0)$, so for
each $\ep>0$ there exist $\vph\in\Dc$ (a core for $h_0$) such that
$\|\psi-h_0\vph\|\leq\ep$. We estimate
\begin{equation}
 \|h(\tau)h_0^{-1}\psi-h(\si)h_0^{-1}\psi\|
 \leq \ep[c_1(\tau)+c_1(\si)]+\|h(\tau)\vph-h(\si)\vph\|\,,
\end{equation}
so the strong continuity follows. The case of $[h(\tau),h_0]h_0^{-1}$ is
analogous. This completes the proof of (B), and its consequence is the
statement of strong continuity of maps in (C).
\end{proof}

\subsection{Evolution}\label{open-evol}

\begin{thm}\label{op-ev}
 Let all the assumptions of Theorem \ref{op-ham} be satisfied. Moreover, let
 $(\tau,\si)\mapsto u(\tau,\si)$ be a strongly continuous
 unitary evolution system in~$\Hc$:
 \begin{equation}\label{evolu}
 u(\tau,\rho)u(\rho,\si)=u(\tau,\si)\,,\quad u(\tau,\tau)=1\,,
 \end{equation}
such that $u(\tau,\si)\Dc=\Dc$.
 Suppose that for $\psi\in\Dc$:
\begin{itemize}
\item[(a)]  the maps $(\tau,\si)\mapsto h(\tau)u(\tau,\si)\psi$ and
    $(\tau,\si)\mapsto h_0u(\tau,\si)\psi$ are strongly continuous,
\item[(b)] the map $(\tau,\si)\mapsto u(\tau,\si)\psi$ is of class $C^1$ in
    the strong sense and the following equations hold:
\begin{equation}\label{schrD}
 i\p_\tau u(\tau,\si)\psi=h(\tau)u(\tau,\si)\psi\,,\quad
 i\p_\si u(\tau,\si)\psi=-u(\tau,\si)h(\si)\psi\,,
\end{equation}
\end{itemize}
Then the following holds:
\begin{itemize}
\item[(A)] the operator $h_0u(\tau,\si)h_0^{-1}$ is bounded, with the norm
\begin{equation}
    \|h_0u(\tau,\si)h_0^{-1}\|\leq \exp\Big[\con\int_\si^\tau c_3(\rho)d\rho\Big]\,,
\end{equation}
so in particular
\begin{equation}
 u(\tau,\si)\Dc(h_0)=\Dc(h_0)\,;
 \end{equation}
\item[(B)] the map $(\tau,\si)\mapsto h_0u(\tau,\si)h_0^{-1}$ is strongly
    continuous;
\item[(C)] all statements of assumptions (a) and (b) remain true for all
    $\psi\in\Dc(h_0)$.
\end{itemize}
\end{thm}
\begin{proof}
To show (A) we use a standard trick.\footnote{See e.g.\ Proposition B.3.3 in
\cite{der97}.} It is easy to see that
\begin{equation}
 \Big[\frac{h_0^2}{1+\ep h_0^2},h(\tau)\Big]
 =\frac{h_0}{1+\ep h_0^2}\big[[h(\tau),h_0],h_0^{-1}\big]_+\frac{h_0}{1+\ep h_0^2}\,,
\end{equation}
so using the second bound in \eqref{hhinv} we have for $\psi\in\Dc$
\begin{equation}
 \p_\tau\Big(u(\tau,\si)\psi,
 \frac{h_0^2}{1+\ep h_0^2}u(\tau,\si)\psi\Big)
 \leq2 c_3(\tau)\Big(u(\tau,\si)\psi,
 \frac{h_0^2}{1+\ep h_0^2}u(\tau,\si)\psi\Big)\,,
\end{equation}
By Gronwall's inequality we obtain
\begin{equation}
 \Big\|\frac{h_0}{(1+\ep h_0^2)^\frac{1}{2}}u(\tau,\si)\psi\Big\|
 \leq e^{\int_\si^\tau c_3(\rho)d\rho}\Big\|\frac{h_0}{(1+\ep h_0^2)^\frac{1}{2}}\psi\Big\|\,.
\end{equation}
As $h_0$ is closed, letting $\ep\searrow 0$ we show that
$u(\tau,\si)\psi\in\Dc(h_0)$ for $\psi\in\Dc$, and
\begin{equation}
 \|h_0u(\tau,\si)\psi\|\leq e^{\int_\si^\tau c_3(\rho)d\rho}\|h_0\psi\|\,.
\end{equation}
But $\Dc$ is a core for~$h_0$, so this inequality extends to $\Dc(h_0)$.
Setting $\psi=h_0^{-1}\chi$, with arbitrary $\chi\in\Hc$, we obtain (A).

For statement (B) we note that for $\psi\in\Hc$ and $\vph\in\Dc$ such that
\mbox{$\|\psi-h_0\vph\|\leq\ep$} (as in the proof of Theorem \ref{op-ham}) we
have
\begin{multline}
 \|h_0[u(\tau,\si)-u(\tau',\si')]h_0^{-1}\psi\|\\
 \leq\|h_0[u(\tau,\si)-u(\tau',\si')]h_0^{-1}\|\ep +\|h_0u(\tau,\si)\vph-h_0u(\tau',\si')\vph\|\,,
\end{multline}
so the map $(\tau,\si)\mapsto h_0u(\tau,\si)h_0^{-1}$ is strongly continuous.

The extension of~(a) to all $\psi\in\Dc(h_0)$ now follows from the present
statement (B), and the statement (B) of Theorem \ref{op-ham}. For the
extension of (b) we write the integral forms of equations \eqref{schrD} for
$\psi\in\Dc$
\begin{equation}
 u(\tau,\si)\psi-\psi=-i\int_\si^\tau h(\rho)u(\rho,\si)\psi\,d\rho
 =-i\int_\si^\tau u(\rho,\si)h(\rho)\psi\,d\rho\,.
\end{equation}
But we already know that both integrands extend to continuous functions of
$\rho$ for $\psi\in\Dc(h_0)$, so this extension followed by differentiation
leads to the extended equations \eqref{schrD}.
\end{proof}

\section{A lemma}\label{lemma}

Here we state a result to be applied in Sections \ref{free} and
\ref{evolelmg}. The lemma below enables the application of Theorem
\ref{op-ham} to our system.

\begin{lem}\label{lembound}
Let $L^i(\z)$, $i=1,2,3$, and $M(\z)$ be four $n\times n$ hermitian matrix
functions on $\mR^3$, and denote
\begin{equation}
 h=\tfrac{1}{2}(L^ip_i+p_iL^i)+ M\,,\quad
 h_0=\tfrac{1}{2}(\pv^2+\z^2)\,.
\end{equation}
 Then for some constants $C_i$ ($i=1,2,3$)
and all $\vph\in C_0^\infty(\mR^3, \mC^n)$ the following inequalities are
satisfied:\\
 (i) if the functions $\zm^{-1}L^i(\z)$, $\p_iL^j(\z)$ and
$\zm^{-2}M(\z)$ are in $L^\infty(\mR^3)$, then
\begin{equation}\label{hh0}
 \|h\vph\|\leq C_1\|h_0\vph\|\,;
\end{equation}
 (ii) if in addition $\p_i\p_jL^k(\z)$ and $\zm^{-1}\p_i
 M(\z)$ are in $L^\infty(\mR^3)$, then
\begin{equation}
 |(h\psi,h_0\vph)-(h_0\psi,h\vph)|
 \leq C_2\|h_0^{\frac{1}{2}}\psi\|\|h_0^{\frac{1}{2}}\vph\|\,;\label{hh01}
\end{equation}
(iii) if in addition to (i) and (ii) also $\p_i\p_j\p_k L^l(\z)$ and
$\p_i\p_j M(\z)$ are in $L^\infty(\mR^3)$, then
\begin{equation}
 \|[h,h_0]\vph\|\leq C_3\|h_0\vph\|\,.\label{hh02}
\end{equation}
\end{lem}
\begin{proof}\footnote{It may be remarked that the proof admits some
weakening  of the assumptions. For instance, in (i) for the derivatives of
$L^i$ it would be sufficient that $\zm^{-2}\p_iL^i$ be in $L^\infty$; in (ii)
for the second derivatives of $L^i$ it suffices that $\zm\Delta L^i$ is in
$L^\infty$. We shall not need such more general results.}
 Before starting the proof of (i)--(iii) we note a few simple facts. For $\vph\in
 C_0^\infty(\mR^3, \mC^n)$ we note the identity
\begin{equation}
 \vph^\hc\pv^2\vph+(\pv^2\vph)^\dag\vph+\Delta|\vph|^2=2|\pv\vph|^2\,,
\end{equation}
which multiplied by $\z^2$ and integrated (transfer $\Delta$ by parts) gives
\begin{equation}\label{ident}
 \Rp(\pv^2\varphi,\z^2\varphi)+3\|\varphi\|^2=\||\z|\pv\varphi\|^2\,.
\end{equation}
Also, it follows from $\pv^2+\z^2\geq 3\cdot\1$ that
\begin{equation}\label{ineq}
 \|\vph\|\leq\tfrac{1}{3}\|(\pv^2+\z^2)\vph\|\,.
\end{equation}
Using \eqref{ident} and \eqref{ineq} one finds
\begin{multline}
 \|(\pv^2+\z^2)\vph\|^2 - \|\pv^2\vph\|^2-\|\z^2\vph\|^2
 =2\Rp(\pv^2\vph,\z^2\vph)\\
 \geq-6\|\vph\|^2
 \geq -\tfrac{2}{3}\|(\pv^2+\z^2)\vph\|^2\,,
\end{multline}
which gives the bound:
\begin{equation}\label{domains}
 \|\pv^2\vph\|^2+\|\z^2\vph\|^2\leq \tfrac{5}{3}\,\|(\pv^2+\z^2)\vph\|^2\,.
\end{equation}

We note that the $L$- and $ M$-parts of the thesis may be proved separately.
Consider the $L$-part first, and assume that $L^i$ satisfy the assumptions
in~(i). Then $N_{\al,i\beta}(\z)=\zm^{-1}L^i_{\al\beta}(\z)$ may be treated
as a bounded operator \mbox{$\mC^{3n}\otimes L^2(\mR^3)\mapsto \mC^n\otimes
L^2(\mR^3)$}, whose norm we denote $d_1$. Then
\begin{multline}
 \|L^ip_i\vph\|^2
 \leq d_1^2\sum_i\|\zm p_i\vph\|^2=d_1^2\big(\||\z|\pv\vph\|^2+(\vph,\pv^2\vph)\big)\\
 \leq d_1^2 \big(\|\pv^2\vph\|\|\z^2\vph\|+3\|\vph\|^2+\|\vph\|\|\pv^2\vph\|\big)
 \leq \tfrac{1}{2}d_1^2\big(2\|\pv^2\vph\|^2+\|\z^2\vph\|^2+7\|\vph\|^2\big)\,,
\end{multline}
where in the third step we used identity \eqref{ident} and the Schwarz
inequality. Taking into account that
\mbox{$\lab\cdot\pv+\pv\cdot\lab=2\lab\cdot\pv-i(\pav\cdot\lab)$} and
$\|\pav\cdot\lab\|_\infty\leq \con$, with the use of \eqref{ineq} and
\eqref{domains} one easily arrives at the estimate \eqref{hh0} for the
$L$-part.

Let now $L^i$ satisfy the assumptions in (ii) and let $\vph\in
C_0^\infty(\mR^3, \mC^n)$. Then using the usual rules for calculation of
commutators one finds
\begin{equation}\label{distcom}
 \big[[L^i,p_i]_+,h_0\big]\vph=-2iz^iL^i\vph
 +\tfrac{i}{2}[p_i,p_j(\p_jL^i)+(\p_jL^i)p_j]_+\vph\,,
\end{equation}
which makes sense as a distributional identity ($p_i$, $p_j$ treated
distributionally). As $L^i$ are twice differentiable, we can transform the
second term on the rhs:
\begin{equation}
 \tfrac{i}{2}[p_i,p_j(\p_jL^i)+(\p_jL^i)p_j]_+\vph
 =ip_i\big((\p_jL^i)+(\p_iL^j)\big)p_j\vph+\tfrac{i}{2}[p_i,\Delta L^i]\vph\,.
\end{equation}
Setting this into \eqref{distcom} and taking the product with $\psi\in
C_0^\infty(\mR^3, \mC^n)$, we obtain
\begin{multline}\label{comL1}
 ([L^i,p_i]_+\psi, h_0\vph) -(h_0\psi,[L^i,p_i]_+\vph)
 =i(p_i\psi,(\p_iL^j+\p_jL^i)p_j\vph)\\
 -\tfrac{1}{2}(\psi,(\Delta L^i)p_i\vph)
 +\tfrac{1}{2}((\Delta L^i)p_i\psi,\vph)
 -2i(z^i\psi,L^i\vph)\,.
\end{multline}
We treat $P_{i\al,j\beta}(\z)=\p_iL^j_{\al\beta}(\z)+\p_jL^i_{\al\beta}(\z)$
as a bounded operator acting in \mbox{$\mC^{3n}\otimes L^2(\mR^3)$}, whose
norm we denote $d_2$. Then
\begin{equation}
 |(p_i\psi,(\p_iL^j+\p_jL^i)p_j\vph)|
 \leq d_2\Big(\sum_i\|p_i\psi\|^2\sum_j\|p_j\vph\|^2\Big)^{\frac{1}{2}}
 \leq 2d_2\|h_0^{\frac{1}{2}}\psi\|\|h_0^{\frac{1}{2}}\vph\|\,.
\end{equation}
Similarly, $Q_{\al,i\beta}(\z)=\Delta L^i_{\al\beta}(\z)$ is a bounded
operator from $\mC^{3n}\otimes L^2(\mR^3)$ into $\mC^n\otimes L^2(\mR^3)$
with the norm $d_3$, and $R_{i\al,\beta}(\z)=\zm^{-1}L^i_{\al\beta}(\z)$ a
bounded operator from $\mC^n\otimes L^2(\mR^3)$ into $\mC^{3n}\otimes
L^2(\mR^3)$ with the norm $d_4$. Therefore
\begin{gather}
 |(\psi,(\Delta L^i)p_i\vph)|
 \leq d_3\|\psi\|\Big(\sum_i\|p_i\vph\|^2\Big)^\frac{1}{2}
 \leq\con\|h_0^{\frac{1}{2}}\psi\|\|h_0^{\frac{1}{2}}\vph\|\,,\\
 |(z^i\psi,L^i\vph)|\leq d_4\|\zm\vph\|\Big(\sum_i\|z^i\psi\|^2\Big)^{\frac{1}{2}}
 \leq\con\|h_0^{\frac{1}{2}}\psi\|\|h_0^{\frac{1}{2}}\vph\|\,,
\end{gather}
which ends the proof of \eqref{hh01} for the $L$-part.

If in addition the assumptions in (iii) are satisfied, then on the rhs of
\eqref{distcom} all operators $p_i$ may be commuted to the right, and the
commutator is a differential operator with functional coefficients:
\begin{equation}
[[L^i,p_i]_+,h_0]=2i(\p_i L^j)p_ip_j +(\Delta L^j+\p_j\pav\cdot\lab)p_j
-\tfrac{i}{2}(\Delta\pav\cdot\lab)-2i\z\cdot\lab\,.\label{comL2}
\end{equation}
Using similar techniques as above we find
\begin{align}
 &\|(\p_iL^j)p_ip_j\vph\|\leq d_5\Big(\sum_{i,j}\|p_ip_j\vph\|^2\Big)^{\frac{1}{2}}
 =d_5\|\pv^2\vph\|\,,\\
 &\|(\Delta L^j+\p_j\pav\cdot\lab)p_j\vph\|
 \leq d_6\Big(\sum_j\|p_j\vph\|^2\Big)^{\frac{1}{2}}
 \leq \con\|h_0^{\frac{1}{2}}\vph\|
 \leq\con\|h_0\vph\|\,,
\end{align}
\begin{equation}
 \|(\Delta\pav\cdot\lab)\vph\|\leq \con\|\vph\|\,,\qquad
 \|\z\cdot\lab\vph\|\leq \con\|\zm^2\vph\|\,.
\end{equation}
Again using \eqref{ineq} and \eqref{domains} one obtains \eqref{hh02} for the
$L$-part.

For the $M$-part note that
\begin{equation}
 \| M\vph\|\leq\con\|\zm^2\vph\|\leq\con\|(\pv^2+\z^2)\vph\|\,,
\end{equation}
which proves \eqref{hh0}. Consider the commutator
\begin{equation}
 [ M,h_0]=\tfrac{i}{2}[(\pav M)\cdot\pv+\pv\cdot(\pav M)]\,.
\end{equation}
If the assumptions in (ii) are satisfied, then
\begin{multline}
 2|(M\psi,h_0\vph)-(h_0\psi,M\vph)|\leq
 |((\p_iM)\psi,p_i\vph)|+|(p_i\psi,(\p_iM)\vph)|\\
 \leq\con\|h_0^{\frac{1}{2}}\psi\|\|h_0^{\frac{1}{2}}\vph\|\,,
\end{multline}
which proves \eqref{hh01}. Finally, with the assumption in (iii),
\eqref{hh02} for this part is obtained similarly as \eqref{hh0} for the
$L$-part.
\end{proof}

\section{Special variables}\label{spvar}

Variables $(\tau,\z)$ are defined by:
\begin{equation}
 x^0=\tau\zm\,,\qquad \x=\tm\z\,,
\end{equation}
see \eqref{tauzed}. Further notation:
\begin{equation}
 \tmz=(\tau^2+\z^2+1)^\frac{1}{2}\,,\quad
 \tMz=\tm\zm+\tmz\,.
\end{equation}
At various points of the calculations in this section the following
identities are used
\begin{equation}
 \tau^2|\z|^2=\tm^2\zm^2-\tmz^2= \tMz(\tm\zm-\tmz)\,.
\end{equation}

The derivative transformation between these coordinates and the Min\-kow\-ski
system is given by
\begin{equation}\label{trvar}
 \frac{\p(x^0,\x)}{\p(\tau,\z)}
 =\begin{pmatrix}\zm & \dfrac{\tau\z^\tr}{\zm}\\ \dfrac{\tau\z}{\tm} & \tm \1\end{pmatrix},\quad
 \frac{\p(\tau,\z)}{\p(x^0,\x)}
 =\begin{pmatrix}\dfrac{\tm^2\zm}{\tmz^2} & -\dfrac{\tau\tm\z^\tr}{\tmz^2}\\[3ex]
 -\dfrac{\tau\zm\z}{\tmz^2} & \dfrac{1}{\tm}\bigg[\1+\dfrac{\tau^2}{\tmz^2}\z\z^\tr\bigg]\end{pmatrix}\,.
\end{equation}
Using these transformations one finds the new metric tensor matrices
\begin{gather}
 (g_{\mu\nu})
 =\begin{pmatrix}\dfrac{\tmz^2}{\tm^2} & 0\\
  0 & -\tm^2\1+\dfrac{\tau^2}{\zm^2}\z\z^\tr\end{pmatrix}\,,\quad
 \gz=-\frac{\tm^4\tmz^2}{\zm^2}\,,\\
 (g^{\mu\nu})
 =\begin{pmatrix}\dfrac{\tm^2}{\tmz^2} & 0\\
  0 & -\dfrac{1}{\tm^2}\bigg[\1+\dfrac{\tau^2}{\tmz^2}\z\z^\tr\bigg]\end{pmatrix}\,
\end{gather}
and the new vector components of the gamma matrices
\begin{align}
 \gax^\tau&=\frac{\tm}{\tmz^2}\big(\tm\zm\ga^0-\tau\z\cdot\gab\big)\,,\\
 \gax^i&=\frac{1}{\tm}\ga^i-\frac{\tau}{\tm\tmz^2}\big(\tm\zm\ga^0-\tau\z\cdot\gab\big)z^i\,.
\end{align}
The operator $K$ defined by \eqref{defK} and discussed in Appendix \ref{trK}
takes the form
\begin{gather}
 K=\frac{1}{\sqrt{2\tmz}}
 \Big[\sqrt{\tMz}+\frac{\tau}{\sqrt{\tMz}}\ga^0\z\cdot\gab\Big]=K^\dagger\,,\\
 K^{-1}=\ga^0K\ga^0\,,
\end{gather}
which leads to the following form of the transformed gamma matrices
\eqref{gaK} and $\la$ matrices introduced by \eqref{mula}:
\begin{gather}
 \gax_K^i=\frac{1}{\tm}\ga_\bot^i+\frac{\zm}{\tmz}\ga_z\nz^i\,,\\
 \begin{gathered}
 \la^i
 =\frac{\tmz}{\tm^2}\alt^i+\frac{\zm}{\tm}\al_z \nz^i\\
 =\frac{\tmz}{\tm^2}\alt^i+\frac{1}{(\zm+|\z|)\tm}\al_z \nz^i
 +\frac{1}{\tm}\al_z z^i\,,
 \end{gathered}\label{laspec}
\end{gather}
where
\begin{gather}
 \nz^i=\frac{z^i}{|\z|}\,,\qquad
 \ga_z=\nzb\cdot\gab\,,\qquad \ga_\bot^i=\ga^i-\ga_z\nz^i\,,\\[1ex]
 \al^i=\ga^0\ga^i\,,\qquad \al_z=\nzb\cdot\alb\,,\qquad \al_\bot^i=\al^i-\al_z\nz^i\,.
 \end{gather}
For the estimation of $X$ in \eqref{Xnorm} it is convenient to have
\begin{equation}\label{gla}
 (g^{\tau\tau})^{\frac{1}{2}}\la^i
 =\frac{1}{\tm}\alt^i+\frac{1}{(\zm+|\z|)\tmz}\al_z \nz^i
 +\frac{1}{\tmz}\al_z z^i\,.
\end{equation}
The derivatives of $\lambda$ matrices are
\begin{equation}
 \p_j\la^i=\frac{|\z|}{\tm^2}\bigg[\frac{\nz^j\alt^i}{\tmz}
 +\frac{\tau^2\alt^j\nz^i}{\tMz}
 +\Big(\frac{\tau^2\delt^{ji}}{\tMz}
 +\frac{\tm\nz^j\nz^i}{\zm}\Big)\al_z\bigg]\,,
\end{equation}
\begin{equation}
 \p_i\la^i=\frac{1}{\tm^2}
 \bigg[\frac{2\tau^2}{\tMz}+\frac{\tm}{\zm}\bigg]\z\cdot\alb\,,
\end{equation}
where
\begin{equation}
 \delta^{ij}_\bot=\delta^{ij}-\nz^i\nz^j\,,
\end{equation}
and their estimates are easily obtained:
\begin{equation}
|\la^i|\leq\con\frac{\zm}{\tm}\,,\quad
 |\p_j\la^i|\leq\con\frac{|\z|}{\tm\zm}\,,\quad
 |\p_j\p_i\la^k|\leq\frac{\con}{\tm\zm}\,.
\end{equation}
The $\mu$ matrix introduced in \eqref{mula}, and its derivative, take the
form
\begin{equation}
 \mu=\frac{\tmz}{\tm}\beta\,,\quad \p_i\mu=\frac{z^i}{\tm\tmz}\beta\,.
\end{equation}
The commutator, and the anti-commutator of matrices $\la^i$, the latter
defining the form $\rho^{ij}$ introduced in \eqref{defro}, yield
\begin{gather}
 [\la^i,\la^j]=\frac{\tmz^2}{\tm^4}[\al_\bot^i,\al_\bot^j]
 +2\frac{\tmz\zm}{\tm^3}\al_z(\nz^i\al_\bot^j-\nz^j\al_\bot^i)\,,\\
\begin{gathered}\label{rhoex}
 \rho^{ij}=\tfrac{1}{2}[\la^i,\la^j]_+
 =-g_{\tau\tau}g^{ij}
 =\frac{\tmz^2}{\tm^4}\delt^{ij}+\frac{\zm^2}{\tm^2}\nz^i\nz^j\\
 =\frac{1}{\tm^2}\Big[\Big(1+\frac{|\z|^2}{\tm^2}\Big)\delt^{ij}+\zm^2\nz^i\nz^j\Big]\,.
\end{gathered}
\end{gather}
Using the latter form, we find for the $\tau$-derivative of $\rho$
\begin{equation}
 \dot{\rho}^{ij}
 =-\frac{2\tau}{\tm^2}\rho^{ij}-\frac{2\tau|\z|^2}{\tm^6}\delt^{ij}\,,
\end{equation}
from which, in the sense of quadratic forms,
\begin{equation}\label{dtro}
 \sgn(\tau)\dot{\rho}^{ij}
 \leq-\frac{2|\tau|}{\tm^2}\rho^{ij}\,,
\end{equation}
an inequality needed in Lemma~\ref{scatlem}. The expression first appearing
in \eqref{XhatH} and needed in the proof of Theorem \ref{scat} takes the form
\begin{equation}\label{divro}
 \p_i\rho^{ij}=2\frac{\tau^2+\tm^2}{\tm^4}z^j\,.
\end{equation}
The square root of the form $\rho$, introduced in \eqref{defs}, is explicitly
\begin{equation}
 s^{ij}=\frac{\tmz}{\tm^2}\delt^{ij}+\frac{\zm}{\tm}\nz^i\nz^j\,.
\end{equation}
The quantity introduced in \eqref{lapila}, and its properties, are:
\begin{align}
 \La_j&=\frac{|\z|}{2\tmz}\Big(\frac{\tau^2}{\tMz}
 -\frac{\tm\zm}{\tmz^2}\Big)i\al_z\al_\bot^j\,,\label{Laspec}\\[1ex]
 s^{ij}\La_j&=\frac{|\z|}{2\tm^2}\Big(\frac{\tau^2}{\tMz}
 -\frac{\tm\zm}{\tmz^2}\Big)i\al_z\al_\bot^i\,,
\end{align}
\begin{equation}\label{LaLa}
 |\La_j|\leq \frac{\tm}{\tmz}\,,\quad |s^{ij}\La_j|\leq \frac{1}{\tm}\,,\quad
 |\dot{\La}_j|\leq\frac{\con}{\tmz}\,,\quad
 |s^{ij}\dot{\La}_j|\leq\frac{\con}{\tm^2}\,.
\end{equation}
The calculation of $Q$ defined in \eqref{QN} is more tedious, and we write
down two of the intermediate steps:
\begin{gather}
 \tfrac{1}{4}\p_j\big([\la^j,\p\cdot\la]_+\big)
 =\frac{\tau^2}{\tm^3\tMz}\Big(2\zm+\frac{|\z|^2}{\zm}\Big)
 +\frac{\tau^2}{\tm^2\tmz\tMz}+\frac{3}{2\tm^2}\,,\\
 \tfrac{1}{4}(\p\cdot\la)^2
 =\frac{\tau^2\zm}{\tm^3\tMz}-\frac{\tau^2\tmz}{\tm^4\tMz}
 +\frac{\tau^2|\z|^2}{\tm^3\zm\tMz}+\frac{|\z|^2}{4\tm^2\zm^2}\,,
\end{gather}
which substituted in \eqref{QN} gives
\begin{align}
 0<Q&=\tm^{-2}\Big(\frac{\tau^2}{\tm^2}+\frac{\tau^2}{\tmz\tMz}
 +\frac{1}{4\zm^2}+\frac{5}{4}\Big)\leq \frac{\con}{\tm^2}\,,\label{estQ}\\
 0<N&\leq\frac{\con}{\tm^2}\,.\label{estN}
\end{align}
The following estimates are needed in the proof of Lemma \ref{scatlem}
\begin{equation}\label{esdemu}
 s^{ij}\p_j\mu=\frac{z^i\zm}{\tm^2\tmz}\beta\,,\qquad
 |s^{ij}\p_j\mu|\leq\frac{|\z|}{\tm^2}\,.
\end{equation}
The matrix $\nu$ defined by \eqref{defnu}, with the use of \eqref{xib} is
shown to satisfy the estimates
\begin{equation}\label{nusnu}
 |\nu|\leq\frac{\kappa}{\tm}\,,\quad
 |s^{ij}\La_j\nu|\leq \frac{\kappa}{\tm^2}\,,
\end{equation}
which are needed in the proof of Lemma \ref{scatlem}. To simplify notation of
its derivative let us denote $\chi(u)=\xi(u)^{-1}\xi'(u)$. Then
\begin{align}
 \p_i\nu&=\frac{1}{\tm}\Big[\chi(\zm)\al_\bot^i
 +\Big(\chi(\zm)+\frac{|\z|^2}{\zm}\chi'(\zm)\Big)\al_z\nz^i\Big]\,,\\
 s^{ij}\p_j\nu&=\frac{\tmz}{\tm^3}\chi(\zm)\al_\bot^i
 +\frac{1}{\tm^2}\Big(\zm\chi(\zm)+|\z|^2\chi'(\zm)\Big)\al_z\nz^i\,.
\end{align}
It follows from \eqref{xib} that $\chi(u)\leq\kappa/u$ and
$|\chi'(u)|\leq\con/u$, so
\begin{equation}\label{esdenu}
 |s^{ij}\p_j\nu|\leq \con\frac{|\z|}{\tm^2}\,.
\end{equation}

The spreading of the past and the future in Minkowski space, in the language
of the coordinates $(\tau,\z)$, is given by the following lemma. This result
is needed for the description of scattering in Section \ref{scattering}.
\begin{lem}\label{pastfut}
The past (resp.\ future) of the set $\tau=\tau_0$, $|\z|\leq r_0$ consists of
the points $(\tau,\z)$ such that $\tau<\tau_0$ (resp.\ $\tau>\tau_0$) and
$|\z|\leq r$, where
 \begin{equation}\label{rad}
 \begin{aligned}
 r&=r_0[\tmo\tm-\tau_0\tau]+\rrmo|\tau\tmo-\tm\tau_0| \,,\\[1ex]
 \rrm &=\rrmo[\tmo\tm-\tau_0\tau]+r_0|\tau\tmo-\tm\tau_0| \,.
 \end{aligned}
 \end{equation}
It follows that
\begin{equation}\label{rmb}
 \rrm\leq \rrmo\left\{
 \begin{aligned}
        &(\tmo-\tau_0)(\tm+\tau)\,, & &\tau\geq\tau_0\,,\\
        &(\tmo+\tau_0)(\tm-\tau)\,, & &\tau\leq\tau_0\,.
 \end{aligned}\right.
\end{equation}
In particular, for $\tau_0=0$ and $\tau\in\mR$ we have
\begin{equation}
 \rrm\leq 2\rrmo\tm\,.
\end{equation}
\end{lem}
\begin{proof}
Simple geometrical analysis shows that the radius $r$ is given by the
 conditions that the vector $(\tau\rrm-\tau_0\rrmo,0,0,\tm r-\tmo r_0)$ is
 lightlike and $\tm r-\tmo r_0>0$, that is
\begin{equation}
 \ep(\tau\rrm-\tau_0\rrmo)=\tm r-\tmo r_0>0\,,
\end{equation}
where $\ep=-1$ (resp.\ $\ep=\!+1$) for past (resp.\ future). We set
$\tau_0=\sinh(T_0)$, \mbox{$r_0=\sinh(R_0)$}, $\tau=\sinh(T)$,
\mbox{$r=\sinh(R)$}, where $R_0>0$, $R>0$. Then the equality above is easily
transformed to $\sinh(R-\ep T)=\sinh(R_0-\ep T_0)$, so $R-R_0=\ep(T-T_0)$.
Treating $R_0$, $T$ and $(T-T_0)$ as independent we put
\mbox{$R=R_0+\ep(T-T_0)$} and $T_0=T-(T-T_0)$ in the inequality, which is
then transformed to $\ep\sinh(T-T_0)\cosh(R_0+\ep T)>0$. Thus
$\ep=\sgn(T-T_0)$ and $R=R_0+|T-T_0|$. Substituting this in $r=\sinh(R)$ and
$\rrm=\cosh(R)$ one obtains \eqref{rad}. Moreover,
\begin{equation}
 \rrm\leq\cosh(R_0)\exp|T-T_0|
 =\begin{cases}
  \rrmo\exp(T)\exp(-T_0)\,, & T\geq T_0\,,\\
  \rrmo\exp(T_0)\exp(-T)\,, & T\leq T_0\,,
  \end{cases}
\end{equation}
which gives \eqref{rmb}.
\end{proof}
Finally, we state the following inequalities needed for the estimation of
electromagnetic fields, discussed in Appendix \ref{elmfields}.
\begin{lem}\label{xest}
 For $x=(\tau\zm,\tm\z)$ the following holds:
\begin{align}
 |x^0|+|\x|+1&\geq\tm\zm\,,\label{xplus}\\
 \big||x^0|-|\x|\big|+1&\geq \frac{\tm^2+\zm^2}{2\tm\zm}\,.\label{xminus}
\end{align}
\end{lem}
\begin{proof}
The first inequality follows from
\begin{equation}
 (|\tau|\zm+\tm|\z|+1)^2\geq |\tau|^2\zm^2+\tm^2|\z|^2+1\geq \tm^2\zm^2\,,
\end{equation}
and the second from
\begin{align}
 \big||\tau|\zm-\tm|\z|\big|+1&=\frac{|\tau^2\zm^2-\tm^2|\z|^2|}{|\tau|\zm+\tm|\z|}+1
 \geq\frac{|\tm^2-\zm^2|+2\tm\zm}{2\tm\zm}\\
 &\geq\frac{[(\tm^2-\zm^2)^2+4\tm^2\zm^2]^{\frac{1}{2}}}{2\tm\zm}
 =\frac{\tm^2+\zm^2}{2\tm\zm}\,.
\end{align}
\end{proof}

\section{Free Dirac equation: the Fourier representation}\label{fDe}

Here, mainly to fix notation, we write down the integral Fourier form of the
solution of the free Dirac equation for a given initial condition.

Let $H$ denote the unit hyperboloid $v^2=1$, $v^0>0$, with the
Lorentz-invariant measure $d\mu(v)=d^3v/v^0$. Moreover, introduce the
projection operators in $\mC^4$ defined by $P_\pm(v)=\tfrac{1}{2}(\1\pm
v\cdot\gamma)$. For $f\in \Sc(H,\mC^4)$ the function
\begin{equation}
 \psi(x)=(2\pi)^{-\frac{3}{2}}\int \big[e^{-ix\cdot v}P_+(v)-e^{ix\cdot
 v}P_-(v)\big]f(v) d\mu(v)
\end{equation}
is a smooth solution of the Dirac equation with the initial condition
\begin{equation}
 \psi(0,\x)=(2\pi)^{-\frac{3}{2}}\int \big[e^{i\x\cdot \vv}P_+(v)-e^{-i\x\cdot
 \vv}P_-(v)\big]f(v) d\mu(v)
\end{equation}
in the space $\Sc(\mR^3,\mC^4)$ (where $\x\cdot\vv$ is the Euclidean product
of the space parts of $x$ and $v$). The transformation
\begin{equation}
 \Sc(H,\mC^4)\ni f\mapsto \mathcal{F}f=\psi(0,.)\in \Sc(\mR^3,\mC^4)
\end{equation}
extends to an isometric transformation from the Hilbert space $L^2_\ga(H)$ of
four-spinor functions with the scalar product $(f_1,f_2)=\int
\ov{f_1(v)}\ga\cdot v f_2(v) d\mu(v)$ into $\Hc=\mC^4\otimes L^2(\mR^3)$.
For~$\psi(0,\x)\in \Sc(\mR^3,\mC^4)$ the inverse transformation is given by
\begin{equation}
 f(v)=(2\pi)^{-\frac{3}{2}}\int\big[
 e^{-i\vv\cdot\x}P_+(v)+e^{i\vv\cdot\x}P_-(v)\big]\ga^0\psi(0,\x)d^3x\,,
\end{equation}
which again extends to an isometric transformation on
 $\mC^4\otimes L^2(\mR^3)$.
Therefore,
\begin{equation}
\mathcal{F}: L^2_\ga(H)\mapsto \Hc
\end{equation}
is a unitary transformation, and again
\mbox{$\mathcal{F}^{-1}:\Sc(\mR^3,\mC^4)\mapsto\Sc(H,\mC^4)$}. Thus for
 $\vph\in \Sc(\mR^3,\mC^4)$ the solution of the free Dirac equation with the initial
condition $\psi(0,\x)=\vph(\x)$ is represented by
\begin{equation}\label{Dirini}
 \psi(x)=(2\pi)^{-\frac{3}{2}}\int \big[e^{-ix\cdot v}P_+(v)-e^{ix\cdot
 v}P_-(v)\big][\mathcal{F}^{-1}\vph](v)d\mu(v)\,.
\end{equation}

\section{Solutions of the wave equation}\label{fwe}

In this and the following two sections we shall use the language more
extensively explained in \cite{her17}, where also bibliographic sources are
indicated. Moreover, we shall repeatedly make use of the following estimates,
which we write down here for the convenience of the reader, and whose proof
may be found in \cite{her95}, Appendix B.
\begin{lem}\label{estlem}
 Let $a>0$, $b\geq0$, $c>0$ and $\al>0$. Then
\begin{equation}\label{est}
 \int_0^c(a+bu)^{-\al}du<
 \left\{
 \begin{aligned}
 &\frac{\al}{\al-1}\frac{c}{a^{\al-1}(a+bc)}\,,& &\al>1\,,\\
 &\frac{1}{1-\al}\frac{c}{(a+bc)^\al}\,,& &\al<1\,.
 \end{aligned}\right.
\end{equation}
\end{lem}

A large class of the solutions of the homogeneous wave equation may be
represented in the following form
\begin{equation}\label{freefi}
 \vph(x)=-\frac{1}{2\pi}\int\fd (x\cdot l,l)d^2l\,,
\end{equation}
where $\fd(s,l)=\p_s\f(s,l)$, with $\f(s,l)$ a homogeneous of degree $-1$
function of a real variable $s$ and null (light-like), future-pointing
vectors $l$. Consequently, $\fd(s,l)$ is homogeneous of degree $-2$, and at
this stage this function is all one needs for the above representation, but
later the function $\f$ will play a role. For $\fd$ in the class $C^0$ the
field $\vph$ is in $C^0$ and satisfies the wave equation in the
distributional sense. For $\f$ in the class $C^3$ the field is in the class
$C^2$ and the wave equation is satisfied in the direct sense. The measure
$d^2l$ is the Lorentz invariant measure on the set of null directions,
applicable to functions of $l$, homogeneous of degree $-2$. If $f(l)$ is such
function, and $(e_0,\ldots,e_3)$ is any Minkowski basis, then
\begin{equation}
 \int f(l) d^2l= \int f(k)d\W(\mathbf{k})\,,
\end{equation}
where $k=l/l^0$, $d\W(\mathbf{k})$ is the spherical angle measure in this
coordinate system, and the value of this integral does not depend on the
choice of basis. The operators of intrinsic differentiation in the cone are
given by $L_{ab}=l_a(\p/\p l^b)-l_b(\p/\p l^a)$, and satisfy
\begin{equation}
 \int L_{ab}f(l)d^2l=0.
\end{equation}
Let $Z^a(l)$ be a vector function, such that $l\cdot Z(l)=0$. If one extends
$Z(l)$ to the neighborhood of the cone with the preservation of this
property, then
\begin{equation}
 L_{0a}[Z^a(l)/l^0]=\p\cdot Z(l)\,,
\end{equation}
independently of any particular extension, and where after differentiation
the rhs is restricted to the cone. Therefore, the rhs is well defined and
independent of the choice of basis. Moreover, let in addition $Z(l)$ be
homogeneous of degree $-1$, then
\begin{equation}\label{intdZ}
 \int\p\cdot Z(l)\, d^2l=0\,.
\end{equation}

We shall use the coordinates $(u,\phi)$ defined by
\begin{equation}\label{luf}
 l=l^0k\,,\quad k=e_0+e_3u+\sqrt{1-u^2}R(\phi)\,,\quad
 R(\phi)=e_1\cos\phi+e_2\sin\phi\,.
\end{equation}
We note for later use that
\begin{equation}\label{RR}
 \ddot{R}(\phi)\equiv\p_\phi^2 R(\phi)=-R(\phi)\,.
\end{equation}
 We choose $e_3=\x/|\x|$, and then
\begin{equation}\label{freefiu}
 \vph(x)=-\frac{1}{2\pi}\int \f(x^0-|\x|u,k(u,\phi))du\,d\phi\,.
\end{equation}

The following theorem gives the properties of decay in spacetime of the
solutions \eqref{freefi}. In the estimates \eqref{dW}, \eqref{dWW2},
\eqref{dWW3} below, and other of similar type to appear later, one should
understand that vectors $l$ are scaled to $l^0=1$. The form of such estimates
does not depend on the choice of a basis, only the bounding constants change.
Recall that $\theta$ is the Heaviside step function.
\begin{thm}\label{estWfi}
Let $\vph(x)$ be represented by the formula \eqref{freefi}, with $\fd(s,l)$
and $\f(s,l)$ as described there, and let $\vep\in(0,1)$.\\
 (i) If $\fd(s,l)$ is of class $C^0$, and for some $n\in\{0,1,\ldots\}$
is bounded by
 \begin{equation}\label{dW}
 |\fd(s,l)|\leq\frac{\con}{(|s|+1)^{n+\vep}}\,,
 \end{equation}
then $\vph(x)$ is a $C^0$ field and the following estimates are satisfied
\begin{equation}\label{estfi}
 |\vph(x)|\leq\con
 \begin{cases}
 \dsp\frac{1}{(|x|+1)^\vep}\,,&
 n=0\,,\\[2ex]
 \dsp\frac{1}{|x|+1}\bigg[\frac{\theta(x^2)}{(|x^0|-|\x|+1)^{n-1+\vep}}
 +\theta(-x^2)\bigg]\,,& n\geq1\,.
 \end{cases}
\end{equation}
(ii) If $\f(s,l)$ is of class $C^1$, and for some $n\in\{1,2,\ldots\}$ is
bounded by
\begin{equation}\label{dWW2}
 |L^\al \f(s,l)|\leq\frac{\con}{(|s|+1)^{n-1+\vep}}\,,\quad |\al|\leq 1\,,
\end{equation}
then $\vph(x)$ is a $C^0$ field and the following stronger estimates hold
\begin{equation}\label{estfidi2}
 |\vph(x)|\leq\frac{\con}{(|x|+1)(||x^0|-|\x||+1)^\vep}\,, \qquad
 \vep<\tfrac{1}{2}\,,\quad n=1\,,
\end{equation}
\begin{multline}
 |\vph(x)|\leq\frac{\con}{|x|+1}\bigg[\frac{1}{(||x^0|-|\x||+1)^{n-1+\vep}}
 +\frac{\theta(-x^2)}{(|x|+1)^\frac{1}{2}(|\x|-|x^0|+1)^\frac{1}{2}}\bigg]\,,\\[1ex]
   n\geq2\,.
 \end{multline}
(iii) If $\f(s,l)$ is of class $C^2$, and for some $n\in\{1,2,\ldots\}$ is
bounded by
\begin{equation}\label{dWW3}
 |L^\al \f(s,l)|\leq\frac{\con}{(|s|+1)^{n-1+\vep}}\,,\quad |\al|\leq 2\,,
\end{equation}
then $\vph(x)$ is a $C^1$ field and the following stronger estimates hold
\begin{equation}\label{estfidi3}
 |\vph(x)|\leq
 \begin{cases}
 \dsp \frac{\con}{(|x|+1)(||x^0|-|\x||+1)^\vep}\,,& n=1\,,\\[2ex]
 \dsp\frac{\con}{|x|+1}\bigg[\frac{1}{(||x^0|-|\x||+1)^{n-1+\vep}}
 +\frac{\theta(-x^2)}{|x|+1}\bigg]\,,& n\geq2\,.
 \end{cases}
\end{equation}
\end{thm}
Note that due to the properties of homogeneity of $\fd$ and $\f$, the bound
\eqref{dW} is satisfied in all three cases (with $n\geq2$ in (ii) and
$n\geq1$ in (iii)). Case (ii) for $n=1$ could also be treated, but the result
is more involved, and we do not need it anyway.

\begin{proof}
(i) Using the bound \eqref{dW} and the integral \eqref{freefiu}, after the
change of variables $\si=\sgn(x^0)(x^0-|\x|u)$ we obtain
\begin{multline}
 |\vph(x)|
 \leq\frac{\con}{|\x|}\int\limits_{|x^0|-|\x|}^{|x^0|+|\x|}\frac{d\si}{(|\si|+1)^{n+\vep}}\\
 \leq \frac{\con}{|\x|}
 \begin{cases}
 \dsp\int_0^{2|\x|}\frac{d\xi}{(\xi+|x^0|-|\x|+1)^{n+\vep}}\,,& x^2\geq0\,,\\[2ex]
 \dsp\int_0^{|x^0|+|\x|}\frac{d\si}{(\si+1)^{n+\vep}}\,,&x^2\leq0\,.
 \end{cases}
\end{multline}
Lemma \ref{estlem} gives now the bounds \eqref{estfi}.

 (ii) and (iii) These cases give more restrictive bounds only in the region
$x^2\leq0$, $|\x|\geq1$, in which $|\x|\geq\con(|x|+1)$. We  use the
equalities \eqref{Fdot1} and \eqref{Fdot2}, respectively, of Lemma \ref{lemF}
below. To shorten notation we denote the exponent in \eqref{dWW2} and
\eqref{dWW3} by $q=n-1+\vep$. For $x^2\leq0$ we then have
\begin{equation}
 |\vph(x)|\leq\frac{\con}{|\x|}\bigg[\frac{1}{(|\x|-|x^0|+1)^q}
 +\int\limits_{-1}^1\left\{\begin{matrix}(1-u^2)^{-\frac{1}{2}}\\1\end{matrix}\right\}
 \frac{du}{(||x^0|+|\x|u|+1)^q}\bigg]\,,
\end{equation}
where $(1-u^2)^{-\frac{1}{2}}$ applies in case (ii) and $1$ in case (iii).
Integration in case (iii) is done as in case (i) and the thesis of this case
follows. The integral in case (ii), after the substitution $1-u=t$, is
bounded by $2$ times the integral
\begin{gather}
 \int_0^1 \frac{dt}{f(a,b,t)}\,,\quad
 f(a,b,t)=t^\frac{1}{2}(|bt-a|+1)^q\,,\\
  0\leq a=|\x|-|x^0|\leq|\x|=b\geq\con(|x|+1)\,.
\end{gather}
For $t\in[0,a/(2b)]$ we have $f(a,b,t)\geq\con\, t^\frac{1}{2}(a+1)^q$, so
\begin{equation}
 \int_0^\frac{a}{2b}\frac{dt}{f(a,b,t)}\leq\con\,b^{-\frac{1}{2}}(a+1)^{\frac{1}{2}-q}
\end{equation}
For $t\in[a/(2b),2a/b]$ we have $f(a,b,t)\geq \con\,
a^\frac{1}{2}b^{-\frac{1}{2}}(|bt-a|+1)^q$. The integral is now split into
intervals $[a/(2b),a/b]$ and $[a/b,2a/b]$ and evaluated with the use of Lemma
\ref{estlem}. This gives
\begin{equation}
 \int_\frac{a}{2b}^\frac{2a}{b}\frac{dt}{f(a,b,t)}
 \leq\con\,b^{-\frac{1}{2}}\times
 \begin{cases}
 \dsp (a+1)^{\frac{1}{2}-q}\,, & q<1\,\\
  \dsp (a+1)^{-\frac{1}{2}}\,, & q>1\,.
 \end{cases}
\end{equation}
Finally, if $2a/b\leq 1$, then for $t\in[2a/b,1]$ we have $f(a,b,2a/b+s)\geq
\con\,s^\frac{1}{2}(a+1+bs)^q$, and by changing further the integration
variable to $r=b(a+1)^{-1}s$ one finds
\begin{equation}
 \int_\frac{2a}{b}^1\frac{dt}{f(a,b,t)}
 \leq \con\times
 \begin{cases}
 \dsp b^{-q}\,, & q<\frac{1}{2}\,,\\
 \dsp b^{-\frac{1}{2}}(a+1)^{\frac{1}{2}-q}\,, & q>\frac{1}{2}\,.
 \end{cases}
\end{equation}
Substituting the values of $a$ and $b$ and noting that
$b^{-\frac{1}{2}}(a+1)^{\frac{1}{2}-q}\leq\con\, b^{-q}$ for $q<\frac{1}{2}$,
$b\geq1$, one obtains the thesis.
\end{proof}
\begin{lem}\label{lemF}
Let $\f(s,l)$ be homogeneous of degree $-1$ and choose $e_3$ in \eqref{luf}
as $e_3=\x/|\x|$. Then the following identities hold
\begin{align}
 \frac{1}{2\pi}\int \dot{\f}(x\cdot l,l)&d^2l
 +\frac{1}{|\x|}\big[\f(x^0-|\x|,(1,\x/|\x|))-\f(x^0+|\x|,(1,-\x/|\x|))\big]\\[1ex]
 &=\frac{1}{2\pi|\x|}\int (l^0)^{-1}\big[L'_{03}-\frac{u}{\sqrt{1-u^2}}R^iL'_{0i}\big]\f(x\cdot l,l)d^2l\label{Fdot1}\\[1ex]
 &=\frac{1}{2\pi|\x|}\int (l^0)^{-1}\big[L'_{03}-u\dot{R}^i\dot{R}^jL'_{0i}L'_{0j}\big]\f(x\cdot l,l)d^2l\,,\label{Fdot2}
\end{align}
where prime at $L'_{ab}$ indicates that these operators act only on the
second argument in $\f(x\cdot l,l)$. The equalities \eqref{Fdot1} and
\eqref{Fdot2} hold for $\f$ of class $C^1$ and $C^2$, respectively.
\end{lem}
\begin{proof}
Using the form \eqref{luf} we find
\begin{equation}
 \p_u=\frac{1}{l^0}\frac{\p l^a}{\p u}L_{0a}=L_{03}-\frac{u}{\sqrt{1-u^2}}R^iL_{0i}\,.
\end{equation}
We apply this identity to $\f(s,l)$, then set $s=x\cdot l$ and divide by
$l^0$. As $l^0$ does not depend on $u$, the result may be written as
\begin{equation}
 \p_u\big[(l^0)^{-1}\f(x\cdot l,l)\big]
 +|\x|\fd(x\cdot l,l)
 =(l^0)^{-1}\Big[L'_{03}-\frac{u}{\sqrt{1-u^2}}R^iL'_{0i}\Big]\f(x\cdot l,l)\,,
\end{equation}
(primes at $L'$ as in the thesis). We integrate this identity over $l$'s. The
first term on the lhs gives
\begin{multline}
 \int \p_u \f\big(x^0-|\x|u,k(u,\phi)\big)du\,d\phi\\
 =2\pi\big[ \f\big(x^0-|\x|,(1,\x/|\x|)\big)- \f\big(x^0+|\x|,(1,-\x/|\x|)\big)\big]\,.
\end{multline}
Dividing by $2\pi|\x|$ and rearranging the terms we obtain equality
\eqref{Fdot1}. Next, using \eqref{RR} we transform the second term of the
integrand as
\begin{multline}
 -\frac{u}{l^0\sqrt{1-u^2}}R^iL'_{0i}\f(x\cdot l,l)
 =\p_\phi\Big[\frac{u}{l^0\sqrt{1-u^2}}\dot{R}^iL'_{0i}\f(x\cdot l,l)\Big]\\
 -\frac{u}{l^0\sqrt{1-u^2}}\dot{R}^i\p'_\phi \big[L'_{0i}\f(x\cdot l,l)\big]\,,
\end{multline}
where we have taken into account that $l^0$ and $x\cdot l$ do not depend on
$\phi$. The first term on the rhs vanishes when integrated, while in the
second term we use another differential identity:
\begin{equation}\label{dphi}
 \p_\phi=\frac{1}{l^0}\frac{\p l^a}{\p \phi}L_{0a}=\sqrt{1-u^2}\dot{R}^aL_{0a}\,.
\end{equation}
Substituting the result into the integral we obtain \eqref{Fdot2}.
\end{proof}

\section{Sources}\label{decay}

Let $\rho(x)$ be a continuous field on Minkowski space. We formulate two
decay conditions, with some fixed numbers $\w>0$ and $\vep\in(0,1)$:
\begin{align}
 &\text{(A)}& & |\rho(x)|\leq\frac{\con}{(|x|+1)^\w}\,,\label{cA}\\
 &\text{(B)}&
 & |\rho(x)|\leq\frac{\con}{(|x|+1)^3}\bigg[\theta(x^2)+\frac{1}{(|x|+1)^\vep}\bigg]\,.\label{cB}
\end{align}
We want to consider the existence and decay properties of the advanced and
retarded solutions
\begin{equation}\label{adv}
 \varphi_{\ret/\adv}(x)=4\pi\int D_{\ret/\adv}(x-y)\rho(y)dy
\end{equation}
of the equation
\begin{equation}
 \Box\,\varphi=4\pi\rho\,,
\end{equation}
and of their difference, the radiation field
\begin{equation}
 \vph_\rad(x)=\vph_\ret(x)-\vph_\adv(x)=4\pi\int D(x-y)\rho(y) dy
 =-\frac{1}{2\pi}\int\dV(x\cdot l,l)d^2l\,,
\end{equation}
where (see \cite{her17})
\begin{equation}
 V(s,l)=\int\delta(s-x\cdot l)\rho(x)dx\,.
\end{equation}
We use the standard notation
\begin{gather}
  D_{\ret/\adv}(x)=\frac{1}{2\pi}\theta(\pm x^0)\delta(x^2)\,,\\
  D(x)=D_\ret(x)-D_\adv(x)=\frac{1}{2\pi}\sgn(x^0)\delta(x^2)\,.
\end{gather}

\begin{lem}\label{estlemrv}\ \\
(i) If $\rho(x)$ satisfies condition (A) for $\w>2$, then
$\vph_{\adv/\ret}(x)$ are well defined and continuous, and the following
bound is satisfied:
\begin{equation}
 |\varphi_\adv(x)|\leq\frac{\con}{(|x|+1)^{\w-2}}\qquad \text{for}\quad x^0\geq 0\,;
\end{equation}
similarly for the retarded solution in the half-space $x^0\leq0$. \\
 (ii) If in
addition to (i) $\rho(x)$ is of class $C^1$, and $\n\rho(x)$ and
$x\cdot\n\rho(x)$ satisfy condition (A) for $\w>2$, then
$\vph_{\adv/\ret}(x)$ are of class $C^1$ and
\begin{equation}\label{xdph}
 x\cdot\n\vph_{\adv/\ret}(x)=\int D_{\adv/\ret}(x-y)(y\cdot\n+2)\rho(y)\,dy\,.
\end{equation}
(iii) If $\rho(x)$ satisfies condition (A) for $\w>3$, or condition (B), then
$V(s,l)$ is well defined and continuous and satisfies the bound
\begin{equation}\label{estV}
 |V(s,l)|\leq\frac{\con}{(|s|+1)^{\w-3}}\,,
\end{equation}
where $\w=3$ in case of condition (B).\\
 (iv) If in addition to (iii)
$\rho(x)$ is of class $C^1$, and $\n\rho(x)$ and $x\cdot\n\rho(x)$ satisfy
condition (A) for $\w>3$ or condition (B), then $V(s,l)$ is of class $C^1$
and
\begin{equation}\label{sds}
 s\p_s\int \delta(s-x\cdot l)\rho(x)\,dx
 =\int\delta(s-x\cdot l)(x\cdot\n+3)\rho(x)\,dx\,.
\end{equation}
\end{lem}

\begin{proof}
 (i) The use of the integral formula \eqref{adv} and the bound \eqref{cA}
 gives
\begin{equation}\label{advbo}
 |\vph_\adv(x)|=\bigg|\int\frac{\rho(x^0+|\z|,\x+\z)}{|\z|}d^3z\bigg|
 \leq \con\int_0^\infty I(x^0,|\x|,r)dr\,,
\end{equation}
where we have denoted $r=|\z|$, and
\begin{equation}
 I(x^0,|\x|,r)=\int_0^2\frac{r\,d u}{[1+(r+x^0)^2+(r-|\x|)^2+2|\x|ru]^{\frac{\w}{2}}}
\end{equation}
with $u=\cos[\measuredangle(\z,\x)]+1$. Using the estimate \eqref{est} we
find
\begin{equation}
 I(x^0,|\x|,r)
 \leq\frac{\con\,r}{[1+(r+x^0)^2+(r+|\x|)^2][1+(r+x^0)^2+(r-|\x|)^2]^{\frac{\w}{2}-1}}\,,
\end{equation}
so the integral on the rhs of \eqref{advbo} is well-defined, and the
continuity of $\vph_\adv(x)$ easily follows. Moreover, for $x^0\geq0$ it
follows
\begin{equation}
 I(x^0,|\x|,r)\leq\frac{\con\,r}{[1+x^0+|\x|+r]^2[1+r+x^0+|r-|\x||]^{\w-2}}\,,
\end{equation}
and then
\begin{equation}
 \int_0^{|\x|}I(x^0,|\x|,r)dr\leq\frac{\con}{[1+x^0+|\x|]^\w}\int_0^{|\x|}rdr
 \leq\frac{\con}{(|x|+1)^{\w-2}}\,,
\end{equation}
and (with $\si=r-|\x|$)
\begin{equation}
 \int_{|\x|}^\infty I(x^0,|\x|,r)dr
 \leq\con\int_0^\infty\frac{d\si}{(\si+x^0+|\x|+1)^{\w-1}}\leq\frac{\con}{(|x|+1)^{\w-2}}\,,
\end{equation}
which ends the proof of (i).

 (ii) We integrate the identity
\begin{multline}
 \frac{x\cdot\n_x\rho(x^0+|\z|,\x+\z)}{|\z|}
 =\frac{(y\cdot\n_y+2)\rho(y)_{|y=(x^0+|\z|,\x+\z)}}{|\z|}\\
 -\frac{\p}{\p z^i}\Big(\frac{z^i\rho(x^0+|\z|,\x+\z)}{|\z|}\Big)
\end{multline}
over $|\z|\leq R$ and transform the integral of the second term on the rhs to
a surface integral. This surface contribution is bounded by $RI(x^0,|\x|,R)$,
so it vanishes in the limit $R\to \infty$. The formula \eqref{xdph} follows.

 (iii) We choose the space basis in which $e_3=\mathbf{l}$, introduce new
 variables $\al=x^0-x^3$, $\beta=x^0+x^3$, and denote $\x_\bot=x^1e_1+x^2e_2$.
Conditions \eqref{cA} and \eqref{cB} may be then written as
\begin{align}
 &\text{(A)}& & |\rho(x)|\leq\frac{\con}{(1+\al^2+\beta^2+|\x_\bot|^2)^\frac{\w}{2}}\,,\\
 &\text{(B)}&
 & |\rho(x)|\leq\frac{\con\,\theta(\al\beta-|\x_\bot|^2)}{(1+\al^2+\beta^2+|\x_\bot|^2)^\frac{3}{2}}
 +\frac{\con}{(1+\al^2+\beta^2+|\x_\bot|^2)^\frac{3+\vep}{2}}\,,
\end{align}
and $V$ takes the form
 $V(s,l)=\frac{1}{2}\int\rho(\al\!=\!s,\beta,\x_\bot)d\beta d^2x_\bot$.
For the case (A) with $\w>3$ we then have
\begin{multline}
 |V(s,l)|\leq\con\int \frac{d\beta d^2x_\bot}{(1+s^2+\beta^2+|\x_\bot|^2)^\frac{\w}{2}}\\
 \leq\con\int_0^\infty\frac{r^2dr}{(1+|s|+r)^\w}=\frac{\con}{(1+|s|)^{\w-3}}\,.
\end{multline}
For the case (B) it is now sufficient to restrict integral to the region
$x^2\geq0$, which leads to the estimation of the integral $\int
K(\beta)d\beta$, with
\begin{align}
 K(\beta)&=\int\frac{\theta(s\beta-|\x_\bot|^2)d^2x_\bot}{(1+s^2+\beta^2+|\x_\bot|^2)^\frac{3}{2}}
 =\con\,\theta(s\beta)\int_0^{s\beta}\frac{d\xi}{(1+s^2+\beta^2+\xi)^\frac{3}{2}}\\
 &\leq\con\frac{\theta(s\beta)s\beta}{(1+s^2+\beta^2)^\frac{3}{2}}
 \leq\con\frac{\theta(s\beta)|s|}{(1+|s|+|\beta|)^2}\,.
\end{align}
This leads to
\begin{equation}
 \int K(\beta)d\beta\leq\con\,|s|\int_0^\infty \frac{d\si}{(1+|s|+\si)^2}\leq\con\,.
\end{equation}

(iv) We integrate the identity
\begin{multline}
 s\p_s\rho(s,\beta,\x_\bot)
 =(x\cdot\n+3)\rho(x)_{|\al=s}\\
 -\p_\beta\big(\beta\rho(s,\beta,\x_\bot)\big)
 -\sum_{i=1}^2\p_{x^i}\big(x^i\rho(s,\beta,\x_\bot)\big)
\end{multline}
over the region $|\x_\bot|\leq R$, $|\beta|\leq B$. The contribution of the
surface terms is bounded by
\begin{equation}\label{BR}
 \sum_{\ep=\pm}B\int_{\mR^2}|\rho(s,\pm B,\x_\bot)|d^2x_\bot
 +R^2\int_{\mR\times S^1}|\rho(s,\beta,R\mathbf{n})|d\beta d\w(\mathbf{n})\,,
 \end{equation}
 where $d\w(\mathbf{n})$ is the angle measure on the unit circle $S^1$.
Condition (A) with $\w>3$ is stronger than condition (B), so it is sufficient
to consider the latter, in the form given in the proof of (iii) above. Then
the first contribution in \eqref{BR} is bounded by
\begin{gather}
 \con\bigg[\int_0^{|s|B}\frac{Bd\zeta}{(1+s^2+B^2+\zeta)^\frac{3}{2}}
 +\int_0^\infty\frac{Bd\zeta}{(1+s^2+B^2+\zeta)^\frac{3+\ep}{2}}\bigg]\\
 \leq \frac{\con\, |s|B^2}{(1+|s|+B)^3}+\frac{\con\, B}{(1+|s|+B)^{1+\ep}}\,,
\end{gather}
where we have substituted $|\x_\bot|^2=\zeta$. The second term in \eqref{BR}
is bounded by
\begin{gather}
 \con\bigg[\int_{R^2/|s|}^\infty \frac{R^2 d\beta}{(1+|s|+R+\beta)^3}
 +\int_0^\infty\frac{R^2 d\beta}{(1+|s|+R+\beta)^{3+\ep}}\bigg]\\
 \leq \frac{\con\,R^2}{(1+|s|+R+R^2/|s|)^2}+ \frac{\con\,R^2}{(1+|s|+R)^{2+\ep}}\,.
\end{gather}
All these bounds vanish in the limit $B,R\to\infty$, which leads to the
identity~\eqref{sds}.
\end{proof}

\section{Electromagnetic fields and Lorenz potentials}\label{elmfields}

We apply now the results of the discussion of the last two sections to the
case of electromagnetic fields. We start with the Lorenz potentials of free
fields, which we assume to have been produced as radiation fields by some
conserved currents.\pagebreak[2]

\begin{thm}\label{elmdec}
Let $J(x)$ satisfy Assumption \ref{FAJ} in Section \ref{typspec}. Then the
Lorenz potential $A_\rad(x)$ of the radiated field of this current and the
scalar field $C_\rad(x)=x\cdot A_\rad(x)$ satisfy the following decay
estimates:
\begin{gather}
  \left.\begin{aligned}
        |A_\rad(x)|&\\[1ex]
        |\n C_\rad(x)|&
        \end{aligned}\ \right\}
        \leq\frac{\con}{|x|+1}\Big[\frac{1}{(||x^0|-|\x||+1)^\vep}
 +\theta(-x^2)\Big]\,,\label{AC}\\[1ex]
 |\n A_\rad(x)|\leq\frac{\con}{|x|+1}\Big[\frac{1}{(||x^0|-|\x||+1)^{1+\vep}}
 +\frac{\theta(-x^2)}{|x|+1}\Big]\,,\label{nA}\\[1ex]
 |\n\n C_\rad(x)|\leq\frac{\con}{|x|+1}\bigg[\frac{1}{(||x^0|-|\x||+1)^{1+\vep}}
 +\frac{\theta(-x^2)}{(|x|+1)^\frac{1}{2}(|\x|-|x^0|+1)^\frac{1}{2}}\bigg]\,,\\
 \label{nnC}
 \end{gather}
\begin{align}
 |\n\n A_\rad(x)|&\leq \frac{\con}{|x|+1}\,,\label{2A}\\[1ex]
  \left.\begin{aligned}
  |(x\cdot\n+1) A_\rad(x)|&\\[1ex]
  |\n(x\cdot\n C_\rad)(x)|&
  \end{aligned}\ \right\}
  &\leq\frac{\con}{(|x|+1)(||x^0|-|\x||+1)^\vep}\,,\label{xdAC}\\[1ex]
 |(x\cdot\n)^k C_\rad(x)|
 &\leq\frac{\con}{(|x|+1)^\vep}\,,\quad k=1,2\,.\label{xdC}
\end{align}
\end{thm}
\noindent
 The estimate \eqref{2A} could be refined, but this is not needed
for our purposes.
\begin{proof}
It follows from the assumption \eqref{estJ} and Lemma \ref{estlemrv} (iii)
that the function
\begin{equation}
 V(s,l)=\int\delta(s-x\cdot l)J(x)dx=\int J\Big(\frac{s+\x\cdot\mathbf{l}}{l^0},\x\Big)\frac{d^3x}{l^0}
\end{equation}
is of class $C^3$. The use of estimate \eqref{esthomJ}, relation \eqref{sds}
and the estimates of the form \eqref{estV} shows that
\begin{equation}
 |\p_s^k L^\al \dV(s,l)|\leq\frac{\con}{(|s|+1)^{k+1+\vep}}\,,\quad
   k+|\al|\leq2\,.
\end{equation}
From the representation
\begin{equation}\label{pA}
 A_\rad(x)=-\frac{1}{2\pi}\int \dV(x\cdot l,l)d^2l\,,
\end{equation}
by the use of the calculation
\begin{equation}
 \p\cdot V(x\cdot l,l)=x\cdot \dV(x\cdot l,l)+\p'\cdot V(x\cdot l,l)
\end{equation}
and the integral identity \eqref{intdZ}, we also find
\begin{align}
 C_\rad(x)&=\frac{1}{2\pi}\int\p'\cdot V(x\cdot l,l)d^2l\,,\\
 \n_a C_\rad(x)&=\frac{1}{2\pi}\int l_a\p'\cdot \dV(x\cdot l,l)d^2l\,.
\end{align}
(Similarly as $L'_{ab}$ in Lemma \ref{lemF} in Appendix \ref{fwe}, $\p'$ acts
only on the second argument.) All estimates now follow from the appropriate
statements of Theorem~\ref{estWfi}: estimates \eqref{AC} from (i) with $n=1$,
estimate \eqref{nA} from (iii) for $n=2$, estimate \eqref{nnC} from (ii) for
$n=2$, estimate \eqref{2A} from (i) with $n=1$. Next, with
$W(s,l)=s\dV(s,l)$, we have
\begin{align}
 (x\cdot \n+1)A_\rad(x)&=-\frac{1}{2\pi}\int \dot{W}(x\cdot l,l)d^2l\,,\label{pxd1A}\\
 \n_a(x\cdot\n C_\rad(x))&=\frac{1}{2\pi}\int l_a\p'\cdot\dW(x\cdot l,l)d^2l\,,\label{pdxdC}
\end{align}
so the estimates \eqref{xdAC} follow from Theorem \ref{estWfi} (ii) with
$n=1$.  Finally, with $U(s,l)=s\dW(s,l)$, we have
\begin{align}
 x\cdot\n C_\rad(x)&=\frac{1}{2\pi}\int \p'\cdot W(x\cdot l,l)d^2l\,,\label{pxdC}\\
 (x\cdot\n)^2C_\rad(x)&=\frac{1}{2\pi}\int\p'\cdot U(x\cdot l,l)d^2l\,, \label{pxdxdC}
\end{align}
so the estimates \eqref{xdC} follow from Theorem \ref{estWfi} (i) for $n=0$.
\end{proof}

Free in/out fields of the type discussed above are now augmented by ret/adv
fields to produce the total electromagnetic field.

\begin{thm}\label{totelm}
Let the electromagnetic potential $A$ satisfy Assumption \ref{FAJ} in Section
\ref{typspec}. Then the following estimates are satisfied:
\begin{equation}
 \left.\begin{aligned}
        |\n^\al A(x)|&\\[1ex]
        |\n C(x)|&
        \end{aligned}\ \right\}
        \leq\frac{\con}{|x|+1}\leq\frac{\con}{\tm\zm}\,,\quad
        |\al|=0,2\,,\label{Ab}
\end{equation}
\begin{multline}
        |\n A(x)|
        \leq\frac{\con}{|x|+1}\Big[\frac{1}{(||x^0|-|\x||+1)^{1+\vep}}
 +\frac{1}{|x|+1}\Big]\\
 \leq\frac{\con}{\tm\zm}\bigg[\frac{1}{\tm\zm}
 +\Big(\frac{\tm\zm}{\tm^2+\zm^2}\Big)^{1+\vep}\bigg]\,,\label{dAb}
\end{multline}
\begin{multline}
        |\n\n C(x)|
        \leq\frac{\con}{|x|+1}\Big[\frac{1}{(||x^0|-|\x||+1)^{1+\vep}}
 +\frac{1}{(|x|+1)^\frac{1}{2}(||x^0|-|\x||+1)^\frac{1}{2}}\Big]\\[1ex]
 \leq\frac{\con}{\tm\zm}\bigg[\frac{1}{\tm^\frac{1}{2}\zm^\frac{1}{2}}
 \Big(\frac{\tm\zm}{\tm^2+\zm^2}\Big)^\frac{1}{2}
 +\Big(\frac{\tm\zm}{\tm^2+\zm^2}\Big)^{1+\vep}\bigg]\,,\label{dCb}
\end{multline}
\begin{multline}
  \left.\begin{aligned}
  |(x\cdot\n+1) A(x)|&\\[1ex]
  |\n(x\cdot\n C(x))|&
  \end{aligned}\ \right\}
  \leq\frac{\con}{(|x|+1)(||x^0|-|\x||+1)^\vep}\\
  \leq\frac{\con}{\tm\zm}\Big(\frac{\tm\zm}{\tm^2+\zm^2}\Big)^\vep\,,\label{xdAb}
\end{multline}
\begin{gather}
 |(x\cdot\n)^k C(x)|\leq\frac{\con}{(|x|+1)^\vep}
 \leq\frac{\con}{\tm^\vep\zm^\vep}\,,\quad k=1,2\,,\\[1ex]
 |x^aF_{ab}(x)|\leq\frac{\con}{|x|+1}\leq\frac{\con}{\tm\zm}\,.
\end{gather}
\end{thm}
\begin{proof}
To prove the estimates on $A$ and $C$ in the half-space $x^0\geq0$ it is
sufficient to add the estimates of $A_\out$ and $C_\out$, as given by Theorem
\ref{elmdec}, to the estimates of $A_\adv$ and $C_\adv$ following from the
assumptions on $J$ and Lemma~\ref{estlemrv} (i) and~(ii). For $x^0\geq0$ one
has
\begin{gather}
 |\n^\al A_\adv(x)|\leq\frac{\con}{(|x|+1)^{|\al|+1}}\,,\quad |\al|\leq3\,,\\
 \begin{aligned}
 |\n^\al(x\cdot\n+1)A_\adv(x)|
 &=\bigg|\int D_\adv(x-y)(y\cdot\n+3+|\al|)\n^\al J(y)dy\,\bigg|\\
 &\leq\frac{\con}{(|x|+1)^{1+|\al|+\vep}}\,,\quad |\al|\leq1\,.
 \end{aligned}
\end{gather}
Moreover,
\begin{gather}
 x\cdot \n C_\adv(x)=x^c(x\cdot\n+1)A_{\adv\,c}(x)\,,\\[1ex]
 (x\cdot\n)^2 C_\adv(x)=x\cdot\n C_\adv(x)
 +x^cx^d\n_d(x\cdot\n+1)A_{\adv\,c}(x)\,,\\[1ex]
 x^aF_{ab}(x)=(x\cdot\n+1)A_b(x)-\n_bC(x)\,.
\end{gather}
and all the estimates in the $x$-form now easily follow. Similarly for
$x^0\leq0$ with the use of $A_\inc$ and $A_\ret$. The $(\tau,\z)$-form of the
estimates is the consequence of Lemma \ref{xest}.
\end{proof}

\section{Special gauge in special variables}\label{spgava}

We consider now our field in the special gauge \eqref{specgauge}, which we
write in the form
\begin{equation}
 \cA=A-\n \gf\,,\quad \gf=\log(\tm\zm)C\,,\quad C(x)=x\cdot A\,.
\end{equation}
We need the components of this potential, and various expressions involving
them, in our special coordinates system. For the sake of this section we
introduce the abbreviation
\begin{equation}
 \ltz=\log(\tm\zm)\,.
\end{equation}
\begin{thm}\label{elmbo}
The special gauge potential $\hA_\mu(\tau,\z)$ is a $C^2$ field, and the
following estimates are satisfied:
\begin{gather}
 |\hF_{i\tau}|\leq\frac{\con}{\tm\zm}\,,\quad |z^i\hF_{i\tau}|\leq\frac{\con}{\tm}\,,\quad
 |\hF_{ij}|\leq\con\,,\quad |z^i\hF_{ij}|\leq\frac{\con}{\zm}\,,\\[1ex]
 |\hA_\tau|\leq\con\bigg[\frac{1+\log\tm}{\tm^{1+\vep}}+\frac{\log\zm}{\tm^3}\bigg]\,,\\[1.5ex]
 |\hA_i|\leq\con\frac{1+\log(\tm\zm)}{\zm}\,,\quad
 |z^i\hA_i|\leq\con\frac{1+\log(\tm\zm)}{\zm^\vep}\,,\\[1ex]
 |\p_i\hA_\tau|+|z^i\p_i\hA_\tau|\leq\con\frac{1+\log\tm}{\tm^{1+\vep}}\,,
\end{gather}
\begin{gather}
 |\p_\tau\hA_i|+|z^i\p_\tau\hA_i|\leq\frac{\con}{\tm}\,,\\
 |\p_i\hA_j|\leq\con(1+\log\tm)\,,\\
 |z^i\p_i\hA_j|+|z^i\p_j\hA_i|\leq\con\frac{1+\log(\tm\zm)}{\zm}\,,\\
 |z^iz^j\p_i\hA_j|\leq\con\frac{1+\log(\tm\zm)}{\zm^\vep}\,.
\end{gather}
\end{thm}
The proof below allows formulation of more restrictive, but at the same time
more sophisticated and less transparent estimates. These given above are
sufficient for our purposes.
\begin{proof}
According to Assumption \ref{FAJ} the potential $A$ is of class $C^3$, so
$\cA$ is of class $C^2$. For the calculation of the components and
derivatives of $A$ and $C$ we use the following form of the coordinates
change:
\begin{equation}\label{cochange}
 \p_\tau x=\frac{\tau}{\tm^2}x+\frac{\zm}{\tm^2}(1,0)\,,\qquad
 \p_i x=\Big(\tau\frac{z^i}{\zm},\tm\delta^{i,.}\Big)\,,
\end{equation}
which leads further to additional identities
\begin{gather}
 z^i\p_ix=x-\frac{\tau}{\zm}(1,0)\,,\\
 \p_i\p_\tau x=\frac{\tau}{\tm^2}\p_i x+\frac{z^i}{\tm^2\zm}(1,0)\,,
 \qquad
 \p_i\p_j x=\frac{\tau}{\zm}d^{ij}(1,0)\,,
\end{gather}
where $d^{ij}$ is given by
\begin{equation}\label{dd}
 d^{ij}=\delta^{ij}-\frac{z^iz^j}{\zm^2}\,,
\end{equation}
and some of its properties are
\begin{gather}
 \p_i\Big(\frac{z^j}{\zm}\Big)=\frac{d^{ij}}{\zm}\,,\qquad
 z^i d^{ij}=\frac{z^j}{\zm^2}\,,\\[1ex]
 \p_kd^{ij}=\frac{-1}{\zm^2}(z^id^{jk}+z^jd^{ik})\,,\qquad
 z^i\p_i d^{jk}=-2\frac{z^jz^k}{\zm^4}\,.
\end{gather}

 A straightforward calculation now gives
\begin{align}
 \hF_{i\tau}&=(\p_i x^a)(\p_\tau x^b)F_{ab}=(\p_ix^a)\Big(\frac{\tau}{\tm^2}F_{ab}x^b
 +\frac{\zm}{\tm^2}F_{a0}\Big)\,,\\
 z^i\hF_{i\tau}&=\Big(\frac{\zm}{\tm^2}+\frac{\tau^2}{\tm^2\zm}\Big)x^aF_{a0}\,,\\
 \hF_{ij}&=(\p_ix^a)(\p_jx^b)F_{ab}\,,\\
 z^i\hF_{ij}&=\Big(x^aF_{ab}-\frac{\tau}{\zm}F_{0b}\Big)(\p_jx^b)\,.
\end{align}
\begin{align}
 \hA_\tau&=\frac{\zm}{\tm^2}A_0-\ltz \p_\tau C\,,\\[1ex]
 \hA_i&=\tm\, d^{ij}A_j-\ltz \p_iC\,,\\[1ex]
 z^i\hA_i&=\frac{\tm}{\zm^2}z^iA_i-\ltz z^i\p_iC
\end{align}
\begin{equation}
 \p_i\hA_\tau=
 \frac{z^i}{\tm^2\zm}A_0+\frac{\zm}{\tm^2}\p_i A_0-\frac{z^i}{\zm^2}\p_\tau C
 -\ltz \p_i\p_\tau C\,,
\end{equation}
\begin{equation}
 z^i\p_i\hA_\tau=\frac{\z^2}{\tm^2\zm}A_0+\frac{\zm}{\tm^2}z^i\p_i A_0
 -\frac{|\z|^2}{\zm^2}\p_\tau C-\ltz z^i\p_i\p_\tau C\,,
\end{equation}
\begin{equation}
 \p_i\hA_j=\tm\big[(\p_id^{jk})A_k+d^{jk}\p_i A_k\big]
 -\frac{z^i}{\zm^2}\p_j C -\ltz \p_i\p_j C\,,
\end{equation}
\begin{equation}
 z^i\p_i\hA_j=\tm\Big[-2\frac{z^jz^k}{\zm^4}A_k
 +d^{jk}z^i\p_i A_k\Big] -\frac{|\z|^2}{\zm^2}\p_jC
 -\ltz z^i\p_i\p_jC\,,
\end{equation}
\begin{equation}
 z^iz^j\p_i\hA_j=\frac{\tm}{\zm^2}\Big[-2\frac{|\z|^2}{\zm^2}z^kA_k
 +z^iz^k\p_i A_k\Big]
 -\frac{|\z|^2}{\zm^2}z^j\p_jC
 -\ltz z^iz^j\p_i\p_jC\,,
\end{equation}

For the estimation of all these field and potential expressions we use
Theorem \ref{totelm}. We note that
\begin{equation}
 |\p_i x|\leq\con\tm\,,\qquad |\p_\tau x|\leq\con\zm\,,\qquad
  |\p_z^\al d^{ij}|\leq\frac{\con}{\zm^{|\al|}}\,,\ |\al|\leq 2\,.
\end{equation}
Moreover, when using the estimates \eqref{dAb}, \eqref{dCb} and \eqref{xdAb}
we shall make use of the obvious inequalities
\begin{equation}\label{tauzedb}
 \frac{\tm\zm}{\tm^2+\zm^2}
 \leq\min\bigg\{\frac{1}{2},\frac{\tm}{\zm},\frac{\zm}{\tm}\bigg\}
 \leq\left\{\begin{aligned}&\tm\zm^{-1}\,,\\
                            &\tm^{-1}\zm\,,\\
                            &2^{-1}\end{aligned}\right.
\end{equation}
(the choice of one of them will be dictated by a particular need). In this
way we find
\begin{align}
 |\hF_{i\tau}|&\leq\con\Big(|F_{ab}x^b|+\frac{\zm}{\tm}|F_{a0}|\Big)
 \leq\frac{\con}{\tm\zm}\,,\\
 |z^i\hF_{i\tau}|&\leq\con\Big(\frac{\zm}{\tm^2}+\frac{1}{\zm}\Big)|x^aF_{a0}|
 \leq\frac{\con}{\tm}\bigg(\frac{1}{\tm^2}+\frac{1}{\zm^2}\bigg)\,,\\
 |\hF_{ij}|&\leq\con\tm^2|F_{ab}|\leq\con\,,\\
 |z^i\hF_{ij}|&\leq\con\Big(\tm |x^aF_{ab}|+\frac{\tm^2}{\zm}|F_{0b}|\Big)
 \leq\frac{\con}{\zm}\,,
\end{align}
which exhausts the claims on the field $F$.

The estimation of the potential expressions is more tedious, and we consider
it in steps. In the first step we estimate the $\tau$- and $\z$-dependent
coefficients, which gives
\begin{gather}
 |\hA_\tau|\leq\frac{\zm}{\tm^2}|A_0|+\ltz |\p_\tau C|\,,\\
 |\hA_i|\leq \con\tm|A_j|+\ltz |\p_iC|\,,\quad
 |z^i\hA_i|\leq \con\frac{\tm}{\zm}|A_i|+\ltz |z^i\p_iC|
\end{gather}
\begin{equation}
 |\p_i\hA_\tau|\leq \frac{1}{\tm^2}\big(|A_0|+\zm|\p_i A_0|\big)+\frac{1}{\zm}|\p_\tau C|
 +\ltz |\p_i\p_\tau C|\,,
\end{equation}
\begin{equation}
 |z^i\p_i\hA_\tau|\leq \frac{\zm}{\tm^2}\big(|A_0|+|z^i\p_i A_0|\big)
 +|\p_\tau C|+\ltz |z^i\p_i\p_\tau C|\,,
\end{equation}
\begin{equation}
 |\p_i\hA_j|\leq \con\tm\Big(\frac{1}{\zm}|A_k|+|\p_i A_k|\Big)
 +\frac{1}{\zm}|\p_j C|+\ltz |\p_i\p_j C|\,,
\end{equation}
\begin{equation}
 |z^i\p_i\hA_j|\leq \con\tm\Big(\frac{1}{\zm^2}|A_k|+|z^i\p_i A_k|\Big)
 +|\p_jC|+\ltz |z^i\p_i\p_jC|\,,
\end{equation}
\begin{equation}
 |z^iz^j\p_i\hA_j|\leq \con\frac{\tm}{\zm}(|A_k|+|z^i\p_i A_k|)
 +|z^j\p_jC|+\ltz |z^iz^j\p_i\p_jC|\,,
\end{equation}
Next, we evaluate the $A$- and $C$-terms:
\begin{gather}
 \p_i A_b=(\p_ix^a)\n_aA_b\,,\qquad
 z^i\p_i A_b=x\cdot\n A_b-\frac{\tau}{\zm}\n_0 A_b\,,\\
 \p_i\p_j A_c=(\p_i x^a)(\p_j x^b)\n_a\n_b A_c+\frac{\tau}{\zm}d^{ij}\n_0 A_c
\end{gather}
(the latter term is not needed for the present proof, but its estimate given
below is used in the proof of Theorem \ref{totelmga} (iii) (b)),
\begin{gather}
 \p_\tau C=\frac{\tau}{\tm^2}x\cdot\n C+\frac{\zm}{\tm^2}\n_0 C\,,\\
 \p_iC= (\p_ix^a)\n_a C\,,\qquad
 z^i\p_iC= x\cdot\n C-\frac{\tau}{\zm}\n_0C\,,\\
 \p_i\p_\tau C=(\p_i x^a)\Big[\frac{\tau}{\tm^2}\n_a(x\cdot\n C)
 +\frac{\zm}{\tm^2}\n_a\n_0 C\Big]+\frac{z^i}{\tm^2\zm}\n_0 C\,,
\end{gather}
\begin{multline}
 z^i\p_i\p_\tau C=\frac{\tau}{\tm^2}\Big[(x\cdot\n)^2C-\n_0^2 C\Big]
 +\Big[\frac{\zm}{\tm^2}-\frac{\tau^2}{\tm^2\zm}\Big]\n_0(x\cdot\n C)\\
 -\frac{1}{\tm^2\zm}\n_0 C\,,
\end{multline}
\begin{gather}
 \p_i\p_jC=(\p_ix^a)(\p_jx^b)\n_a\n_b C+\frac{\tau}{\zm}d^{ij}\n_0C\,,\\
 z^i\p_i\p_jC=(\p_jx^a)\Big[\n_a(x\cdot\n-1)C-\frac{\tau}{\zm}\n_a\n_0C\Big]
 +\frac{\tau z^j}{\zm^3}\n_0C\,,
\end{gather}
\begin{multline}
 z^iz^j\p_i\p_jC=(x\cdot\n-1)x\cdot\n C-2\frac{\tau}{\zm}\n_0(x\cdot\n C)
 +\frac{\tau^2}{\zm^2}\n_0^2C\\
 \hspace{15em}+\frac{\tau}{\zm}\Big(2+\frac{|\z|^2}{\zm^2}\Big)\n_0C\,.
\end{multline}

Now we can estimate the $A$- and $C$-terms:
\begin{gather}
 |A_a|\leq\frac{\con}{\tm\zm}\,,\qquad
 |\p_i A_b|\leq\con\,\min\bigg\{\frac{1}{\zm}, \frac{1}{\tm}\bigg\}\,,\\[1ex]
 |z^i\p_iA_b|\leq\frac{\con}{\tm\zm}\,,\qquad
 |\p_i\p_j A_b|\leq\con\frac{\tm}{\zm}\,,
\end{gather}
\begin{gather}
 |\p_\tau C|\leq\con\Big(\frac{1}{\tm^{1+\vep}\zm^\vep}
 +\frac{1}{\tm^3}\Big)\,,\\[1ex]
 |\p_iC|\leq \frac{\con}{\zm}\,,\qquad
 |z^i\p_iC|\leq \frac{\con}{\zm^\vep}\,,\\[1ex]
 |\p_i\p_\tau C|\leq\frac{\con}{\tm^{1+\vep}\zm^{1-\vep}}\,,\qquad
 |z^i\p_i\p_\tau C|\leq \frac{\con}{\tm^{1+\vep}\zm^\vep}\,,
\end{gather}
\begin{gather}
 |\p_i\p_j C|\leq\con\,\min\bigg\{1, \frac{\tm}{\zm}\bigg\}\,,\\[1ex]
 |z^i\p_i\p_j C|\leq\frac{\con}{\zm}\,,\qquad
 |z^iz^j \p_i\p_j C|\leq\frac{\con}{\zm^\vep}\,.
\end{gather}
Substituting these estimates in the bounds on the potential expressions
obtained earlier in the proof, one arrives at the estimates given for them in
the thesis. The only less obvious estimation on the way is that of the term
\begin{equation}
 \log(\tm\zm)\min\bigg\{1, \frac{\tm}{\zm}\bigg\}
\end{equation}
appearing in $|\p_i\hA_j|$. For $\zm\leq\tm$ this is bounded by $2\log\tm$,
while for $\zm\geq\tm$ we write $\log(\tm\zm)=2\log\tm+\log(\zm/\tm)$ and
estimate the term by
\begin{equation}
 2\log\tm+\Big(\frac{\zm}{\tm}\Big)^{-1}\log\Big(\frac{\zm}{\tm}\Big)
 \leq\con(1+\log\tm)\,.
\end{equation}
Finally, we write
\begin{equation}
 |\p_\tau\hA_i|\leq|\p_i\hA_\tau|+|\hF_{i\tau}|\,,\qquad
 |z^i\p_\tau\hA_i|\leq|z^i\p_i\hA_\tau|+|z^i\hF_{i\tau}|
 \end{equation}
 \begin{equation}
 |z^j\p_i\hA_j|\leq |z^j\p_j\hA_i|+|z^j\hF_{ij}|
 \end{equation}
to obtain the last missing estimates.
\end{proof}

\end{document}